\documentclass{article}
\usepackage[utf8]{inputenc}

\title{Free phases of Majorana fermions: 
\\ Tenfold ways compared}
\author{Luuk Stehouwer}
\date{\today}

\usepackage{comment}
\usepackage{amsmath}
\usepackage{amsfonts}
\usepackage{amssymb}
\usepackage{amsthm}
\usepackage{enumerate}
\usepackage{color}
\usepackage{hyperref}
\usepackage{graphicx}
\usepackage{stmaryrd}

\usepackage{tikz-cd}

\newcommand{\Z}{\mathbb{Z}}

\newcommand{\C}{\mathbb{C}}
\newcommand{\R}{\mathbb{R}}

\newcommand{\Spec}{\operatorname{Spec}}

\newcommand{\Aut}{\operatorname{Aut}}
\newcommand{\Hom}{\operatorname{Hom}}

\newcommand{\End}{\operatorname{End}}
\newcommand{\pt}{\operatorname{pt}}
\newcommand{\Mod}{\operatorname{Mod}}
\newcommand{\Spin}{\operatorname{Spin}}
\newcommand{\Pin}{\operatorname{Pin}}
\newcommand{\id}{\operatorname{id}}
\newcommand{\op}{op}
\DeclareMathOperator{\colim}{\operatorname{colim}}

\newcommand{\SPT}{\operatorname{SPT}}
\newcommand{\Pol}{\operatorname{Pol}}
\newcommand{\Grad}{\operatorname{Grad}}

\newcommand{\Diff}{\operatorname{Diff}}

\newcommand{\GrMod}{\mathbf{GrMod}}

\newcommand{\ol}{\overline}

\newtheorem{theorem}{Theorem}[section]
\newtheorem{proposition}[theorem]{Proposition}
\newtheorem{lemma}[theorem]{Lemma}

\newtheorem{corollary}[theorem]{Corollary}
\newtheorem{conclusion}[theorem]{Conclusion}

\theoremstyle{definition}

\newtheorem{definition}[theorem]{Definition}

\theoremstyle{remark}

\newtheorem{remark}[theorem]{Remark}
\newtheorem{example}[theorem]{Example}
\newtheorem{warning}[theorem]{Warning}

\begin{document}

\maketitle

\begin{abstract}
We provide a mathematically rigorous classification of symmetry-protected topological (SPT) phases of neutral free fermions. \sloppy
Our approach utilizes Karoubi triples with negative squares, thought of as polarizations.
We prove that neutral free fermion SPT phases protected by a symmetry algebra $A$ are classified by the real $K$-theory group $K_2(A^{\op})$, and demonstrate how our classification reproduces known results in the presence of charge.
Our formalism also allows for symmetries described by groups, potentially with time-reversal, using the formalism of fermionic groups and their fermionic group $C^*$-algebras.
Our classification extends to positive spatial dimensions and includes weak phases using the crystalline equivalence principle.
Our approach clarifies and unifies various existing tenfold way classifications by establishing their equivalence through Morita equivalences of symmetry algebras.
We expect our classification to be the natural domain for the free-to-interacting map proposed by Freed and Hopkins.
\end{abstract}

\tableofcontents
\newpage


\section{Introduction}

Topological phases of matter have attracted considerable attention due to their novel physical properties and potential applications, particularly in topological quantum computation using Majorana zero modes~\cite{nayak2008non}. 
From a mathematical perspective, free fermionic Symmetry-protected topological (SPT) phases are mostly well-understood, classified rigorously by $K$-theory through a framework developed by Freed--Moore~\cite{freedmoore}, Thiang~\cite{thiangk-theoretic}, Gomi~\cite{gomi2017freed} and others~\cite{kubota2016notes, kellendonk2017c}, following work of Kitaev~\cite{kitaev2009periodic} and Altland--Zirnbauer~\cite{altlandzirnbauer}. 

Meanwhile, one prominent approach to topological phases with strong interactions has been to explore the low-energy behavior through the associated gapped quantum field theory (QFT) in the continuum.
What makes the assumption that the low-energy effective theory is topological appealing for SPT phases specifically, is a computable classification of potential partition functions using bordism invariants~\cite{kapustin2014symmetry,kapustinfermionicSPT}.
The seminal work of Freed and Hopkins~\cite{freedhopkins} further clarified this approach through the observation that the nondegenerate ground state of the SPT phase requires the zero temperature limit to be an invertible topological QFT.
Moreover, they show that unitary\footnote{The classification of non-unitary invertible field theories is more subtle~\cite{KST, reneesimona, Carmenitft}.} invertible topological QFTs are classified by the associated bordism invariant, thus making a close connection with stable homotopy theory. 
Additionally, a significant advance towards bridging free and interacting phases has been achieved by Freed--Hopkins, who proposed a map between free fermionic and interacting topological phases using the mathematical framework of Thom spectra~\cite[(9.75)]{freedhopkins}.

Despite these advances, several important mathematical aspects concerning \emph{neutral} free fermionic phases have not yet been fully addressed. 
Neutral fermions are relevant for describing topological superconductors such as the Kitaev chain~\cite{kitaev2001unpaired}.
Moreover, interacting SPT classifications using bordism don't necessitate charge assumptions; it is customary to formulate class D using $\Spin$ bordism, as opposed to twisted $\Spin^c$-bordism.
This contrasts sharply with the charge-based formulations common in existing free fermion classifications (such as Freed--Moore and Thiang). 
Although some references have tackled neutral fermions \cite{heinznerzirnbauer, alldridgemaxzirnbauer,ogataCRT, bourne2021locally}, also see \cite[5.1,5.2]{bourne2020}, a systematic, mathematically rigorous classification that clearly compares with the more common charge-conserved setting has been lacking.

In this paper, we fill this gap by developing a rigorous $K$-theoretic classification of free fermionic SPT phases without charge assumptions from first principles. 
 To allow for non-invertible symmetries~\cite{MR4703407}, we assume our free fermions to be protected by a symmetry algebra $A$ as opposed to a group.
 We will always require (even particle-hole) symmetries to commute with the Hamiltonian.
 Mathematically, we allow $A$ to be an arbitrary real $C^*$-algebra with a $\Z_2$-grading encoding potential time-reversing symmetries.
Specifically, we modify the approach of Thiang, utilizing Karoubi triples with negative squares rather than gradations.
More specifically, neutral free fermion systems are typically implemented using Bogoliubov--de-Gennes Hamiltonians, which correspond to Hamiltonians on Fock space which are quadratic in Majorana fields, and so are the correct non-charge-conserving analogue of free Hamiltonians acting on one particle Hilbert space.
Karoubi triples with negative squares are suited to neutral fermions because our flattened gapped Bogoliubov--de-Gennes Hamiltonians will square to $-1$. 
This reasoning results in the group of SPT phases protected by $A$ as proposed in Definition \ref{def:SPT phase}.

\begin{theorem}[Theorem \ref{th:IQPVK-theory}]
    The group of SPT phases protected by $A$ is classified by $K_2(A^{op})$.
\end{theorem}

Here $K_n(B)$ denotes the Karoubi $K$-theory of a real $\Z_2$-graded $C^*$-algebra $B$.\footnote{The reason for the opposite algebra in the theorem is that we use Karoubi's convention for the definition of $\Z_2$-graded $C^*$-algebras, which differs from other conventions, in particular Kasparov's, see Remark \ref{rem:Kasparovconvention}.}
Mathematically, we prove a relationship between Karoubi $K$-theory with negative and positive squares involving a shift by two, which might be of independent mathematical interest, compare \cite{masstermgomiyamashita}.
At first sight, Theorem \ref{th:IQPVK-theory} classifies SPT phases in zero spatial dimensions, but it generalizes to positive spatial dimensions using the crystalline equivalence principle~\cite{crystallineequivalence}, as we discuss in Section \ref{sec:posdim}.

Our framework also allows for symmetry groups as opposed to symmetry algebras through the formalism of fermionic groups $G$ and their associated fermionic group $C^*$-algebras (Definition \ref{def:twisted gp alg}). This allows us to connect directly to twisted $G$-equivariant $K$-theory. 
Moreover, this construction aligns more closely with common interacting phase classifications via twisted spin bordism of classifying spaces $BG$.
We also show how a twisted Peter--Weyl theorem \ref{th: twisted peterweyl} allows one to compute the group of SPT phases from the (real and projective) representation theory of $G$.

As mentioned before, we can then further extend the applicability of our classification framework to spatially extended systems by incorporating translation symmetries into the symmetry group. 
This method also includes the classification of weak topological phases.

In the presence of an additional $U(1)$ symmetry, we demonstrate how our framework recovers known results about the classification of SPT phases for charged fermions:

\begin{theorem}[Theorem \ref{th: main theorem}]
    Let $(G,\phi, \tau,c)$ be an extended QM symmetry class in the sense of Freed--Moore.
    Let $C^*_{uc}(G^\tau)$ be the `unit charge' twisted group $C^*$-algebra (Definition \ref{def:unitcharge}).
    The group of SPT phases protected by $C^*_{uc}(G^\tau)$ graded by $\phi+c$ is isomorphic to the Freed--Moore $K$-theory (Definition \ref{def: Freed--Moore K-theory}), i.e.
    \[
    K_0^c(C^*_{uc}(G^\tau)^{op}) \cong K_2^{\phi + c}(C^*_{uc}(G^\tau)^{op}),
    \]
    where the superscript denotes the $\Z_2$-grading used on the $C^*$-algebra $C^*_{uc}(G^\tau)$.
\end{theorem}

It is well-known that the tenfold way symmetry classes can be represented by $\Z_2$-graded division algebras $A$ over $\R$~\cite{quantumsyms, geikomoore, baez2020tenfold}.
The importance of the tenfold way in condensed matter is a consequence of the mathematical fact that there are exactly ten irreducible building blocks in the representation theory of semisimple $\Z_2$-graded algebras, so that every well-behaved symmetry splits up into a sum of tenfold way types.
A direct computation shows that the group $K_2(A^{op})$ (and its generalization in higher spatial dimensions) recovers the famous classification tables~\cite{kitaev2009periodic}.
In particular, $K_2(A^{op})$ is the natural domain for the free-to-interacting map proposed by Freed--Hopkins~\cite{freedhopkins} for those ten algebras $A$.
More broadly, we conjecture that this $K$-theoretic group structure is indeed the correct mathematical setting for the comparison with interacting SPT phases protected by general symmetry algebras. 
We provide more evidence of this claim by sketching an explicit construction of the free-to-interacting map in $(0+1)$d in Section \ref{sec: interacting}.

Finally, since our approach encompasses both charged and neutral fermions, it provides a robust platform for comparing and clarifying existing tenfold way classification schemes in the literature.
We demonstrate explicitly that different tenfold way approaches yield equivalent $K$-theoretic classifications because they correspond to Morita-equivalent symmetry algebras, thus unifying the at first sight disparate classification frameworks found in the literature such as \cite{schnyder2008classification, altlandzirnbauer} and \cite{freedhopkins}.

The paper is structured as follows. In Section 2, we establish the physical background and motivation for neutral fermion systems. 
Section 3 rigorously defines free fermionic SPT phases within our $K$-theoretic framework, emphasizing the use of Karoubi triples with negative squares. 
It also reviews the approach to the tenfold way using the classification of $\Z_2$-graded division algebras.
In Section 4, we apply the formalism to fermionic group symmetries by constructing twisted group $C^*$-algebras, in particular resulting in a definition of SPT phase in positive spatial dimension. 
Section 5 explicitly compares our neutral fermion framework with the established charged approaches, specifically those by Freed--Moore~\cite{freedmoore} and Thiang~\cite{thiangk-theoretic}. Section 6 compares the tenfold way approach using $\Z_2$-graded division algebras with the more common approach of \cite{schnyder2008classification}.
Section 7 is an outlook section which discusses consequences for comparing free and interacting phases.

\subsubsection*{Conventions}

If $A$ is a $\Z_2$-graded real or complex $C^*$-algebra, $K_n(A)$ denotes its Karoubi $K$-theory in degree $n \in \Z$.
We allow our SPT phases to be a non-trivial invertible phase after forgetting the symmetry.
With the exception of Section \ref{sec: interacting}, we will only consider SPT phases of free fermions, even when we don't mention this explicitly.

\subsubsection*{Acknowledgements}

I am grateful to Cameron Krulewski for numerous enlightening discussions on the topic of this paper.
I would like to thank my PhD advisor Peter Teichner for many useful discussions and giving opportunities to share my work with other researchers in his group. 
I also thank Guo Chuan Thiang for sharing his understanding of phases of Majoranas and answering several elementary questions on $C^*$-algebras.
Finally I wanted to thank Arun Debray, Theo Johnson-Freyd, Kyle Kawagoe, Lukas M\"uller, Natalia Pacheco-Tallaj, David Penneys, David Reutter and Stephan Stolz for useful input.

\section{Physical motivation}
\label{sec: physical motivation}

\subsection{From one particle to Nambu space}
\label{sec: nambu}

To reveal the mathematical structure behind uncharged quantum systems involving free fermions, we start with usual complex quantum mechanics.
So consider one-particle Hilbert space $V$, a complex separable Hilbert space.
In the Nambu formalism~\cite{MR122485}, one then considers the orthogonal direct sum $W = V \oplus V^*$ called Nambu space, where $V^*$ is the continuous dual.
This space can be viewed as the collection of particle annihilation and creation operators in the (Hilbert completed) Fock space $\bigwedge V = \bigoplus_i \bigwedge^i V$ by mapping $v \in V$ to the creation operator
\[
a_v^\dagger := v \wedge ( - )\colon \bigwedge V \to \bigwedge V.
\]
To a covector $f \in V^*$ we assign the annihilation operator $a_f = \iota_{f}$ given by contracting with $f$.

Nambu space has a canonical nondegenerate symmetric complex bilinear pairing
\[
\{(v_1,f_1), (v_2,f_2)\} = f_1(v_2) + f_2(v_1),
\]
known in mathematics as the hyperbolic form.
This bilinear form corresponds to a restriction of the anticommutator on $B(\bigwedge V)$, so that the operators $a_f^\dagger, a_v$ satisfy the Clifford algebra relations with respect to this bilinear form.
These relations are better known as the canonical anticommutation relations.
Explicitly, in an orthonormal basis $\{e_i\}$ of $V$ with dual basis $\{\epsilon^i\}$, write $a_i^\dagger$ for $a_{e_i}^\dagger$ and $a_i$ for $a_{\epsilon_i}$; then the $a_i, a_i^\dagger$ satisfy the anticommutation relations 
\begin{align*}
    \{a_i,a_j\} = \{a_i^\dagger,a_j^\dagger\} = 0, 
    \quad
    \{a_i,a_j^\dagger\} &= \delta_{ij}. 
\end{align*}
Note that $a_i^\dagger$ is indeed the Hermitian adjoint of $a_i$ as an operator on Fock space, as the notation suggests.

Nambu space also has a canonical real structure (i.e. anti-unitary operator of square one)
\[
\gamma :=
\begin{pmatrix}
0 & RF^{-1} \\
RF & 0
\end{pmatrix}
\]
where $RF\colon V \to V^*$ is the (anti-unitary) Riesz-Fr\'echet isomorphism $v \mapsto \langle v, - \rangle$.
From the perspective of annihilation and creation operators, the real structure maps $a_v$ to $a_v^\dagger$ and vice-versa.
Since in the Fermi sea picture, $a_v$ can be reinterpreted as creating a hole in $V^*$, the operator $\gamma$ is called particle-hole conjugation in \cite{zirnbauerparticlehole}.
There is a real Hilbert space $M = \{ w \in W: \gamma(w) = w \}$ of real vectors in Nambu space.
In the annihilation-creation operator perspective, these are the Majorana operators.
In terms of our chosen basis of $V$, one orthonormal basis of $M$ is
\begin{align*}
    \gamma^+_i = \frac{a_i^\dagger + a_i}{\sqrt{2}} \quad \gamma^-_i = i \frac{a_i^\dagger - a_i}{\sqrt{2}}.
\end{align*}
They are normalized to satisfy the Clifford relations
\begin{align*}
    \{\gamma^+_i,\gamma^+_j\} = \{\gamma^-_i,\gamma^-_j\} = \delta_{ij}
    \\
    \{\gamma^+_i,\gamma^-_j\} = 0.
\end{align*}
There is a relation between the Hilbert space structure $\langle .,. \rangle$, the bilinear form and real structure given by
\[
\langle w_1,w_2 \rangle = \{\gamma(w_1), w_2 \}
\]
for $w_1,w_2 \in W$

A charge-conserving free Hamiltonian is a self-adjoint operator $h\colon V \to V$.\footnote{This operator is possibly unbounded, which can give troubles in defining for example commutators. Since in the end we will mainly be interested in flattened Hamiltonians, we will ignore these difficulties.}
A Hamiltonian is defined only up to shifting the spectrum by an arbitrary real scalar, and we exploit this freedom to set the Fermi energy to zero, simplifying the forthcoming formulas.
The Nambu Hamiltonian $H\colon W \to W$ associated to the one particle Hamiltonian $h$ is then defined to be
\[
H = 
\begin{pmatrix}
h & 0 \\
0 & - RF \circ h \circ RF^{-1}
\end{pmatrix}
\]
where $RF \circ h \circ RF^{-1} = (h^\dagger)^*\colon V^* \to V^*= h^*\colon V^* \to V^*$. 
Notice that the spectrum of $h$ has been doubled by a reflection over the Fermi energy $\Spec H = \Spec h \cup \Spec -h$, so $H$ is typically unbounded both above and below.
A Bogoliubov--de-Gennes (BdG) Hamiltonian is a Nambu Hamiltonian $H\colon W \to W$ as above in which we also allow pair-condensing terms of the following form
\[
H = 
\begin{pmatrix}
h & \Delta \circ RF^{-1} \\
- RF \circ \Delta & - RF \circ h \circ RF^{-1}
\end{pmatrix},
\]
where we need the \emph{gap function} $\Delta\colon V \to V$ to satisfy $\Delta^\dagger = - \Delta$ for $H$ to be self-adjoint.
Moreover, $H$ is imaginary with respect to $\gamma$ in the sense that $\gamma H \gamma = -H$.
Conversely, any imaginary self-adjoint $H \colon W \to W$ is a BdG Hamiltonian.\footnote{A useful fact to verify this and other statements in this section is that for a (not necessarily complex-linear) $U\colon W \to W$ decomposed as
\[
\begin{pmatrix}
A & B \circ RF^{-1} \\
RF \circ C & RF \circ D \circ RF^{-1} 
\end{pmatrix},
\] $U$ is real if and only if $A = D$ and $B = C $.}

Define the charge operator $Q\colon W \to W$ to be $+1$ on $V$ and $-1$ on $V^*$. 
Then a BdG Hamiltonian comes from a charge-conserving Hamiltonian $h\colon V \to V$ (i.e. $\Delta = 0$) if and only if $[H,Q] = 0$.
In other words, $Q$ is a symmetry for $H$ if and only if it is charge-conserving.
BdG Hamiltonians are the correct generalization of free charge-conserving Hamiltonians to the neutral setting.
Indeed, they are in bijection with quadratic Hamiltonians on Fock space, also see Example \ref{ex:secondQfreeham}.

According to Wigner's theorem, a symmetry of a one-particle space $U \colon V \to V$ is either unitary or anti-unitary.
We will extend $U$ to a real operator on $W$ and so we have no choice but to let $U$ act by $RF \circ U \circ RF^{-1}$ on $V^*$:
\[
O = 
\begin{pmatrix}
U & 0 \\
0 & RF \circ U \circ RF^{-1}
\end{pmatrix}
\]
Note that $U$ commutes with charge $Q$.
Conversely, every real unitary or anti-unitary $O\colon W \to W$ that commutes with $Q$ comes from such a $U$.

Particle-hole symmetries on the other hand can be modeled by real linear maps $K\colon W \to W$ that map $V$ to $V^*$ and therefore instead anticommute with $Q$ \cite[Definition 2.1]{zirnbauerparticlehole}.
Such particle-hole symmetries are of the form
\[
K = 
\begin{pmatrix}
0 & \Gamma \circ RF^{-1} \\
 RF \circ \Gamma & 0
\end{pmatrix}.
\]
One source of many confusions is that $\Gamma\colon V \to V$ is unitary if $K\colon W \to W$ is anti-unitary and vice-versa.
It can therefore be unclear what an author means when they speak about a unitary particle-hole reversing symmetry; are they talking about the single-particle operator $\Gamma\colon V \to V$ or about $K\colon W \to W$?

In condensed matter, $K$ is often complex antilinear and sometimes referred to as a `chiral symmetry'.
This is the case for a sublattice symmetry such as in the SSH model, see Example \ref{ex: sublattice}.
It is also sometimes assumed that a particle-hole symmetry is a complex antilinear $\Gamma$ corresponding to a complex linear $K$.
One physical example of such is the charge conjugation operator of a massive relativistic Dirac fermion.
For that reason $K$ is still called `charge conjugation' in the literature on topological phases by analogy, see \cite[Example 3]{zirnbauerparticlehole}.

Another source of confusions is that $\Gamma$ anticommutes with a one particle Hamiltonian $h$ if and only if $K$ commutes with the induced BdG Hamiltonian $H$.
In the literature physicists often work on the single-particle Hilbert space $V$ with a single-particle Hamiltonian $h$ and study particle-hole reversing symmetries through $\Gamma$ and therefore they get $\Gamma h = - h \Gamma$ instead of the more reasonable $K H = H K$.
In this case $\Gamma$ is oftentimes called a \emph{pseudosymmetry} (or \emph{antisymmetry}) of $h$ with corresponding physical symmetry $K$.
Note how $\Gamma$ maps states of energy $E$ to states with energy $-E$, simply because $\gamma$ does.
To allow for this situation, we will define symmetries in Definition \ref{def:freefermion} as real operators $O\colon W \to W$, either unitary or anti-unitary, which commute with $H$.

\begin{example}[sublattice symmetry]
\label{ex: sublattice}
Let $V$ be the one particle space of a tight binding model on a lattice $L$ in space.\footnote{For our purposes $L$ is simply a discrete set without accumulation points.}
For example, $L = \Z^d$ is a square lattice in $d$-dimensional space and $V = \ell^2(\Z^d,U)$, where $U$ is a finite-dimensional complex vector space of local degrees of freedom.
Now assume that this lattice $L$ splits up into two subsets $L_A$ and $L_B$ so that $V = V_A \oplus V_B$ is a sum of Hilbert spaces.
Suppose we are given a Hamiltonian $h$ that only has hopping terms between $A$ and $B$, so $h$ is odd with respect to the grading operator $S\colon V \to V$ given by $1$ on $V_A$ and $-1$ on $V_B$.
We can then form a complex antilinear particle-hole reversing symmetry $K$ by taking $S$ to be our $\Gamma$ as above.
We see that this symmetry is antilinear after second quantization, even though it acts complex linearly on the one particle Hilbert space.
A symmetry of this form is usually called a \emph{sublattice symmetry}.
This might be confusing, since naively a sublattice operation definitely preserves the direction of time, but we remind the reader that the naive sublattice operation $S$ itself is just a pseudosymmetry of $h$.
To get the actual symmetry we had to compose with the particle-hole conjugation $\gamma\colon W \to W$, giving an anti-unitary symmetry of $H$.
\end{example}

\begin{table}[h!]
    \centering
    \begin{tabular}{c|c|c}
        $V$ & $W = V \oplus V^*$ & $M$ 
        \\
        charged one particle space & Nambu space & Majorana space 
        \\ \hline
        $h \colon V \to V$ & $H \colon W \to W$ & $\mathfrak{h} \colon M \to M$
        \\
         one particle Hamiltonian & BdG Hamiltonian & skew-adjoint Hamiltonian $iH$ 
         \\ 
    \end{tabular}
\end{table}

\begin{table}[h!]
    \centering
    \begin{tabular}{c|c}
        $\gamma\colon W \to W$ & $J = \mathfrak{h}/|\mathfrak{h}| = -i H/|H|$ 
        \\ \hline 
        particle hole conjugation (real structure) & Polarization associated to $H$
    \end{tabular}
\end{table}

\begin{table}[h!]
    \centering
    \begin{tabular}{c|c}
        $\Gamma\colon V \to V$ & $K \colon W \to W$ 
        \\ \hline 
        Pseudosymmetry, $\Gamma h = - h \Gamma$ & Associated particle hole symmetry, $KH = HK$
    \end{tabular}
    \caption{Overview of physics notation.}
    \label{tab:notation}
\end{table}

\subsection{Neutral free fermions}
\label{sec:neutralfermion}

In order to get a quantum mechanical picture independent of the single particle space $V$, we turn the logic from the last section upside-down and define uncharged free fermion systems only referring to $W$.
The main idea is that since we are not interested in arbitrary interacting Hamiltonians, we don't need to work on all of Fock space, see Example \ref{ex:secondQfreeham} for further discussion.
We will reintroduce charge in Section \ref{sec:charged}.

The material in this section contrasts with the more common charged situation and has appeared in the literature in various forms under various names such as `neutral', `quasifree', `Majorana' and `self-dual'.
The following definitions are motivated by the exposition in the last section.

\begin{definition}
A \emph{Nambu space} $(W, \gamma)$ consists of a complex Hilbert space $W$ and a real structure $\gamma\colon W \to W$ called \emph{particle-hole conjugation}.
The real separable Hilbert space 
\[
M := \{ w \in W: \gamma(w) = w\}
\]
is the space of \emph{Majorana fields}.
A \emph{free fermion (BdG) Hamiltonian} is an imaginary self-adjoint operator $H\colon W \to W$, so $H^\dagger = H$ and $\gamma H \gamma = -H$.
We say $H$ is \emph{gapped} if $0$ is not in the spectrum of $H$.
\end{definition}

\begin{remark}
\label{rem:majoranavsnambu}
Equivalently, we could have started with a real separable Hilbert space $M$ (interpreted as the classical fermionic mode space, i.e. an odd degree symplectic vector space) and define Nambu space as $W = M \otimes_\R \C$. 
It is a complex Hilbert space and $\gamma(zm) = \bar{z} m$ for $m \in M$ and $z \in \C$ defines a canonical real structure on $W$.
Indeed, the category of vector spaces equipped with a real structure is equivalent to the category of real vector spaces, and this generalizes to Hilbert spaces with anti-unitary real structure~\cite{huybrechtscomplex, complexification}.

From a mathematical point of view it is then natural to define a free fermion Hamiltonian as a skew-adjoint operator $iH\colon M \to M$, an approach we will use in the rest of the paper.
One disadvantage of this perspective is that it only makes sense to talk about eigenstates of fixed real-valued energy inside the complexification $W$.
\end{remark}

Given a free fermion Hamiltonian $H$, we can use the spectral theorem for self-adjoint unbounded operators~\cite{rudin} to pick a projection-valued measure $\pi$ such that
\[
H = \int_{\R} E \, d\pi(E).
\]
If $H$ is gapped, define $W_+ := \pi((0,\infty))$ to be the space of \emph{fields above the gap} and $W_- := \pi((-\infty,0))$ the \emph{fields below the gap}, so that there is an orthogonal splitting $W = W_+ \oplus W_-$. 
Define the \emph{spectrally flattened Hamiltonian} as
\[
\frac{H}{|H|} = \int_{\R} \frac{E}{|E|} \, d\pi(E).
\]
Because $iH$ is real and skew-adjoint, the operator $J := -i \frac{H}{|H|}$ restricts to an orthogonal complex structure on $M$.
Note that if $[O,H] = 0$ is a symmetry, then $O J = \pm JO$, where the sign is negative if $O$ is anti-unitary.
Following common language in geometric quantization of fermions~\cite{MR1183739}, we will call $J$ the \emph{polarization} associated to $H$.

In the study of topological phases of matter, we are only interested in deformation classes of Hamiltonians.
In an appropriate topology, every Hamiltonian can be connected with the corresponding polarization by a continuous path in a similar fashion to \cite[Lemma 5.14]{freedmoore}.
Therefore we might as well flatten our Hamiltonians and study (abstract) polarizations in their own right:

\begin{definition}
A \emph{polarization} is a real orthogonal complex structure on $W$.
\end{definition}

\begin{remark}
This notion of a polarization is directly related to the notions of Lagrangians inside Hilbert spaces, compare \cite[Section 2.2]{whatisanellipticobject}, \cite{ludewiglangrangian}.
    The Lagrangian subspace associated to the complex structure $J$ is $W_+ \subseteq W$. 
\end{remark}

\begin{remark}
In \cite{kennedyzirnbauer} and \cite{alldridgemaxzirnbauer}, the polarization associated to $H$ is called its \emph{quasi-particle vacuum}.
The physical reason is that the data of $J$ is equivalent to a many-body ground state in the sense of Hartee-Fock-Bogoliubov mean field theory (which are also known as quasifree states \cite{araki1971quasifree}), see Section \ref{sec: intermezzo}.
\end{remark}

\subsection{Intermezzo: Hartree-Fock-Bogoliubov second quantization}
\label{sec: intermezzo}

We briefly discuss positive energy second quantization to verify that symmetries of BdG Hamiltonians yield symmetries of second quantized Hamiltonians on many-body Hilbert space.
This discussion is also relevant for the generalization to interacting SPT phases.

Recall that a Bogoliubov--de-Gennes Hamiltonian $H$ is typically unbounded below.
At least when $H$ is gapped however, we can still second quantize it in a many-body Fock space in which the vacuum line is the filled sea of negative energy states:
\[
\mathcal{H} = \bigwedge W_+.
\]
Here $\hat{H}:\bigwedge W_+ \to \bigwedge W_+$ is called a second quantization of $H$ if $[\hat{H}, a_w^\dagger] = a_{Hw}^\dagger$ for all $w \in W_+$. 
In the mathematical literature, $\hat{H}$ is called an \emph{implementer} of $H$.
Up to analytic difficulties when $H$ is unbounded\footnote{It might be better at this point to work with the one-parameter unitary group $e^{iHt}$, since quantizing unitary operators is straightforward. 
Some references that handle the unbounded situation are \cite{wurzbacher2001fermionic}, \cite[Section 3.2]{ottesen2008infinite}, and \cite{kristel2020spinor}.}, we have
\[
\hat{H} (w_1 \wedge \dots \wedge w_N) = \sum_{i = 1}^N w_1 \wedge \dots \wedge w_{i-1} \wedge H w_i \wedge w_{i+1} \wedge \dots \wedge w_N
\] 
if $w_1, \dots w_N \in W_+$.

\begin{example}
\label{ex:secondQfreeham}
    Let $\{e_i\}$ be an orthonormal eigenbasis of a one particle Hamiltonian $h\colon V \to V$ with only strictly positive eigenvalues $E_i$.
    Let $H$ be the charge conserving BdG Hamiltonian associated to $h$ so that $W_+ = V \subseteq V \oplus V^* = W$.
    Then up to analytic difficulties we obtain a collection of decoupled (fermionic) harmonic oscillators of frequency $E_i$
    \begin{equation}
    \label{eq:2ndqham}
    \hat{H} = \sum_i E_i a_i^\dagger a_i.
    \end{equation}
For a general BdG Hamiltonian, formula \eqref{eq:2ndqham} then has to be modified with appropriate particle-hole creation and particle-hole annihilation terms.
In fact, a BdG Hamiltonian is the same as an appropriately Weyl-ordered Hamiltonian which is quadratic in the fields, see \cite[Section 2.1]{heinznerzirnbauer} for details.
\end{example}

Suppose $G$ is a finite group of symmetries of $H$.
By this we mean that $W$ is a representation $R$ of $G$, $\gamma R(g) = R(g) \gamma$ and $R(g)$ commutes with the Hamiltonian.\footnote{In particular, we restrict ourselves to symmetries acting on Nambu space, thereby excluding more interesting symmetries acting only on Fock space. Note that we are not assuming the symmetry acts on one particle space, which would exclude particle-hole symmetries.}
This definition also makes sense for time-reversal symmetries which are required to be complex antilinear.
Then $R(g)$ maps $W_+$ to itself and so we can second quantize\footnote{We use the group-theoretic instead of Lie algebra-theoretic variation of second quantization here. So $\widehat{R}(g):\bigwedge W_+ \to \bigwedge W_+$ is called a second quantization of $R(g)$ if $\widehat{R}(g) a_w^\dagger \widehat{R}(g)^{-1} = a_{R(g)w}^\dagger$ for all $w \in W_+$.} it as
\[
\widehat{R}(g)(w_1 \wedge \dots \wedge w_N) = R(g) w_1 \wedge \dots \wedge R(g) w_N.
\]
Since $R(g)$ is unitary (or anti-unitary), it has operator norm $1$, and so $\widehat{R}(g)$ extends to a bounded operator on $\bigwedge V$.
Note that $\widehat{R}(g)$ is (anti-)unitary if $R(g)$ is (anti-)unitary.
An important observation is that $[\widehat{R}(g), \hat{H}] = 0$, so $\widehat{R}(g)$ is indeed a symmetry of the actual physical many-body Hamiltonian $\hat{H}$.

\begin{remark}
\label{rem:-1vs-1F}
Suppose that in the above setting $g \in G$ satisfies $R(g) = -\id_W$.
We then have that 
\[
\widehat{R}((-1)^F)(w_1 \wedge \dots \wedge w_N) = (-1)^N
\]
counts the number of fermions modulo two, which is the natural supergrading $(-1)^F$ on $\mathcal{H}$.
The conclusion of this discussion is that the second quantization $\widehat{R}$ of a free fermion symmetry $R$ satisfies mathematically similar properties as $R$.
However, one difference is that equations of the form $T^2 = (-1)^F$ in second quantization should be replaced by equations of the form $T^2 = -1$ on $W$.
Since we will not consider interactions in this paper, we will refrain from working on $\mathcal{H}$ in this document.
\end{remark}

If we are given a charge-conserving gapped Hamiltonian $h \colon V \to V$ (or just its polarization), we can decompose the unit charge Nambu space $V \oplus V^*$ further into positive-energy fields and negative energy fields as
\[
V_+ \oplus V_- \oplus V_+^* \oplus V_-^*.
\]
In a condensed matter situation, we might interpret $V$ as the one-particle space of electrons. 
Then $V_+$ is the space of (creation operators of) conduction electrons and $V_-$ the space of valence electrons, i.e. it prescribes the decomposition of the one particle Hamiltonian $h \colon V \to V$ into positive and negative eigenstates.
Note that $W_+ = V_+ \oplus V_-^*$.
In the Hartree-Fock-Bogoliubov mean field theory picture with respect to our preferred quasi-particle vacuum however, we re-interpret the space of valence electron annihilation operators $V_-^*$ as the space of hole excitations (conducting holes in the valence band). 
The act of removing a valence electron from the filled Fermi sea in the original picture is equivalent to adding the quasiparticle excitation in the new  picture. 
Second quantizing $e^{itQ}$ on the positive energy Fock space (Section \ref{sec: intermezzo}) yields 
\[
\widehat{e^{itQ}}(v_1 \wedge \dots \wedge v_p \wedge \phi_1 \dots \wedge \phi_q) = e^{i(p-q)t}
\]
if $v_1 ,\dots, v_p \in V_+$ are conducting electrons and $\phi_1, \dots, \phi_q \in (V_-)^*$ are conducting holes.
So holes have the opposite charge of electrons as expected.
Also note that a second quantized particle-hole symmetry (i.e. a symmetry that anticommutes with $Q$) will map electrons to holes and vice-versa.

The discussion above contrasts with the situation in which we start with a one-particle space $V$ and create the Fock space $\bigwedge V$ in which the vacuum line is the true vacuum
\[
\wedge^0 V \subseteq \bigwedge V.
\]
Namely, here we can second quantize a BdG Hamiltonian if and only if the off-diagonal term $\Delta\colon V \to V^*$ is Hilbert-Schmidt, for example if $H$ is charge-conserving.
Similarly, we can second quantize unitary and anti-unitary real $S\colon W \to W$ on $\bigwedge V$ if and only if $[T, Q]$ is Hilbert-Schmidt.
For example, $\gamma$ can never be second quantized if $V$ is infinite-dimensional because it would have to map the vacuum to a multiple of $e_1 \wedge e_2 \wedge \dots$ for some orthonormal basis $\{e_i\}$ of $V$.
This cannot happen as $e_1 \wedge e_2 \wedge \dots \notin \bigwedge V$ cannot be written as a sum over many-body states with only finitely many particles.

For a more physically relevant example, particle-hole reversing symmetries $K$ anticommute with $Q$ by definition and so can only be second quantized on $\bigwedge V$ if $V$ is finite-dimensional.
Concretely, in case $\dim V = N < \infty$ this is done as follows.
Write $K = \gamma \Gamma$ for some unitary $\Gamma$ with $[\Gamma, Q] = 0$ as in Section \ref{sec: nambu}, then $K$ can be second-quantized as the composition
\[
\bigwedge^n V \xrightarrow{\hat{\Gamma}} \bigwedge^n V \xrightarrow{\hat{\gamma}} \bigwedge^{N-n} V,
\]
where $\hat{\gamma}$ can be shown to be a complex antilinear version of the Hodge star after a (noncanonical) choice of orientation on $V$, see \cite[Sections 3.2,3.3]{zirnbauerparticlehole}.
Note how this symmetry maps a many-body state consisting of $n$ electrons to a many-body state consisting of $n$ holes.

We also briefly connect with the viewpoint of algebraic quantum mechanics.
In the current setting we can form the Clifford algebra $Cl_{alg}(W)$ on the vector space $W$ using the symmetric bilinear form $\{.,.\}$.
Extend the real structure $\gamma$ on $W$ uniquely to a $*$-structure on $Cl_{alg}(W)$.
Because $\langle .,. \rangle = \{ \gamma(.), . \}$ is positive definite, every element of the form $a^* a$ has positive spectrum.
We can define $Cl(W)$ to be its completion to a $C^*$-algebra.
Given a gapped BdG Hamiltonian, the positive energy Fock space on $W_+$ is an irreducible $\Z_2$-graded representation of $Cl(W)$ and every representation is of this form.
See \cite[Section 3]{kristel2020spinor} for mathematical details.

Note that if $W$ is odd-dimensional, $Cl(W)$ has no Fock modules, in other words there exists no gapped BdG Hamiltonian.
This is a mathematical interpretation of the fact that an odd number of Majoranas in one spacetime dimension have an anomaly, see \cite[Section 3.1]{delmastro2021global} and \cite{freed2024odd}.
Of course if we constructed $W$ from a one particle Hilbert space as $V \oplus V^*$, this situation can never occur.

\section{Neutral free fermion SPT phases and \texorpdfstring{$K$}{K}-theory}

\subsection{Symmetry-protected free fermions}

We now turn to the formal implementation of  symmetry algebras $A$ protecting free fermion systems.
As the main example the reader should have the group algebra $\R[G]$ of a finite group in mind, which will recover a system protected by the symmetry $G$.
In Section \ref{sec:groups}, we will study phases protected by symmetry groups in detail.

We will assume $A$ is a $C^*$-algebra to get the appropriate topologies needed to define topological phases.
This requirement is also reasonable from the perspective of unitarity in algebraic quantum mechanics.
We will for now also assume that $A$ is unital and that all homomorphisms preserve the unit, leaving a more detailed discussion to Remark \ref{rem:nonunital}.

The difference between time-preserving and time-reversing symmetries will be encoded by a $\Z/2$-grading on $A$, and the degree of $a \in A$ will be denoted by $|a| \in \{0,1\}$.
We emphasize that physically this $\Z/2$-grading has no relation with the distinction between bosons and fermions.
In this section, we will assume readers are familiar with the theory of $\Z/2$-graded real $C^*$-algebras, see Appendix \ref{sec: C*} for the aspects we need here.
Physically, we need real as opposed to complex $C^*$-algebras because we work without $U(1)$-charge.

Recall from Section \ref{sec:neutralfermion} that in a quantum mechanical system protected by the symmetry algebra $A$, we can obtain a polarization by multiplying a flattened Hamiltonian by $i$.
Symmetries $a \in A$ should commute with this Hamiltonian for them to second quantize to symmetries that commute with the second quantized Hamiltonian as discussed in Section \ref{sec: intermezzo}.
    Since odd elements of $A$ model time-reversing symmetries which anticommute with $i$, this motivates us to require these to anticommute with $J$ in the following definition.

\begin{definition}
\label{def:freefermion}
Let $A$ be a $\Z/2$-graded $C^*$-algebra and $M$ an $A$-module (without a $\Z/2$-grading).
    An \emph{$A$-symmetric polarization} is a linear map $J\colon M \to M$ such that
    \[
    J^2 = - \id_{M} \quad Ja = (-1)^{|a|} aJ.
    \]
\end{definition}

In physical contexts, we will refer to $A$ as the \emph{symmetry algebra} and $M$ as the \emph{Majorana space protected by $A$}.

\begin{remark}
    Physically, it is reasonable to require $M$ to have some Hilbert space-like structure.
    For the most general connection with $K$-theory, this inner product needs to be an $A$-valued Hilbert module structure on $M$.
    This notion is reviewed in Appendix \ref{sec: C*} and is in general not equivalent to a $\C$-valued inner product, see Remark \ref{rem:repsvsmods}. 
\end{remark}


\begin{remark}
\label{rem:hamsvspol}
    Similarly to the case without the symmetry, $A$-symmetric polarizations are in practice usually flattened gapped $A$-symmetric BdG Hamiltonian $\mathfrak{h}$ on an $A$-module $M$.
    Let $\mathfrak{h}$ be an invertible linear map $\mathfrak{h}\colon M \to M$ such that $\mathfrak{h}a = (-1)^{|a|} a\mathfrak{h}$ for all $a \in A$.
    Suppose $\mathfrak{h}$ is skew-adjoint with respect to a given Hilbert $A_{ev}$-module structure $\langle .,. \rangle$.\footnote{An adjointable operator for an $A$-valued inner product on $M$ would automatically be $A$-linear~\cite[Lemma 15.2.3]{weggeolsen}, so in that setting we would not be able to talk about skew-adjointness of $\mathfrak{h}$.}
    We can then define its flattening as $\mathfrak{h}/|\mathfrak{h}|$.
    Indeed, adjointable operators form a $C^*$-algebra and the function $\Spec \mathfrak{h} \to \{-i,i\}$ given by $x \mapsto x/|x|$ is continuous, so we can use functional calculus.
    Note that $\mathfrak{h}/|\mathfrak{h}|$ is an $A$-symmetric polarization.
    By the argument presented in Section \ref{sec:neutralfermion}, we expect the inclusion of $A$-symmetric polarizations into the space of gapped BdG Hamiltonians to be a deformation retract.
    For topological properties of the space of Hamiltonians, it therefore suffices to study $A$-symmetric polarizations instead.
    We also expect the above considerations to generalize to unbounded Hamiltonians.
\end{remark}

\begin{remark}
\label{rem:symmetricmajoranavsnambu}
    Only after complexifying to $W := M \otimes_\R \C$ can one consider the operator $H:=-i\mathfrak{h}$.
    Then $H$ will be self-adjoint and imaginary in the canonical real structure, compare Remark \ref{rem:majoranavsnambu}.
    If we extend the action of $a \in A$ to $W$ complex-antilinearly for odd elements as motivated by the discussion in Section \ref{sec: nambu}, then $H$ will be genuinely $A$-linear as desired from physics.
\end{remark}

\subsection{The group of free SPT phases}
\label{sec:the group of free SPT}

To mathematically define topological phases of matter, we need to consider deformations of physical systems subordinate to specified symmetry constraints.
In the formalism of this paper, this would mean SPT phases should be defined as path components of the space of certain gapped Hamiltonians.
In other words, we will consider homotopy classes of gapped BdG Hamiltonians commuting with this symmetry.
Note that since we are working with free fermions, stacking is defined by direct sum as opposed to tensor product.

\begin{definition}
Let $(M_1, \mathfrak{h}_1)$ and $(M_2, \mathfrak{h}_2)$ be $A$-symmetric free fermion systems.
Then their \emph{stacking} is the $A$-symmetric free fermion system $(M_1 \oplus M_2, \mathfrak{h}_1 \oplus \mathfrak{h}_2)$.
\end{definition}

Even though the space of $A$-symmetric free fermion systems is a commutative monoid under stacking, it is clearly not a group.
To get the desired connection with $K$-theory, it will be necessary to quotient out by a submonoid of `trivial Hamiltonians'.\footnote{We are not aware of a good physical argument to require this stabilization maneuver. In fact, nontrivial phases that become trivial after stabilization are being studied under the name `fragile' topological phases~\cite{fragile}.}
To do this rigorously, we will consider pairs of Hamiltonians as in \cite{thiangk-theoretic}, which we think of as formal differences.\footnote{The reader should be warned that we do not simply Grothendieck complete the monoid of free fermion systems as this will not give the correct $K$-theory group.}
For technical reasons, we will restrict to finitely generated modules, see Remark \ref{rem:infinitedimKaroubi} for further discussion.
The resulting definition of SPT phase is as follows.

\begin{definition}
    Let $A$ be a $\Z_2$-graded $C^*$-algebra and $M$ a finitely generated\footnote{Whenever we say finitely generated, we mean in the algebraic sense, see \cite[Section 15.4]{weggeolsen} for a discussion.}  projective $A$-module.
    Let $\Pol_A(M)$ be the space\footnote{For details on how to obtain the topology on this space, see Section \ref{sec:Karoubi}.} of $A$-symmetric polarizations on $M$.
    \begin{enumerate}
        \item Let $\Diff_A(M) = \Pol_A(M) \times \Pol_A(M)$, which we think of as formal differences of polarizations;
        \item A formal difference $(J_1,J_2) \in \Diff_A(M)$ of $A$-symmetric polarizations on $M$ is called \emph{elementary} if $J_1$ and $J_2$ are homotopic in the space $\Pol_A(M)$;
    \item Two formal differences of $A$-symmetric polarizations $(J_1,J_2) \in \Diff_A(M)$ and $(J_1',J_2') \in \Diff_A(M')$ are \emph{isomorphic} if there is an $A$-module isomorphism  $M \to M'$ which intertwines $J_1$ with $J_1'$ and $J_2$ with $J_2'$;
    \item The sum of $T = (J_1,J_2) \in \Diff_A(M)$ and $T' = (J_1',J_2') \in \Diff_A(M')$ is defined as 
    \[
    T \oplus T' := (J_1 \oplus J_1', J_2 \oplus J_2') \in \Diff_A(M \oplus M').
    \]
    \end{enumerate}
\end{definition}

\begin{definition}
    \label{def:SPT phase}
    The set $\SPT^A$ of \emph{(free fermionic) SPT phases protected by $A$} is the quotient of the monoid of differences of $A$-symmetric polarizations quotiented by the equivalence relation
    \[
T_1 \sim  T_2 \iff \exists T_1',T_2' \text{ elementary s.t. } T_1 \oplus T_1' \cong T_2 \oplus T_2'.
\]
\end{definition}

In Section \ref{sec: K-theory for free}, we will see how this definition is related to $K$-theory.

\begin{remark}
\label{rem:infinitedimKaroubi}
    It is possible to allow infinite-dimensional separable Hilbert modules in the above definition. 
    This is physically well-motivated, because one particle Hilbert space is typically not finite rank over the symmetry algebra.
    However, in that case it is necessary to put a finiteness condition on the difference between the two polarizations to avoid the Eilenberg swindle.
    For example, one could work with compact perturbations of a fixed Hamiltonian.
    Since we expect the result to be isomorphic\footnote{See \cite[Section 4.3]{gomi2017freed} for a result of this type in the setting of $K$-theory of topological stacks.}, we will not follow this direction here for simplicity.
\end{remark}

\begin{remark}
    Note that for an element $[T] = [(J_1, J_2)] \in \SPT^A$, we do not require any compatibility conditions between $J_1$ and $J_2$.
\end{remark}

\begin{remark} 
The group $\SPT^A$ is Karoubi's relative $K$-theory of the functor from the Banach category of (ungraded) $(Cl_{-1} \otimes A)$-modules to $A$-modules.
    We conjecture that the relative Karoubi $K$-theory of a functor $\phi \colon \mathcal{C} \to \mathcal{C}'$ is given by the homotopy groups of the fiber of an induced map $K(\phi) \colon K(\mathcal{C}) \to K(\mathcal{C}')$. Here $K(\mathcal{C})$ is the periodic $K$-theory spectrum of a Banach category.
    Karoubi $K$-theory will be reviewed in more detail in the next section.
\end{remark}

We take a look at the case without symmetries, giving another motivation for the approach using pairs of Hamiltonians. 
See \cite{thiang2015topological} for a similar argument using class AIII topological insulators.

\begin{example}[class D]
Consider a BdG system in spatial dimension zero with $A = \R$.
According to standard tenfold way classification tables the group of SPT phases in this setting is $KO^{-2}(\pt) \cong \Z_2$.

A finitely generated Majorana space (protected by $A = \R$) is the same as a real vector space $M$, which we imagine to sit at the only site in space.
A Hilbert module structure on $M$ is a real inner product.
The space of polarizations is the space of complex structures on $M$.
This space is only nonempty if $\dim M$ is even, which we will assume from now on.
In that case the space is homotopy equivalent to $O(\dim M)/U(\dim M/2)$, which in the limit as $\dim M$ goes to infinity forms a classifying space for $KO^{-2}$.
For $\dim M$ finite, this space also has two path components, labeled by the sign of the Pfaffian of $iH$. 
However, the important subtlety is that the Pfaffian depends on a fixed orientation on $M$.
An equivalent way to arrive at the same conclusion is to realize that there are two path components of complex structures on $M$, but no canonical one.


We conclude that there are two path components of gapped Hamiltonians, but in the current formulation there is no well-defined notion of `trivial phase' on a fixed $A$-module unless one makes arbitrary choices.
In more mathematical language, if we stick to defining SPT phases as path components of gapped BdG Hamiltonians, we will obtain that class D topological phases form a $\Z_2$-torsor.

If we insist on making equivalence classes of BdG Hamiltonians into a group, we thus have to fix an arbitrary reference polarization. 
In practice, one takes this reference polarization to be the bare vacuum $J = iQ$ associated to some $U(1)$-charge $Q$ in case $M \otimes_\R \C = V \oplus V^*$ comes from a one particle space $V$ as in Section \ref{sec: nambu}.
However, this approach will not work if $A$ contains particle-hole symmetries as the bare vacuum is not invariant under them, also see the explanation of charge symmetries at the end of Section \ref{sec: Freed--Moore symmetry}.
\end{example}

\subsection{Karoubi \texorpdfstring{$K$}{K}-theory}
\label{sec:Karoubi}

In this more technical subsection, we will review Karoubi $K$-theory.
Another common model of $K$-theory that is closely related to topological phases of matter is van Daele $K$-theory~\cite{vandaele1, vandaele2, kellendonk2017c, alldridgemaxzirnbauer, lorenzo2024comparison}.
We start by introducing the classical notions of Karoubi triples following~\cite{thiangk-theoretic}.

Let $A$ be a $C^*$-algebra and $M$ a finitely generated projective module over $A$.
The fact that all finite-dimensional real vector spaces admit the structure of a real Hilbert space generalizes to the fact that all finitely generated projective $A$-modules $M$ admit the structure of a Hilbert module ~\cite[Theorem 15.4.2]{weggeolsen}.
Indeed, if we pick a complement $A$-module $N$ so that $M \oplus N \cong A^n$, then we can restrict the canonical Hilbert module structure on $A^n$ to $M$.
This gives a topology on $M$ and an operator norm topology on $B(M)$, the continuous linear endomorphisms of $M$. 
Just as for vector spaces, the choice of a Hilbert module structure on $M$ is noncanonical and can lead to different norms.
However, the induced topologies on $M$ and hence on $B(M)$ are independent of the chosen embedding in $A^n$~\cite[Section 11.2]{blackadar}.

As a special case, we can look at the $A$-module endomorphisms of the finitely generated module $M$.
It is well-known that such maps are automatically adjointable and hence continuous. 
In particular, they form a $C^*$-algebra and the $C^*$-topology agrees with the topology on $B(M)$ as discussed above. 
In particular, adjointable endomorphisms of $A^n$ form the $C^*$-algebra $M_n(A)$~\cite[15.2.12]{weggeolsen}.
In this way the category $\Mod_A$ of finitely generated projective $A$-modules becomes a topological category.

Now suppose $A$ is additionally $\Z_2$-graded but $M$ is not necessarily $\Z_2$-graded.
Since polarizations and grading operators are not $A$-linear, we don't get a topology on these spaces by realizing them as morphisms in $\Mod_A$.

If $T:M \to N$ is a real-linear map between ungraded $A$-modules, we call it \emph{$A$-skewlinear} if $T(ma) = (-1)^{|a|} T(m) a$. 
Note that every $A$-skewlinear map is bounded as is it is $A_{ev}$-linear and every finitely generated projective $A$-module is finitely generated projective over $A_{ev}$.
Therefore, even though the set of $A$-skewlinear operators $M \to N$ is not contained in $\Mod_A(M,N)$, it still inherits a topology from the topology on continuous maps from $M$ to $N$.
Note that a $\Z_2$-grading on $M$ is equivalent to an $A$-skewlinear endomorphism with square $1$ and so we get:

\begin{definition}
Let $M$ be a Hilbert $A$-module.
Let $\Grad_A(M)$ be the set of $\Z_2$-gradings of $M$, i.e. skew-linear maps $\epsilon \colon M \to M$ such that $\epsilon^2 = \id_M$.
Since such operators will automatically be bounded, $\Grad_A(M)$ inherits an operator norm from $B(M)$.
    \end{definition}

Similarly note that an $A$-symmetric polarization is an $A$-skew-linear endomorphism with square $-1$, giving a topology on $\Pol_A(M)$.

\begin{remark}
Recall that if $M$ is a Hilbert $A$-module with grading operator $\epsilon$, then we can't ask $\epsilon$ to be self-adjoint unless $A = A_{ev}$.
On the other hand, $M$ is said to define a \emph{graded Hilbert module} if and only if 
\[
(-1)^{|\langle m_1, m_2 \rangle|} \langle m_1, m_2 \rangle = \langle \epsilon(m_1), \epsilon(m_2) \rangle
\]
for all $m_1, m_2 \in M$.
We therefore expect this to be the correct replacement of `self-adjoint grading'.
We also expect that the subset of gradings satisfying this condition is a deformation retract in the space of all gradings.
    Similar to Remark \ref{rem:hamsvspol}, we expect that we can weaken the assumption that $\epsilon$ is a self-adjoint grading to an arbitrary invertible operator $h$ satisfying $h a = (-1)^{|a|} a h$ which is self-adjoint in some appropriate sense.
\end{remark}

\begin{remark}
    Note that the space of invertible $A$ skew-linear automorphisms $T \colon M \to M$ is in general not homotopy equivalent to $\Grad_A(M)$.
    For example, for $A = \R$ and $M = \R^n$ we have that 
    \[
    \Grad_A(M) = \bigsqcup_{0 \leq k \leq n} Gr_k(\R^{n-k}),
    \]
    where $Gr_k(\R^{n-k})$ is the Grassmannian of $k$-dimensional subspaces of $\R^{n-k}$.
    We see that $\Grad_\R(\R^n)$ has $n+1$ components, while $GL_n(\R)$ has $2$ for all $n>0$.
\end{remark}

 The following definition due to Karoubi only differs from Definition \ref{def:SPT phase} by a sign.

\begin{definition}
Let $M$ be a finitely generated projective $A$-module.
\begin{itemize}
    \item The space of \emph{Karoubi triples} is the space of pairs $\epsilon_1,\epsilon_2 \in \Grad_A(M)$.
\item A Karoubi triple $(M,\epsilon_1,\epsilon_2)$ is called \emph{elementary} if $\epsilon_1$ and $\epsilon_2$ are homotopic inside the space $\Grad_A(M)$. 
\item The sum of two Karoubi triples $(M, \epsilon_1, \epsilon_2), (M',\epsilon_1', \epsilon_2')$ is $(M \oplus M', \epsilon_1 \oplus \epsilon_1', \epsilon_2 \oplus \epsilon_2')$.
\item An \emph{isomorphism} of Karoubi triples $(M, \epsilon_1, \epsilon_2) \cong (M',\epsilon_1', \epsilon_2')$ consists of an isomorphism $M \to M'$ of Hilbert $A$-modules which intertwines the $\epsilon_i'$.
\item The \emph{Karoubi $K$-theory} $K_0(A)$ of $A$ is defined as the collection of Karoubi triples modulo the equivalence relation that  
\[
T_1 \sim  T_2 \iff \exists T_1',T_2' \text{ elementary s.t. } T_1 \oplus T_1' \cong T_2 \oplus T_2'.
\]
It is a group under $\oplus$.
\end{itemize}
\end{definition}

If $\phi \colon A \to B$ is a graded $*$-homomorphism, there is an induced map $K_0(A) \to K_0(B)$.
Indeed, there is a direct sum preserving functor from finitely generated projective $A$-modules to finitely generated projective $B$-modules given by $M \mapsto B \otimes_A M$.

Let $Cl_n = Cl_{+n}$ denote the real Clifford algebra on $n$ generators with positive squares, which we view as a $\Z_2$-graded $C^*$-algebra.
Similarly, $Cl_{-n}$ denotes the Clifford algebra with negative squares and $Cl_{p,q}$ the mixed signature Clifford algebra with $p$ positive and $q$ negative squares.

\begin{example}
\label{ex:Karoubicomputation}
Let us compute $K_0(Cl_{+1})$. 
Note that $Cl_{+1} \cong \R \oplus \R$ as an ungraded algebra and so a finitely generated projective $Cl_{+1}$-module is just a $\Z_2$-graded vector space.
We will denote the $\Z_2$-graded vector space $\R^{p+q} = \R^p \oplus \R^q$ by $\R^{p|q}$.
This graded $\R$-module admits a further $\Z_2$-grading making $Cl_{+1}$ act in a graded way if and only if $p = q$.
A grading on $\R^{p|p}$ can then be written
\[
\begin{pmatrix}
    0 & A \\
    A^{-1} & 0
\end{pmatrix}
\]
for some $A \in GL_p(\R)$.
There are two homotopy classes of such $A$ if $p \neq 0$.
It is straightforward to verify that there is a group isomorphism $K_0(Cl_{+1}) \to \Z_2$ given by mapping a Karoubi triple $[\R^{p|p}, A_1, A_2]$ to zero if $A_1$ and $A_2$ are in the same path component and $1$ otherwise.
\end{example}

More generally, we have that $K_0(Cl_{+n}) \cong KO_n(\pt)$ recovers the Bott song, motivating the following definition.

\begin{definition}
\label{def:higherkarkth}
The \emph{higher Karoubi $K$-theory} of $A$ is
    \[
    K_n(A) := K_0(Cl_{+n} \otimes A).
    \]
\end{definition}

\begin{remark}\label{rem:nonunital}
    Up until this point we have assumed that all $C^*$-algebras are unital and $*$-homomorphisms preserve the unit.
    However, many relevant examples of symmetry algebras---such as the group $C^*$-algebras of infinite compact groups of Section \ref{sec:infinitegroups}---are non-unital.
    The treatment of non-unital algebras in $C^*$-algebra textbooks is standard, and we summarize the content most relevant for this work here.
    Not necessarily unital $C^*$-algebras admit several unitizations. 
Geometrically, one should think of unitizations as compactifications, as this is what they correspond to for commutative $C^*$-algebras under Gelfand's representation theorem.
The smallest unitization $A_+ := A \oplus \C$ is called the \emph{Dorroh extension}~\cite{MR1562332}.
It can be shown to have a $C^*$-norm and corresponds under Gelfand duality to the one-point compactification in the commutative case.
The multiplier algebra on the other hand is in some sense the largest unitization, see Definition \ref{def:multiplieralg} for details.

    If $A$ is non-unital, there is an equivalence of categories between $A$-modules $M$ and $A_+$-modules satisfying $1 \cdot m = m$ for all $m \in M$.
    We call an $A$-module finitely generated projective if it is finitely generated and projective as an $A_+$-module, see \cite{1540627} for a related discussion.
    If $X$ is a locally compact Hausdorff space which is noncompact, then functions $C_0(X)$ on $X$ vanishing at infinity are a non-unital $C^*$-algebra.
    We have that $C_0(X)_+ = C(X_+)$ is the $C^*$-algebra of functions on the one point compactification $X_+$ of $X$.
    In particular, finitely generated modules over $C_0(X)$ are vector bundles over $X_+$, not over $X$, compare \cite{rennie2001poincare}.
    
    As is common both in the $C^*$-algebra and topological $K$-theory literature, we define the $K$-theory of a non-unital $\Z_2$-graded $C^*$-algebra $A$ as
\[
K_n(A) := \ker \left( K_n(A_+) \to K_n(\mathbb{F})\right),
\]
where the map is induced by the obvious (split) $*$-homomorphism $A_+ \to \mathbb{F}$.
Here we used the grading on $A_+$ induced by the grading on $A$.
Note that from a cohomological perspective, this defines $K$-theory with compact supports.
We make the analogous definition for $\SPT^A$ when $A$ is nonunital.
\end{remark}



\subsection{Free SPT phases and K-theory}
\label{sec: K-theory for free}

The goal of this section is to show that the version of Karoubi $K$-theory using polarizations we introduced in Section \ref{sec:the group of free SPT} is related to the usual notion of Karoubi $K$-theory up to a shift, compare \cite{masstermgomiyamashita}.
For the proof, we will need the notion of an opposite $C^*$-algebra.

\begin{definition}
    Let $A$ be a $\Z_2$-graded $C^*$-algebra.
The opposite $\Z_2$-graded $C^*$-algebra $A^{op}$ is defined to be equal to $A$ as a vector space and
\[
a^{op} b^{op} := (-1)^{|a||b|} (b a)^{op} \quad (a^{op})^* = (-1)^{|a|} (a^*)^{op}.
\]
It can be shown that $A^{op}$ is a $\Z_2$-graded $C^*$-algebra.   
\end{definition}

If $A,B$ are two $\Z_2$-graded $C^*$-algebras, we have $(A \otimes B)^{op} \cong B^{op} \otimes A^{op}$.\footnote{Recall that the tensor product between $\Z_2$-graded algebras has its product defined with the appropriate Koszul sign $(a_1 \otimes b_1) (a_2 \otimes b_2) = (-1)^{|b_1||a_2|} a_1 a_2 \otimes b_1 b_2$, see Appendix \ref{sec: C*} for details.}
Note that $(Cl_{p,q})^{op} \cong Cl_{q,p}$.

There is an equivalence of categories between $\Z_2$-graded right $A$-modules and $\Z_2$-graded left $A^{op}$-modules given by the formula
\[
a^{op} m = (-1)^{|a| |m|} m a.
\]
The equivalence restricts to the subcategory of finitely generated projective modules.

\begin{definition}
Let $A$ be a $\Z_2$-graded $*$-algebra and $M$ a $\Z_2$-graded left $A$-module.
Define the \emph{conjugate of $M$} as the $A^{op}$-module $\overline{M}$ which is equal to $M$ as a $\Z_2$-graded vector space with module structure
\[
a^{op} \cdot m  := (-1)^{|a||m|} a^* m
\]
If $T\colon M \to N$ is a graded $A$-module map, define $\overline{T}\colon \overline{M} \to \overline{N}$ by $T$ when $T$ is even and $T \circ \epsilon_M$ when $T$ is odd, where $\epsilon_M\colon M \to M$ is the grading operator $m \mapsto (-1)^{|m|} m$.
\end{definition}

\begin{warning}
Note that $\overline{\overline{M}} \neq M$ in general.
Instead, it is the $A$-module in which the action of $a \in A$ is twisted by $(-1)^{|a|}$.
\end{warning}

\begin{remark}
We could have also defined $\overline{T}$ by $\epsilon_N \circ T$ for odd $T$.
This definition would differ from the above by $(-1)^{|T|}$. 
We have no argument to prefer our definition over the other.
\end{remark}

\begin{remark}
\label{rem:opvstau}
    If $A$ is a $\Z_2$-graded algebra, let $A^\tau$ denote the $\Z_2$-graded algebra that is equal to $A$ as a graded vector space but has multiplication
    \[
    a^\tau \cdot b^\tau := (-1)^{|a||b|} (a b)^\tau,
    \]
    where we denoted the element of $A^\tau$ corresponding to the element $a \in A$ by $a^\tau$.
    If $A$ is additionally a $\Z_2$-graded $*$-algebra, then $*$ gives an isomorphism of $\Z_2$-graded algebras $A^{op} \cong A^\tau$ and thus a correspondence $M \to \widetilde{M}$ between $A^{op}$-modules and $A^\tau$-modules.
    Note that if $*$ is complex antilinear, then so is this algebra isomorphism.
    If we define $(a^\tau)^* := (-1)^{|a|} (a^*)^\tau$, then this becomes a $*$-algebra isomorphism.
    If $A$ is $C^*$, then so is $A^\tau$.
    For $M$ a $\Z_2$-graded $A$-module the $A^\tau$-module $M^\tau$ defined by
    \[
    a^\tau \cdot m^\tau = (-1)^{|a||m|} (am)^\tau
    \]
    corresponds to $\ol{M}$ under the isomorphism $A^\tau \cong A^{op}$.
    Note that we do have $M^{\tau \tau} = M$.
    This does not contradict the fact that $\ol{\ol{M}} \ncong M$ because in general $\widetilde{M^\tau} \ncong (\widetilde{M})^{\tau}$ under the obvious identification $(A^\tau)^{op} = (A^{op})^\tau$ with the ungraded opposite.
\end{remark}

Let $\mathbf{GrMod}_A$ be the category of finitely generated projective $\Z_2$-graded $A$-modules and graded-linear $A$-module maps.
In other words, morphisms are either even and $A$-linear or odd and $T(am)=(-1)^{|a|}aT(m)$.
Let $\mathbf{GrMod}^{ev}_A$ denote that category where we restricted to the even degree maps.
Note that a finitely generated projective $\Z_2$-graded $A$-module is equivalent to a finitely generated projective ungraded $A \otimes Cl_{+1}$-module by making the grading operator $\epsilon$ the action of the preferred odd generator of $Cl_{+1}$.
Under that correspondence we have $\mathbf{GrMod}^{ev}_A \cong \mathbf{Mod}_{A \otimes Cl_{+1}}$.
We thus get that $\mathbf{GrMod}^{ev}_A$ is a topological category.
Moreover, odd graded $A$-linear maps correspond to skew-$A$-linear maps, which also have a topology.
Therefore all of $\mathbf{GrMod}_A$ also defines a topological category.

\begin{remark}
    Recall that there is also a notion of $\Z_2$-graded Hilbert module, but we prefer to work with $A \otimes Cl_{+1}$-valued inner products here, also see Remark \ref{rem:gradedvsClvalued}.
\end{remark}

\begin{lemma}
\label{lem:barequiv}
The conjugate defines an equivalence of topological categories
\[
\overline{(.)}\colon \mathbf{GrMod}^{ev}_A \cong \mathbf{GrMod}^{ev}_{A^{op}},
\]
which preserves direct sums.
It also defines homeomorphisms $\mathbf{GrMod}_A(M_1, M_2) \to \mathbf{GrMod}_{A^{op}}(\ol{M_1}, \ol{M_2})$ preserving direct sums, but is not a functor $\mathbf{GrMod}_A \to \mathbf{GrMod}_{A^{op}}$.
 Instead we have
\[
\overline{T}_1 \overline{T}_2 = (-1)^{|T_1| |T_2|} \overline{T_1 T_2}.
\]
\end{lemma}
\begin{proof}
The proof is a straightforward sign computation.
For bookkeeping, we will write the element of $\ol{M}$ corresponding to $m \in M$ as $\ol{m}$.
The fact that $\overline{M}$ is an $A^{op}$-module follows from
\begin{align*}
    (a^{op} b^{op}) \cdot \ol{m} &= (-1)^{|a||b|} (ba)^{op} \cdot \ol{m} 
    \\
    &= (-1)^{|a||b| + |ba| |m| }  \ol{(ba)^* m}
    \\
    &= (-1)^{|a||b| +  (|a| + |b|) |m|}   \ol{a^* b^* m}
    \\
    &= (-1)^{|a||b| + |a||m| + |b||m| +|a||b^* m|} a^{op} \ol{b^* m}
    \\
    &=  a^{op} (b^{op} \ol{m }).
\end{align*}
Clearly $\ol{M}$ is a $\Z_2$-graded $A^{\op}$-module.
We obviously have an isomorphism $\ol{M_1} \oplus \ol{M_2} \cong \ol{M_1 \oplus M_2}$ for all $\Z_2$-graded $A$-modules $M_1,M_2$.
If $T\colon M \to N$ is graded $A$-linear we defined $\ol{T}(m) = (-1)^{|T||m|} \ol{T(m)}$ so that
\begin{align*}
    \ol{T}(a^{op} \ol{m}) &= (-1)^{|a||m|} \ol{T}(\ol{a^* m} )
    \\
    &= (-1)^{|a||m| + |T||a^* m|} \ol{T(a^* m)}
    \\
    &= (-1)^{|a| |m| + |T| |m|} \ol{a^* T(m)} 
    \\
    &= (-1)^{|a| |m| +|T| |m| + |a| |T(m)|} a^{op} \ol{T(m)} 
    \\
    &= (-1)^{|a||T|} a^{op} \ol{T}(\ol{m}).
\end{align*}
Finally, we have that
\begin{align*}
\ol{T}_1 \ol{T}_2(\ol{m}) &= (-1)^{|T_2||m|} \ol{T_1} (\ol{T_2(m)})
\\
&= (-1)^{|T_2| |m| + |T_1| |T_2(m)|} \ol{T_1 T_2(m)} 
\\
&= (-1)^{|T_2| |m| + |T_1| |T_2| + |T_1| |m| + |T_1 T_2||m|} \ol{T_1 T_2}(\ol{m})
\\
&= (-1)^{|T_1| |T_2|} \ol{T_1 T_2}(\ol{m})
\end{align*}
To finish the proof, we have to show that  the map $\mathbf{GrMod}^{ev}_A(M_1,M_2) \to \mathbf{GrMod}^{ev}_{A^{\op}}(\ol{M_1},\ol{M_2})$ is continuous for all finitely generated projective $\Z_2$-graded $A$-modules $M_1,M_2$.
Applying the comments in Remark \ref{rem:opvstau}, we can equivalently work with $A^\tau$ and $M^\tau$ instead of $A^{\op}$ and $\ol{M}$.
Given how we constructed the topologies on operator spaces, it suffices to assume that as an ungraded $A \otimes Cl_{+1}$, we have $M_1 = M_2 = (A \otimes Cl_{+1})^n$ for some $n$.
Translating this back to the graded formalism, we realize that as a graded $A$-module $A \otimes Cl_{+1}$ is two copies of the $A$-module $A$, but one with flipped grading.
If we apply $M \mapsto M^\tau$ to the $\Z_2$-graded $A$-module $A$, we clearly obtain $A^\tau$ as a module over itself.
Since $M \mapsto M^\tau$ preserves direct sums and gradings, we thus get that $((A \otimes Cl_{+1})^n)^\tau \cong (A^\tau \otimes Cl_{+1})^n$.
Because the identity map from $A$ to $A^\tau$ is a homeomorphism, so is the identification $((A \otimes Cl_{+1})^n)^\tau \cong (A^\tau \otimes Cl_{+1})^n$.
Clearly the map $\GrMod_A(A^n,A^n) \to \GrMod_A(A^n,A^n)$ given by the identity on degree-preserving maps and postcomposing by $\epsilon_{A^n}$ on degree-reversing maps is continuous.
We see that $T \mapsto \ol{T}$ is a homeomorphism $\GrMod_A(A^n,A^n) \to \GrMod_{A^{op}}((A^{op})^n,(A^{op})^n)$.
\end{proof}

\begin{remark}
Suppose $M$ is a $\Z_2$-graded $A$-module and $T \colon M \to M$ is graded $A$-linear. 
If $T$ is odd and $T^2 = \pm \id_M$, then $\overline{T}^2 = \mp \id_M$.
\end{remark}

\begin{remark} 
Suppose $A$ is not purely even.
If $M$ is a Hilbert $A$-module, then $\overline{M}$ does not seem to have an obvious Hilbert $A^{op}$-module structure.
Working with $A^\tau$ for simplicity, a short computation shows that there is no $A^\tau$-valued inner product on $M^\tau$ of the form
\[
\langle m_1^\tau, m_2^\tau \rangle = \omega(m_1,m_2) \langle m_1, m_2 \rangle^\tau
\]
for scalars $\omega(m_1,m_2)$ that only depend on $|m_1|$ and $|m_2|$.
To obtain a topology on $\overline{M}$ for $M$ finitely generated projective, we will therefore resort to embedding in a large enough $(A^{op})^n$.

\end{remark}

\begin{remark}
    The $K$-theory of a real $C^*$-algebra is typically denoted $KO_0(A)$ for real and $K_0(A)$ for complex $C^*$-algebras.
    We will denote both by $K_0(A)$ and only use $KO$ for the real $K$-theory of topological spaces.
\end{remark}

\begin{definition}
\label{def:finite-dim K-theory}
Let $A$ be a real $\Z_2$-graded algebra.
Let $\pi_0 \Mod_A$ be the set of isomorphism classes of finitely generated projective $\Z_2$-graded modules over $A$.
The \emph{Atiyah-Bott-Shapiro $K$-theory of $A$} is the quotient all modules by those that extend to $A \otimes_\R Cl_{-1}$ under the inclusion of the first tensor factor
\[
K^{ABS}_0(A) = \frac{\pi_0 \Mod_A}{\pi_0 \Mod_{A \otimes_\R Cl_{-1}} }.
\]
More generally, we define $K^{ABS}_n(A):=K^{ABS}_0(A \otimes Cl_{-n})$.
\end{definition}

It can be shown that $K^{ABS}_0(A)$ is an abelian group under $\oplus$ with the inverse of a $\Z_2$-graded module $M$ being the same representation with opposite grading. 

\begin{remark}
\label{rem:ABSKth}
Let $A$ be a real $\Z_2$-graded $C^*$-algebra.
    There is a group homomorphism
    \[
K^{ABS}_0(A^{op}) \to K_0(A) \quad [M,\epsilon] \mapsto [M, \epsilon, -\epsilon].
    \]
    It can be shown that this map is an isomorphism if $A$ is finite-dimensional or purely even.
    Note that this is rather surprising as $K^{ABS}$ did not use any information about the topology on $A$ or its modules.
    However, the map is not an isomorphism in general.
    For example, if $A = C(S^1) \otimes \C l_1$ one can show that $K^{ABS}_0(A) = 0$ and $K_0(A) = \Z$~\cite{andrektheory}, see \cite[Example 3.61]{luukmasters} for details.
\end{remark}

Given Definition \ref{def:higherkarkth}, we also have $K^{ABS}_n(A^{op}) \cong K_n(A)$ for finite-dimensional $A$.

\begin{remark}
    \label{rem:ABSconvention}
    We defined $K^{ABS}$ as in the original reference \cite{abs} by quotienting by modules that extend to $Cl_{-1}$.
    If we wanted to agree with the conventions natural from the perspective of Karoubi $K$-theory, we should have used $Cl_{+1}$ instead.
    These two conventions differ by an opposite since $M \mapsto \ol{M}$ induces an isomorphism
    \[
    \frac{\pi_0 \Mod_{A}}{\pi_0 \Mod_{A \otimes_\R Cl_{+1}} } \cong \frac{\pi_0 \Mod_{A^{op}}}{\pi_0 \Mod_{A^{op} \otimes_\R Cl_{-1}} } = K^{ABS}_0(A^{op})
    \]
    by Lemma \ref{lem:barequiv}.
    For example, we have $K_0(Cl_{+1}) \cong \Z_2$ and $K_0(Cl_{-1}) = 0$ while $K^{ABS}_0(Cl_{+1}) = 0$ and $K^{ABS}_0(Cl_{-1}) = \Z_2$, see Example \ref{ex:Karoubicomputation}.
\end{remark}

\begin{remark}
    \label{rem:Kasparovconvention}
The convention used in Definition \ref{def:finite-dim K-theory} also agrees with a Kasparov style definition of $K$-theory of $\Z_2$-graded $C^*$-algebras, where we instead have $KK_0(\R , Cl_{+n} \otimes A) \cong KK_{-n}(\R, A)$.
In general we expect that $KK_0(\R,A) \cong K_0(A^{op})$, compare \cite[Section 4.2]{gomi2017freed}.
\end{remark}

\begin{lemma}
\label{lem:morita}
    There is an isomorphism of groups
    \[
    \SPT^A \cong \SPT^{A \otimes Cl_{1,1}}.
    \]
\end{lemma}
\begin{proof}
Recall that $Cl_{1,1} \cong M_{1|1}(\R)$ is the $C^*$-algebra of endomorphisms of $\R^{1|1}$, see Remark \ref{rem:morita}.
The Morita equivalence of algebras between $A$ and $M_2(A)$ induces an equivalence of categories
    \[
    \Mod_A \to \Mod_{M_2(A)} \quad M \mapsto M \otimes \R^{2}
    \]
    of ungraded modules.
    However, $A \otimes M_{1|1}(\R)$ is not isomorphic to $M_2(A)$ as an ungraded algebra, since off-diagonal matrices anticommute with odd elements of $A$.
    
    We solve this by an ad-hoc construction
    \[
    \Mod_A \to \Mod_{A \otimes M_{1|1}(\R)} \quad M \mapsto M \oplus M,
    \]
where we give $M \oplus M$ the $A \otimes M_{1|1}(\R)$-module structure
\[
(a \otimes 1)(m_1,m_2) = (am_1, (-1)^{|a|} am_2)
\]
and matrices act in the obvious way.
The fact that this is a module can be checked directly on homogeneous elements.
 This is an equivalence of categories and by restricting to free modules, we see that it is a homeomorphism on hom-spaces.
 
    A straightforward computation shows that every $(A \otimes M_{1|1}(\R))$-skew linear endomorphism is of the form
    \[
    \begin{pmatrix}
        T & 0
        \\
        0 &  -T
    \end{pmatrix}
    \]
    for an $A$-skewlinear morphism $T \colon M \to M$.
    We in particular obtain a homeomorphism $\Pol_A \to \Pol_{A \otimes M_{1|1}(\R)}$ of topological monoids, inducing the desired group isomorphism.
\end{proof}

\begin{remark}
    The above lemma generalizes to the fact that $\SPT^A$ is Morita-invariant, also see Remark \ref{rem:morita}.
    Even though we won't prove this directly, this is known for Karoubi $K$-theory and so it will follow from the next theorem.
\end{remark}

\begin{theorem}
\label{th:IQPVK-theory}
There is an isomorphism
$\SPT^A \cong K_2(A^{op})$
\end{theorem}
\begin{proof}
Note that it suffices to show that
\[
\SPT^{A \otimes Cl_{+1}} \cong K_0(A^{op} \otimes Cl_{+1})
\]
Indeed, we can replace $A$ with $A \otimes Cl_{-1}$ and apply the lemma to deduce that \begin{align*}
\SPT^A &\cong \SPT^{A \otimes M_{1|1}(\R)} 
\\
&\cong K_0((A \otimes Cl_{-1})^{op} \otimes Cl_{+1}) = K_0(A^{op} \otimes Cl_{+2}) 
\\
&\cong K_{2}(A^{op}).
\end{align*}
We proceed to prove $\SPT^{A \otimes Cl_{+1}} \cong K_0(A^{op} \otimes Cl_{+1})$ by first providing a relation on the level of modules.
Recall that there is a one-to-one correspondence between ungraded $A \otimes Cl_{+1}$-modules $M$ and $\Z_2$-graded $A$-modules $M$, where a map $T\colon M \to M$ of ungraded $A \otimes Cl_{+1}$-modules corresponds to an even map of $\Z_2$-graded $A$-modules, while a skew-map corresponds to an odd graded-skew map.
In particular if $T \in \Diff_{A\otimes Cl_{+1}}(M)$, then $\overline{T} = T \circ \epsilon_M\colon \overline{M} \to \overline{M}$ is invertible, graded skew-$A^{op}$-linear, odd and $\ol{T}^2 = -1$.
Therefore we get that $\overline{T} \in \Grad_{A^{op} \otimes Cl_{+1}}(\overline{M})$.
This defines a homeomorphism
\[
\Diff_{A\otimes Cl_{+1}}(M) \cong \Grad_{A^{op} \otimes Cl_{+1}}(\overline{M}).
\]
It is also clear that an isomorphism of Karoubi skew-triples is mapped to an isomorphism of Karoubi triples.
The correspondence preserves sums of Karoubi triples.
We thus obtain a map
\[
\SPT^{A \otimes Cl_{+1}} \to K_0(A^{op} \otimes Cl_{+1}).
\]
Using the inverse of the bijection above we also obtain an inverse of this map.
\end{proof}

By replacing $A$ with $A \otimes Cl_{+n}$ we also obtain
$\SPT^{A \otimes Cl_{+n}} \cong K_{2-n}(A^{op})$.

\subsection{The tenfold way and \texorpdfstring{$\Z_2$}{Z2}-graded division algebras}
\label{sec: ten division algebras}

The tenfold way is typically phrased as providing a collection of ten symmetry algebras we can use to protect our system.\footnote{It is more common to use symmetry groups rather than algebras, but for our setting algebras are more convenient. We will discuss tenfold symmetry groups in Section \ref{sec: ten fermionic groups}.}
Unfortunately, the exact list of ten algebras can differ between references, also see \cite{stehouwer2022interacting}.
In this subsection, we will present one list of ten that is mathematically especially natural~\cite{geikomoore, freedhopkins, baez2020tenfold}. 
This tenfold way bears close resemblance to the original \cite{altlandzirnbauer} (the only exceptions to equality being class AI and class BDI),
but mostly deviates from \cite{schnyder2008classification, freedmoore}. 
In Section \ref{sec:tenfold}, we will compare with other lists of ten symmetry algebras in the literature.

A $\Z_2$-graded algebra is a $\Z_2$-graded division algebra if every nonzero homogeneous element is invertible.

\begin{theorem}\cite{wall1964graded, deligne1999quantum}
There are ten $\Z_2$-graded division algebras over $\R$:
\begin{align*}
&\C, \quad \C l_1, 
\\
&\R, \quad Cl_{-1}, \quad Cl_{-2}, \quad Cl_{-3}, \quad \mathbb{H}, \quad Cl_{3}, \quad Cl_{2}, \quad Cl_{1}
\end{align*}
\end{theorem}

Every $\Z_2$-graded division algebra $D$ admits the structure of a graded $C^*$-algebra.
Recall from Remark \ref{rem:ABSKth} that $K_0(D^{op}) \cong K^{ABS}_0(D)$.
It follows by direct computation that as $D$ runs through the ten $\Z_2$-graded division algebras, $K^{ABS}_0(D)$ runs through the complex and real Bott song: 
\begin{align*}
\label{song}
K^{ABS}_0(\C) &= \Z, \quad K^{ABS}_0(\C l_1) = 0, 
\\
K^{ABS}_0(\R) &= \Z, \quad K^{ABS}_0(Cl_{-1}) = \Z_2, \quad K^{ABS}_0(Cl_{-2}) = \Z_2, \quad K^{ABS}_0(Cl_{-3}) = 0, 
\\
K^{ABS}_0(\mathbb{H}) &= \Z, \quad K^{ABS}_0(Cl_{3}) = 0, \qquad K^{ABS}_0(Cl_{2}) = 0, \qquad K^{ABS}_0(Cl_{1}) = 0.
\end{align*}
This follows by a direct computation using Definition \ref{def:finite-dim K-theory}.
As a consequence we have

\begin{conclusion}
\label{con:zerodimfree}
The group 
\[
\SPT^D \cong K_{2}(D^{\op}) \cong K^{ABS}_2(D)
\]
of SPT phases in spatial dimension zero with symmetry algebra $D$ recovers the usual periodic table in the condensed matter literature in spatial dimension zero, also see Table \ref{tab:tenfold}.
\end{conclusion}

We consider a generalization to positive spatial dimension in Section \ref{sec:posdim}.

\begin{remark}
\label{rem:morita}
The $\Z_2$-graded matrix algebra $M_{p|q}(\R) = \End \R^{p|q}$ is defined to be the algebra of all linear maps $\R^{p|q} \to \R^{p|q}$ with its grading given by grading-preserving versus grading-reversing maps.
Explicitly, it is the algebra of matrices of size $p+q$ by $p+q$ in which block diagonal matrices 
\[
\begin{pmatrix}
A_1 & 0 \\
0 & A_2
\end{pmatrix}
\quad A_1 \in M_{p}(\R), A_2 \in M_q(\R)
\]
are even and block off-diagonal matrices are odd.
A fact that will be crucial for us, is that $K$-theory is \emph{Morita-invariant}.
In other words, there is an isomorphism $K_0(M_{p|q}(\R) \otimes A) \cong K_0(A)$ for all $A$ and $p,q$ such that $p+q >0$, compare Lemma \ref{lem:morita}.
\end{remark}

We now clarify the mathematical reason the tenfold way arises.
Let $A$ be a finite-dimensional semisimple $\Z_2$-graded algebra.
A $\Z_2$-graded version of the Artin-Wedderburn theorem tells us that $A$ is the direct sum of graded matrix algebras over $\Z_2$-graded division rings $D$ over $\R$.
The fact that $K$-theory is additive under direct sums together with Morita-invariance allows us to determine the $K$-theory of $A$ if we know the decomposition of $A$ into matrix algebras over $\Z_2$-graded division rings.
In other words, the ten $\Z_2$-graded division algebras are the irreducible building blocks through which we can understand the representation theory of $A$ and through that free fermion SPT phases via the $K$-theory of the symmetry algebra.
Summarizing, we have

\begin{conclusion}
\label{con:finitedimsplit}
    Let $A$ be a finite-dimensional semisimple $\Z_2$-graded algebra over the reals.
    It admits a direct sum into matrix algebras over one of the 10 $\Z_2$-graded division algebras.
    Then $\SPT^A$ is the sum over $\SPT^D \cong K_2^{ABS}(D)$ for the division algebras $D$ appearing in this direct sum.
\end{conclusion}

Conclusion \ref{con:finitedimsplit} allows for straightforward computation of $\SPT^A$ for $A$ finite-dimensional semisimple, as the $K$-theory of division algebras is provided by the Bott song.
For example, if $A = \R G$ is the real group algebra of a finite group (seen as an algebra with trivial grading) its decomposition into matrix algebras over $\R,\C$ and $\mathbb{H}$ can be found by studying the usual complex representation theory of $G$ and its Frobenius-Schur indicators.
This approach has been generalized further to certain $\Z_2$-graded groups \cite{superfrobeniusschur}.

\section{Phases protected by a group}
\label{sec:groups}

In previous sections, we considered $\Z_2$-graded $C^*$-algebras with odd elements encoding time-reversing symmetries.
In the physics literature, it is more common to present symmetries by groups.
In this section, we study phases protected by symmetry groups and the relationship with phases protected by symmetry algebras.
This in particular allows us to define phases in positive spatial dimension in Section \ref{sec:posdim}.

\subsection{Fermionic groups}

We will encode symmetries of fermionic systems as follows, see \cite[Chapter 3]{luukthesis} for a more comprehensive discussion.

\begin{definition}
\label{def: fermionic group}
A \emph{fermionic group} is a second countable\footnote{This is a technical assumption which allows us to use the existing theory of twisted group $C^*$-algebras.} Hausdorff topological group $G$ together with a central element $(-1)^F$ of square one and a continuous homomorphism $\theta\colon G \to \Z_2^T = \{0,1\}$ such that $\theta((-1)^F) = 0$. 
\end{definition}

Here the group element $(-1)^F$ plays the role of the fermion parity operator.
We again record the possibility of time-reversing symmetries that act anti-unitarily on the Hilbert space through a $\Z_2$-grading, so for time-reversing symmetries $T$ we have $\theta(T) = 1$.
We require fermion parity to be time-preserving.
The underlying \emph{bosonic group} of a fermionic group $(G,(-1)^F,\theta)$ is $G_b := G/ (-1)^F $.
We say that a fermionic group is \emph{bosonic} if $(-1)^F = 1$.

\begin{definition}
\label{def:oppfermgrp} 
Define the \emph{opposite} $G^{op}$ of a fermionic group $G$ as the same set but with multiplication $*$ given by
\[
g_1 * g_2 := 
\begin{cases}
(-1)^F g_1 g_2  &\text{ both odd,}
\\
g_1 g_2 &\text{otherwise.}
\end{cases}
\]
\end{definition}
This is again a fermionic group.
The logic behind this definition is that the obvious definition of `opposite multiplication with Koszul sign' is 
\[
g_1 *' g_2 := 
\begin{cases}
(-1)^F g_2 g_1 & \text{ both odd,}
\\
g_2 g_1 &\text{otherwise,}
\end{cases}
\]
and we can make this multiplication isomorphic to the simpler one in the definition above by applying the anti-homomorphism $g \mapsto g^{-1}$.

Note that unlike for bosonic or purely even groups, $G \ncong G^{op}$. 
For example, give the Pin groups $\Pin^{\pm}_d$ the fermionic group structure in which $(-1)^F = -1 \in \Pin^\pm_d$ is the nontrivial element in the kernel of $\Pin^\pm_d \to O_d$ and $\theta$ is the determinant $\Pin^+_d \to O_d \xrightarrow{\det} \Z_2$.
Then $(\Pin^+_d)^{op} = \Pin^-_d$.

\subsection{Free SPT phases for finite symmetries}
\label{sec:finiteSPT}

Before discussing more complicated cases in the rest of this section, we consider for simplicity the case where $G$ is a finite fermionic group.
Our main goal is to answer the question: which algebra $A$ can we assign to a fermionic group $G$ so that the $K$-theory of $A$ classifies SPT phases protected by $G$?
The answer we propose will be a kind of twisted group algebra of $G$ we call the fermionic group algebra.

\begin{definition}
The \emph{fermionic group algebra} is the real $\Z_2$-graded algebra
\[
C^*_f(G):= \frac{\R[G]}{((-1)^F + 1)},
\]
where $\R[G]$ is the real group algebra, $((-1)^F + 1)$ is the ideal enforcing $(-1)^F = -1$ and the grading is given by $\theta$.
\end{definition}

We make $C^*_f(G)$ into a $C^*$-algebra by $g^* := g^{-1}$.
Note that this is well-defined since $((-1)^F g)^{-1} = (-1)^F g^{-1}$ in $\R[G]$.
Note that $C^*_f(G^{op}) \cong C^*_f(G)^{op}$, also see Remark \ref{rem:opvstau}.

\begin{example}[class DIII]
\label{example: class DIII}
The fermionic group denoted $\Z_4^{TF}$ is the group $\Z_4$ generated by an element denoted $T$ (time-reversal).
Both $\theta$ and $(-1)^F$ are nontrivial, so that $\theta(T) = 1$ is time-reversing and $T^2 = (-1)^F$.
This is symmetry class DIII in the tenfold way, see Section \ref{sec: ten fermionic groups}.
The fermionic group algebra is isomorphic to $Cl_{-1}$.
\end{example}

\begin{example}
\label{ex:bosonicgroupalg}
    If $G$ is bosonic, then $C^*_f(G) = 0$.
\end{example}

A finitely generated projective module over $C^*_f(G)$ is a finite-dimensional real representation $(M,R)$ of $G$ such that $R((-1)^F) = -1$, compare Remark \ref{rem:-1vs-1F}.
A Hilbert $C^*_f(G)$-module is not immediately the same as a unitary representation of $G$ such that $R((-1)^F) = -1$ because the inner product is valued in the group algebra.
However, we can use the trace 
\[
tr(g) = 
\begin{cases}
 1 &     g=1
 \\
 -1 & g = (-1)^F
 \\
 0 & \text{else}
\end{cases}
\]
on $C^*_f(G)$ to define a Hilbert space in this finite-dimensional setting, also see Remark \ref{rem:repsvsmods}.
Since this makes $C^*_f(G)$ into an $H^*$-algebra, there is also a converse construction: if $\mathcal{H}$ is a unitary representation of $G$ such that $R((-1)^F) = -1$, then we can define a Hilbert $C^*_f(G)$-module by 
\[
\langle m_1, m_2 \rangle = act_{m_1}^\dagger act_{m_2}(1),
\]
where $act_m$ for $m \in M$ denotes the linear map $C^*_f(G) \to M$ given by $g \mapsto gm$.
We refer to \cite[Construction 3.4.4]{UQSL} for details.

As a special case of Definition \ref{def:SPT phase} we obtain:

\begin{definition}
\label{def:zerodimSPT}
The group of SPT phases protected by the finite fermionic symmetry $(G,\theta,(-1)^F)$ in spatial dimension zero is the group of SPT phases protected by the symmetry algebra $C^*_f(G)$.
\end{definition}

The reason that we refer to this group as the phases in the spatial dimension zero, is that from a condensed matter perspective we think of the Majorana space as the finite-dimensional vector space localized to a single atomic site.

\begin{proposition}
The group of SPT phases protected by the $G$ in spatial dimension zero is
\[
K^{ABS}_2(C^*_f(G)),
\]
where we used the definition of $K$-theory of finite-dimensional $\Z_2$-graded algebras as in Definition \ref{def:finite-dim K-theory}.
\end{proposition}
\begin{proof}
Since $C^*_f(G)$ is semisimple~\cite[Example 3.5.1]{GK14}, we can replace $K$ by $K^{ABS}$, see Remark \ref{rem:ABSKth}.
This is now a consequence of Theorem \ref{th:IQPVK-theory}.
\end{proof}

\subsection{The tenfold way using fermionic groups}
\label{sec: ten fermionic groups}

It can be desirable to give the tenfold way in terms of symmetry groups as opposed to symmetry algebras as we did in Section \ref{sec: ten division algebras}.
We have already discussed how to obtain a $\Z_2$-graded $C^*$-algebra from a fermionic group.
Luckily, there is also a canonical way to pass from $\Z_2$-graded division algebras to fermionic groups, giving us a special family of ten fermionic groups as follows.
Note that any $\Z_2$-graded division algebra $D$ has the structure of a real $C^*$-algebra. Define the sphere of $D$ to be the group
\[
S(D) := \{d \in D^\times : \|d\| = 1\},
\]
where $D^\times \subseteq D$ is the group of nonzero homogeneneous (hence invertible) elements.
We make $S(D)$ into a (potentially infinite) fermionic group by using $(-1)^F = [-1] \in S(D)$ and $\theta[d] = |d|$ is the grading of $D$.
We obtain a list of ten fermionic groups which can be interpreted as the internal symmetry groups of the tenfold way.
In mathematical notation these groups are
\begin{align*}
 &\Spin^c(1), \Pin^c(1)
    \\
&\Spin(1), \Pin^-(1), \Pin^-(2), \Pin^-(3)
    \\
    &\operatorname{Sp}(1), \Pin^+(3), \Pin^+(2), \Pin^+(1).
\end{align*}
In notation closer to the physics we can write them as
\begin{align*}
&U(1)_Q, U(1)_Q \times \Z_2^T, 
\\
&\Z_2^F, \Z_4^{TF}, \frac{U(1)_Q \rtimes \Z_4^{TF}}{\Z_2^F}, SU(2) \times \Z_2^T, 
\\
&SU(2), \frac{SU(2) \times \Z_4^{TF}}{\Z_2^F}, U(1)_Q \rtimes \Z_2^T, \Z_2^F \times \Z_2^T,
\end{align*}
where we abused the superscript $T$ to denote a time-reversing symmetry no matter what commutation properties it satisfies with respect to the other symmetries. 
The rest of this section is devoted to interpreting several of these groups physically.

\begin{example}[class AIII]
\label{ex:class AIII}
The fermionic group $\Pin^c(1) = U(1)_Q \times \Z_2^T$ consists of charge-symmetry and a time-reversing symmetry $T$ satisfying $e^{iQt} T = T e^{iQt}$.
Since $T$ is time-reversing, it anticommutes with $i$.
Therefore $T$ anticommutes with charge, and so is a particle-hole symmetry, see Example \ref{ex: sublattice}.
This is the internal symmetry group of class AIII topological insulators, of which the Su-Schrieffer-Heeger model is an Example, see \cite[Section 4.3]{zirnbauerparticlehole} for further discussion.
We will study the effect of interactions on the SSH model in future work~\cite{camluuk}.
\end{example}

\begin{example}[Class CII]
\label{ex:class CII}
We can realize the symmetry group $\Pin^-(3) \cong SU(2) \times \Z_2^T$ by letting the $\Z_2^T$ be generated by an anti-unitary particle-hole symmetry $K$ such as a sublattice symmetry.
To recover the whole symmetry group, we include charge $Q$ and a time-reversal $T$ with $T^2 = (-1)^F$.
Then $iQ$ and $TK$ are unitary symmetries that anticommute and square to $(-1)^F$ and so generate a quaternion algebra yielding the $SU(2)$.
A translation back to Clifford algebra language $\Pin^-(3) \subseteq Cl_{-3}$ is 
\[
\gamma_1 \leftrightarrow T, \quad \gamma_1 \gamma_2 \leftrightarrow iQ, \quad \gamma_1 \gamma_2 \gamma_3 \leftrightarrow K,
\]
where $\gamma_1, \gamma_2, \gamma_3$ are generators of $Cl_{-3}$ with square $-1$.
\end{example}

\begin{example}
\label{ex:groupvsalgebratenfold}
Let $D$ be a $\Z_2$-graded division algebra with $D_{ev} = \R$, so $D = \R, Cl_{+1}$ or $C l_{-1}$, corresponding to classes D, BDI and DIII respectively.
Then $S(D)$ is a finite fermionic group and $C^*_f(S(D)) = D$ and so $S(D)$-protected SPT phases agree with $D$-protected SPT phases.
As explained in Section \ref{sec: ten division algebras}, the $K$-theory groups 
\[
 \SPT^{C^*_f(S(D))} \cong \SPT^D \cong K^{ABS}_2(D)
\]
agree with the usual classification for free fermion topological phases in spatial dimension zero.
\end{example}

Unfortunately, the operations $D \mapsto S(D)$ and $G \mapsto C^*_f(G)$ are not inverses to each other for other choices of $D$.
Namely, $S(D)$ is not even finite.
We will discuss the fermionic group algebra of infinite fermionic groups in the next section.
However, we postpone a deeper discussion of the difference between $C^*_f(S(D))$ and $D$ to Sections \ref{sec: charge 1} and \ref{sec: spin1/2}.

\subsection{Infinite groups and the fermionic group \texorpdfstring{$C^*$}{C*}-algebra}
\label{sec:infinitegroups}

The main task of this section is to answer the question: which $C^*$-algebra $A$ can we assign to a fermionic group $G$ so that the $K$-theory of $A$ classifies SPT phases protected by $G$?
The proposed answer will be a suitable generalization of the fermionic group algebra of $G$ (Definition \ref{def:twisted gp alg}), which in general will not be unital.
In case $G$ is a compact group, we explain how to use the Peter Weyl theorem to do computations (Theorem \ref{th: twisted peterweyl}).

To deal with infinite groups, we have to replace the group algebra $\R[G]$ by the real group $C^*$-algebra\footnote{Since we will only consider amenable groups, we will not distinguish between the group $C^*$-algebra and the reduced group $C^*$-algebra. In fact, the groups we have in mind for our condensed matter applications are extensions of compact Lie groups by abelian groups.} $C^*(G)$.
This is a $C^*$-algebra of which its representations are in one-to-one correspondence with unitary representations of $G$ on Hilbert spaces.
However, when defining the fermionic group $C^*$-algebra for infinite groups we face the problem that elements of $G$ are in general not elements of $C^*(G)$, but only of the multiplier algebra of $C^*(G)$ (Definition \ref{def:multiplieralg}).
We hope readers not familiar with group $C^*$-algebras will gain some intuition from the following example.

\begin{example}[class A]
\label{example: class A}
Consider the real division algebra $D = \C$ which has associated sphere $G := S(D) = U(1)$ as a fermionic symmetry group.
The real group $C^*$-algebra can be defined as a certain closure of the algebra of continuous periodic functions $f\colon [0,2\pi] \to \R$.
Taking Fourier series yields an isomorphism between the complex group $C^*$-algebra and the $C^*$-algebra $C_0(\Z, \C)$ of functions $\Z \to \C$ vanishing at infinity.
The function $e^{inx}/2\pi $ is mapped to the function $\delta_{n,\bullet}\colon \Z \to \C$. 

Note that this $C^*$-algebra is not unital because a nontrivial constant map $\Z \to \C$ does not vanish at infinity. 
More generally, $U(1)$ cannot be made into a subset of the group $C^*$-algebra in the expected way because $z \in U(1)$ would correspond to the delta function at $z$, but the group $C^*$-algebra is not large enough to contain the delta functions.
Namely, the delta function at $a \in [0,2\pi]$ corresponds to the map $n \mapsto e^{ina}$, which does not vanish at infinity.
In particular, the delta function at $-1 \in U(1)$ corresponding to the element $(-1)^F$ is mapped to $n \mapsto (-1)^n$.
However, $U(1)$ is inside the multiplier algebra of $C^*(U(1))$, which is isomorphic to the $C^*$-algebra of bounded functions on $\Z$, i.e. the $C^*$-algebra of all functions on the Stone-\v{C}ech compactification of $\Z$.

Because $\overline{e^{inx}} = e^{-inx}$, the real group $C^*$-algebra of $U(1)$ is isomorphic to the real $C^*$-algebra of functions $\alpha\colon \Z \to \C$ vanishing at infinity with the extra property that $\alpha(n) = \overline{\alpha(-n)}$.
In other words, since we are now interested in real representations of $U(1)$ and the nontrivial irreducible representations are of complex type, we have to pair charge $n$ and charge $-n$ into a single Weyl-type representation
\begin{align}
\label{eq: charge n rep}
\rho_n\colon U(1) \to O(2), \quad \rho_n(e^{i\alpha}) = 
\begin{pmatrix}
\cos n \alpha & - \sin n \alpha \\
\sin n \alpha & \cos n\alpha
\end{pmatrix}.
\end{align}
For example, the functions $\Z \to \C$ that map $m$ to
\[
\frac{\delta_{n,m} + \delta_{-n,m}}{2}, \quad \frac{\delta_{n,m} + \delta_{-n,m}}{2i} \quad n >0
\]
correspond with the real-valued $2\pi$-periodic functions $\sin(nx), \cos(nx)$.

The idea of the spin-charge relation is to identify the group element $-1 \in U(1)$ with the element $-1 \in \R$ in the ground field.
However, unlike for finite groups the expression
\[
\frac{C^*(G)}{(1+(-1)^F)}
\]
is not defined since $-1 \in U(1) \notin C^*(U(1))$.
Intuitively, $(-1)^F = -1$ would only leave irreducible (complex) representations of odd charge $n \in \Z$.
Therefore we might cut out the even charge representations as elements of the Fourier transformed space $\Z$ and define the real $C^*$-algebra
\[
C^*_f(U(1)) := \{f\colon \Z \to \C: f(-n) = \overline{f(n)}, f(2n) = 0\}
\]
\end{example}

The above argument generalizes to any compact group because a version the Peter Weyl theorem, which can be formulated as follows.
The direct sum of a family $\{A_i: i \in \mathbb{N}\}$ of $C^*$-algebras is defined as 
\[
\bigoplus_{n \in \mathbb{N}} A_i :=  \colim_n \left( \bigoplus_{i = 1}^n A_i \right)
\]
in the category of (not necessarily unital) $C^*$-algebras.
An explicit model is given by the subalgebra of $\prod_i A_i$ given by collections $(a_i)$ for which $\lim_{i \to \infty} \| a_i \| < \infty$, see \cite[Appendix L]{weggeolsen}.

\begin{theorem}
\label{th:PeterWeyl}
If $G$ is a compact group, there is an isomorphism of real $C^*$-algebras
\[
C^*(G) \cong \bigoplus_{\rho \text{ f.d. irrep}} M_{\dim_{D(\rho)} (\rho)}(\End_{D(\rho)} \rho)
\]
where $D(\rho) = \R, \C, \mathbb{H}$ is the division algebra of $G$-equivariant endomorphisms of $\rho$, $\dim_{D(\rho)} \rho$ is the rank of $\rho$ as a $D(\rho)$-module and $\End_{D(\rho)} \rho$ are the $D(\rho)$-linear endomorphisms of $\rho$. 
\end{theorem}

\begin{example}
We have that
    \[
    C^*(U(1)) \cong \R \oplus \bigoplus_{i \in \mathbb{N}} \C \cong \{\alpha \in  C_0(\Z): \alpha(-n) = \overline{\alpha(n)}\},
    \]
    recovering Example \ref{example: class A}.
\end{example}

\begin{corollary}
\label{cor:split}
    If $G$ is a compact group, there is an isomorphism of groups
    \[
    \SPT^{C^*(G)} \cong \bigoplus_{\rho \text{ f.d. irrep}} K_{2}^{ABS}(D(\rho)).
    \]
\end{corollary}
\begin{proof}
    Note that $K$-theory is additive for infinite direct sums (which follows from compatibility with directed colimits \cite[Proposition 6.2.9]{weggeolsen}).
    The result follows from Theorem \ref{th:PeterWeyl} and Morita invariance of $K$-theory.
\end{proof}

The above corollary can be used to compute SPT phases protected by a compact group from its representation theory in a similar fashion to computing the $K$-theory of finite-dimensional semisimple algebras as in Conclusion \ref{con:finitedimsplit}.
Indeed, we can use Frobenius–Schur theory to determine which of the three division algebras $D(\rho)$ is.

In SPT phases, one is often also interested in groups that include translation symmetries of the spatial lattice of atoms $G = \Z^r$, which are not compact and so don't have discrete representation theory.
For this we will need other techniques to compute the group of SPT phases, see Section \ref{sec:posdim}.

We will define the fermionic group algebra as a twisted version of the group $C^*$-algebra of $G_b$, see \cite{echterhoff2008k} for an exposition of twisted group algebras in a $K$-theoretical setting.
In order to get a short exact sequence
\[
1 \to \Z_2^F \to G \to G_b \to 1
\]
of topological groups, we will assume $(-1)^F \neq 1$ (see Example \ref{ex:bosonicgroupalg}).
By analogy to how one can define twisted group algebras for discrete groups, we would like to pick a section $s\colon G_b \to G$ and define a 2-cocycle $\nu\colon G_b \times G_b \to \Z_2^F$ as $s(g) s(h) = (-1)^{\nu(g,h)} s(gh)$.
However, if the double cover $G \to G_b$ is topologically nontrivial, no continuous section exists.
Fortunately, a (Borel) measurable section exists since $G$ is assumed to be second countable~\cite{rieffel1966extensions}.\footnote{For Lie groups, locally continuous sections always exist which can be easier to work with for concrete computations. }
This is sufficient for defining the relevant twisted semidirect products. 
So assume we are given a measurable section $s\colon G_b \to G$ of the double cover such that $s(1) = 1$. 
This makes $\nu\colon G_b \times G_b \to \Z_2^F$ defined as $s(g) s(h) = (-1)^{\nu(g,h)} s(gh)$ measurable.
We obtain a concrete representative for the cohomology class in Moore's measurable group cohomology $H^2_m(G_b,\Z_2^F)$ corresponding to the extension.
A choice of such a representative cocycle, then allows us to define the fermionic group algebra as the $\nu$-twisted real group $C^*$-algebra of $G_b$.

\begin{definition}
\label{def:twisted gp alg}
Let $\nu$ be the twist coming from a cocycle in  $Z^2_m(G_b,\Z_2^F)$ representing the extension
\[
1 \to \Z_2^F \to G \to G_b \to 1
\] 
under the inclusion $\Z_2^F \subseteq \R$ by $\pm 1$.
The $\nu$\emph{-twisted fermionic group $C^*$-algebra} is the twisted semidirect product
\[
C^*_f(G) := \R \rtimes_\nu G_b,
\]
where we take the twisted $C^*$-dynamical system $(\nu,\alpha)$ of $G$ on $\R$ to have trivial action $\alpha$ of $G_b$ on $\R$.
\end{definition}

We refer to Appendix \ref{sec:twistedsemidirect} for relevant details, in particular the definitions of twisted $C^*$-dynamical systems and twisted semidirect product $C^*$-algebras.
The main idea is that the fermionic group algebra is the universal $C^*$-algebra of which real representations are suitably continuous $\nu$-twisted representations of $G$.
The homomorphism $\theta$ defines a $\Z_2$-grading on $\R \rtimes_\nu G_b$ by giving the involution which is $-1$ on the odd part.

\begin{example}
    If $G = \Z_2^F \times G_b$, then $C^*_f(G) \cong C^*(G_b)$ as $C^*$-algebras.
\end{example}

\begin{remark}
    It follows from Proposition \ref{prop: fermionic group alg} that for a compact fermionic group the fermionic group algebra is independent of the choice of the choice of section.
    We expect that this is true more generally, but we won't prove it here.
\end{remark}

\begin{remark}
\label{rem:segalcoh}
There is a whole zoo of possible cohomology theories for topological groups, some being better behaved on a larger class of groups, others being closer to a classical definition of group cohomology using group cocycles.
See \cite{wagemann2015cocycle} for a review of this topic.
One of the most general and best behaved cohomology theories of topological groups is Segal cohomology~\cite{segalgroupcohomology}, see \cite[Section 4.1]{schommerpries2group} for a brief review.
A modern alternative is the cohomology of condensed groups \cite{scholzecondensed}.
Since we will only be interested in the case where $G_b$ is a Lie group (possibly noncompact with countably infinitely many components), measurable cohomology is isomorphic to Segal cohomology which classifies extensions~\cite[Corollary 4.8, Remark 4.13]{wagemann2015cocycle}.
Since $\Z_2^F$ is a discrete, $G_b$-module Segal cohomology with values in it is isomorphic to $H^n(BG_b, \Z_2)$.
\end{remark}

Looking at Theorem \ref{th:PeterWeyl}, it would be reasonable to expect that for compact groups the fermionic group $C^*$-algebra is the subalgebra of the group algebra given by only allowing representations in which $(-1)^F$ acts as $-1$.
In the appendix, we will prove is indeed true, which will be useful for $K$-theory computations:

\begin{theorem}
\label{th: twisted peterweyl}
Let $(G,(-1)^F, \theta)$ be a compact fermionic group.
The fermionic group $C^*$-algebra is isomorphic to the real $C^*$-algebra
\[
\bigoplus_{\substack{\text{f.d. irrep }\rho \\\text{s.t. }\rho((-1)^F) = -1}}  M_{\dim_{D(\rho)} (\rho)}(\End_{D(\rho)} \rho)
\]
which is $\Z_2$-graded by requiring $\rho(g)$ to be even/odd depending on whether $\theta(g) = 0$ or $\theta(g) = 1$.
\end{theorem}

\begin{remark}
\label{rem:compute Kth}
If $G$ is a compact fermionic group with trivial $\theta$, the above theorem readily allows one to compute its $K$-theory as
\begin{align*}
    K_n(C^*_f(G)) = \bigoplus_{\substack{\text{f.d. irrep }\rho \\\text{s.t. }\rho((-1)^F) = -1}}  K_n(D(\rho)),
\end{align*}
in a similar fashion to Corollary \ref{cor:split}.
So knowledge of the irreducible representations of $G$, in which one $(-1)^F$ acts as $-1$, and their type (real/complex/quaternionic) determines the $K$-theory.
\end{remark}

\begin{remark}
\label{rem:graded compute Kth}
If $\theta$ is nontrivial, then the above direct sum of real $C^*$-algebras is in general not a direct sum of $\Z_2$-graded $C^*$-algebras.

For example, consider the internal symmetry group $G = \Pin^+_1$ of class DIII (Example \ref{example: class DIII}).
It has two irreducible (real) representations for which $(-1)^F$ acts as $-1$ depending on whether $T$ acts as $\pm 1$ on the one-dimensional real line.
Hence $C^*_f(G) = \R \oplus \R$ as a real $C^*$-algebra.
However, since $T$ is odd, the element $(1,-1) \in \R \oplus \R$ is odd and so $C^*_f(G) = Cl_{+1}$ as a $\Z_2$-graded real $C^*$-algebra.

Therefore in practice we have to see how different ungraded irreducible representations fit together into $\Z_2$-graded irreducible representations to compute the $K$-theory.
It would be useful for computations to have a systematic way to find the decomposition of the fermionic group algebra into matrix algebras over the ten $\Z_2$-graded division algebras over $\R$, compare \cite{superfrobeniusschur}.
\end{remark}

We can now generalize Definition \ref{def:zerodimSPT} to infinite fermionic groups.

\begin{definition}
\label{def:GSPT}
The group of SPT phases protected by the fermionic group $(G,\theta,(-1)^F)$ is $\SPT^{C^*_f(G)}$.
\end{definition}

We obtain as an immediate consequence of Theorem \ref{th:IQPVK-theory}:

\begin{proposition}
The group of SPT phases protected by the $G$ is
\[
K_2(C^*_f(G)^{\op}).
\]
\end{proposition}

\begin{remark}
For a non-unital $C^*$-algebra $A$, the $K$-theory of $A$ as defined in \ref{rem:nonunital} is in general very different from the $K$-theory of the multiplier algebra of $A$.
For example, we can compute
\begin{align*}
K_0(M(C_\C^*(U(1)))) &\cong K_0(M(C^0(\Z,\C))) \cong K_0(C_b(\Z,\C)) 
\\
&\cong \{f \colon \Z \to \Z: f \text{ is bounded}\},
\end{align*}
where the subscript $\C$ means we are taking the complex group $C^*$-algebra for simplicity.
Note that 
\[
K_0(C_\C^*(U(1))) \subseteq K_0(M(C_\C^*(U(1))))
\]
is the strict subset of eventually zero sequences.
In other words, elements of this $K$-theory group are potentially infinite-dimensional representations in which every irreducible representation has finite multiplicity.
This group is too big for practical purposes to be used to define free fermion phases. 
For example, this group of $U(1)$-protected free SPT phases will then not map to the group of interacting phases, see Remark \ref{rem:continuousG}.
\end{remark}

\subsection{Free SPT phases in positive spatial dimension}
\label{sec:posdim}

Until now, we imagined the fermionic group $G$ to be an on-site symmetry, so that Definition \ref{def:GSPT} would only define SPT phases in spatial dimension $0$.
To upgrade the current situation to higher dimensions, a high-energy physicist would resort to continuum quantum field theory on spacetime.
Condensed matter theory deals with spatial dimension $d > 0$ by putting the system on a $d$-dimensional lattice of atoms in Euclidean space.
Suppose for now the fermionic symmetry group we considered in the past few sections is $\Z^{d} \times \Z_2^F$, which we think of as a collection of discrete spatial translations.\footnote{In other words, we are invoking a crystalline equivalence principle.}

For understanding the the applications in condensed matter of finitely generated and projective modules over the group $C^*$-algebra $C^*(\Z^{d}) \cong C^*(\Z^d \times \Z_2^F)$, we use a Fourier transform isomorphism to pass to momentum space.

\begin{definition}
The \emph{Brillouin zone} $BZ(d)$ is the torus obtained by taking the $d$-dimensional cube $[-\pi,\pi]^{d}$ in the dual of real space and identifying opposite edges.    
\end{definition}

Mathematically, it is more natural to view the Brillouin zone as the Pontryagin dual $\Hom(\Z^{d},U(1))$ of the lattice $\Z^{d}$ in real space.
These spaces are homeomorphic via
\[
[-\pi, \pi]^d \to \Hom(\Z^d,U(1)) \quad k \mapsto (n\mapsto e^{ikn}).
\]
Then the Fourier transform gives an isomorphism 
\[
C^*(\Z^{d}) \otimes_\R \C \cong C(BZ(d),\C)
\]
of complex $C^*$-algebras.
Under this isomorphism, finitely generated projective modules $M$ over $C^*(\Z^{d}) \otimes_\R \C$ correspond to complex vector bundles $V \to BZ(d)$ via $M = \Gamma_{cts}(BZ(d);V)$ by the Serre-Swan theorem. 
Moreover, translation invariant (i.e. $C^*(\Z^d)$-linear) operators $H \colon M \to M$ correspond to vector bundle maps.
Connecting with the language of band theory, this means that a Bloch Hamiltonian $H$ gives a collection of linear maps $H_k \colon V_k \to V_k$ on the Bloch waves of momentum $k$.\footnote{Bloch bundles are typically infinite-dimensional Hilbert bundles with unbounded Hamiltonians, but we will not go into this subtlety. 
See \cite[Section 10.3]{thiangk-theoretic} for a detailed discussion.}
We see that physically the assumption that our free fermion system has symmetry algebra $C^*(\Z^d) \otimes_\R \C$ is analogous to a tight binding model.
By our discussion in Section \ref{sec:the group of free SPT}, a pair of Bloch Hamiltonians defines a class in $\SPT^{C^*(\Z^d)\otimes_\R \C}$.

\begin{remark}
    By the universal property of $C^*(\Z^d)$, a representation of $C^*(\Z^d)$ is equivalent to an orthogonal representation of $\Z^d$ on a real Hilbert space.
    However, a representation of a $C^*$-algebra is in general different from a Hilbert module.
    For $C^*(\Z^d)$, there is an associated representation to a Hilbert $C^*(\Z^d)$-module $M$.
    In case $M = \Gamma_{cts}(BZ(d);V)$ for a vector bundle $V \to BZ(d)$ with a metric, this representation is given by $L^2$-sections of $V$, see Remark \ref{rem:repsvsmods}.
    Using the Fourier transform, we can get back to a spatial description such as Example \ref{ex: sublattice}, where the Hilbert space is $\ell^2(\Z^d)$.
\end{remark}

More generally, we could take the symmetry group in the past sections to be of the form $\Z^{d} \times G$ where $G$ is any fermionic group.
We could also consider continuous translation symmetries with the goal of defining SPT phases in the continuum.

\begin{definition}
    \label{def:posdimSPTs} 
The group of \emph{free continuum SPT phases} in spatial dimension $d$ protected by the symmetry algebra $A$ is $\SPT^{C^*(\R^d) \otimes A}$.
The group of \emph{free lattice SPT phases} in spatial dimension $d$ protected by the symmetry algebra $A$ is $\SPT^{C^*(\Z^d) \otimes A}$.
\end{definition}

In the above definition, we give $C^*(\Z^d)$ and $C^*(\R^d)$ the trivial grading.

We can compute the group of free SPT phases in positive dimensions using techniques mentioned above.
One caveat is that Fourier transform involves multiplication by $e^{ikx}$ and so does not immediately apply to the real $C^*$-algebra $C^*(\Z^{d})$.
To upgrade $C^*(\Z^{d}) \otimes_\R \C \cong C(BZ(d),\C)$ to an isomorphism of real $C^*$-algebras, we have to give complex-valued continuous functions on the Brillouin zone the real structure given by mapping $f(k)$ to $\overline{f(-k)}$.
In other words, we take the real structure $\tau$ on the space $BZ(d)$ to be $k \mapsto -k$.
There is an isomorphism of $C^*(\Z^d)$ with the real $C^*$-algebra
\[
C(BZ(d))_\tau := \{f \in C(BZ(d),\C): f(-k) = \overline{f(k)}\},
\]
also see the discussion in \cite[Section 9.1]{thiangk-theoretic}.
In particular modules over the real group $C^*$-algebra $C^*(\Z^{d})$ correspond to complex vector bundles $E \to BZ(d)$ equipped with an antilinear vector bundle automorphism $j\colon E \to E$ such that $j^2 = 1$ covering the map $k \mapsto -k$ on $BZ(d)$.
A sphere has simpler algebraic topology than a torus and so continuum SPT phases are often easier to compute:

\begin{proposition}
\label{prop:posdim}
    The group of free continuum SPT phases protected by $A$ in spatial dimension $d$ is given by
    \[
    \SPT_d^A \cong K_{2-d}(A^{op}).
    \]
\end{proposition}
\begin{proof}
Similarly to the lattice case, we can use the Fourier transformation to get that 
\[
C^*(\R^{d}) \otimes \C \cong C_0(\R^{d},\C)
\]
is given by continuous functions on the momentum space $\R^d$ vanishing at infinity.
Here similarly to what happened for discrete symmetries we have to give the right side the real structure $\tau \colon k \mapsto -k$.
We thus compute
    \[
    \SPT_d^A = \SPT^{C^*(\R^d) \otimes A} = \SPT^{C_0(\R^d)_\tau \otimes A} = K_2(C_0(\R^d)_\tau \otimes A)
    \]
    It follows by the periodicity theorem of Real $K$-theory~\cite{MR206940, schroder1993k} that this group is $K_{2-d}(A^{\op})$.
\end{proof}

The group of lattice SPT phases protected by $A$ is instead
\[
\bigoplus_{i = 0}^d K_{2-i}(A^{op})^{\oplus \genfrac(){0pt}{}{d}{i}},
\]
also see \cite[Theorem 11.8]{freedmoore} and \cite{weak}.

    Let $D$ be one of the ten real $\Z_2$-graded division algebras.
Generalizing Conclusion \ref{con:zerodimfree}, we see 

\begin{conclusion}
    The group of $d$-dimensional continuum SPT phases protected by $D$ is given by $K^{ABS}_{2-d}(D)$.
This recovers the usual tenfold way table in the condensed matter literature in all dimensions and symmetry classes, also see Table \ref{tab:tenfold}.
\end{conclusion}


\begin{remark}
It would be wonderful to have a direct comparison between SPT phases on a lattice and in the continuum by some continuum limit procedure.
It is very tempting to try to define the continuum limit of a free lattice SPT phase as its image under a map
\[
K_{2}(A^{\op} \otimes C^*(\Z^{d})) \to K_{2}(A^{\op} \otimes C^*(\R^{d}))
\]
induced by the inclusion $\Z^{d} \subseteq \R^{d}$. 
Then we would also be able to define a lattice SPT phase to be \emph{weak} if it is in the kernel of this map, otherwise it would be called \emph{strong}.
However, a group homomorphism $f \colon G \to H$ only gives a $*$-homomorphism $C^*(G) \to M(C^*(H))$.
If $G$ and $H$ are abelian we can use Gelfand duality to see that this homomorphism lands in $C^*(H)$ exactly when $\hat{f} \colon \hat{H} \to \hat{G}$ is proper \cite{15093}.
Since the quotient map $\R^d \to \R^d/\Z^d$ is not proper, this approach does not seem to work. 
It might still be the case that there is a suitable Kasparov $(\Z^{d}, \R^{d})$-bimodule which will induce this map.

In practice, one usually compares weak and strong topological phases by pulling back along an appropriate collapse map $BZ(d) \to S^{d}$ on the level of momentum space.
Namely one can quotient out the faces of the hypercube $BZ(d)$ to get a sphere.
Here we think of $S^{d}$ as the one-point compactification of the continuum momentum space, which we give the real structure induced by $\R^d$.
This induces a map
\[
 K_2(A^{op} \otimes C(S^d)) \to K_{2}(A^{op} \otimes C(BZ(d))).
\]
by pullback of vector bundles and $C(S^d)$.
Now we can use that $K_2(A^{op} \otimes C(S^d)) \cong K_2(\R) \oplus K_2(A^{op} \otimes C^*(\R^d))$ to relate to continuum SPT phases since $C(S^d)$ is the unitization of $C_0(\R^d)$.
Moreover, using the fact that the torus is stably equivalent to a wedge of spheres: 
\[
BZ(d) \sim S^{d} \wedge S^{d-1} \wedge \dots 
\]
we get a noncanonical splitting of this map, see \cite{weak} for more on this approach.
\end{remark}

\begin{remark}
A different approach to getting rid of the weak invariants is to replace the group $C^*$-algebra of $\Z^n$ by the (non-uniform) Roe algebra of $\Z^n$~\cite{ewert2019coarse}.
The Roe algebra of a metric space is a coarse invariant and so the Roe algebra of $\Z^n$ is isomorphic to the Roe algebra of $\R^n$, which has the same $K$-theory as the group $C^*$-algebra of $\R^n$.
\end{remark}

\begin{remark}
\label{rem:crystalline}
Definition \ref{def:posdimSPTs} generalizes straightforwardly to crystalline topological phases.
Indeed, we can simply replace $\Z^{d}$ by the space group $S$ of a crystal.
In that case $\SPT^{C^*_f(S)}$ can still be expressed as twisted equivariant $K$-theory of the Brillouin zone, see \cite[Sections 9 and 10]{thiangk-theoretic} for further discussion.
\end{remark}

\begin{definition}
\label{def:freeSPTgroup}
If $G$ is a fermionic group, we define continuum/lattice SPT phases protected by $G$ to be continuum/lattice SPT phases protected by $C^*_f(G)$.
\end{definition}

\begin{remark}
\label{rem:computeSPT}
If $G$ is a compact fermionic group, Theorem \ref{th: twisted peterweyl} allows us to compute the group of free (say continuum) SPT phases
\[
\SPT_n^{C^*_f(G)} \cong K_{2-n}(C^*_f(G)^{op})
\]
in arbitrary dimension $n$ once we understand the irreducible representations of $G$, see Remarks \ref{rem:compute Kth} and \ref{rem:graded compute Kth}. 
\end{remark}

\begin{remark}
Let $D$ be one of the ten $\Z_2$-graded division algebras.
Note that unless $D = \R,Cl_{+1}$ or $Cl_{-1}$ the fermionic group algebra $C^*_f(D)$ is infinite-dimensional, compare Example \ref{ex:groupvsalgebratenfold}.
In those cases it follows from compactness of $S(D)$ and Remark \ref{rem:computeSPT} that $\SPT^{C^*_f(S(D))}_d$ is usually infinitely generated.
In particular, $\SPT^{C^*_f(S(D))}_d$ does not agree with the usual tenfold way classification $\SPT^{D}_d$.
We show how to cut the $K$-theory of $C^*_f(S(D))$ down to the usual classification in Sections \ref{sec: charge 1} and \ref{sec: spin1/2}, see Example \ref{ex:revisitclassA} for a detailed discussion for class A topological insulators.
\end{remark}

\begin{remark}
\label{rem:Green-Julg}
    Using an appropriate twisted Green-Julg theorem\cite[Section 11.7]{blackadar}, we expect that the group of lattice SPT phases protected by $C^*_f(G)$ is related to the twisted $G_b$-equivariant $KR$-theory of the Brillouin zone torus.
    Here $G$ acts trivially on the Brillouin zone but the real structure on the Brillouin zone is given by $k \mapsto -k$.
    There is a similar statement for continuum SPT phases.
\end{remark}

\section{Phases with charge conservation}
\label{sec:charged}

In this section, we will consider the case where our fermionic symmetry group $(G,(-1)^F,\theta)$ contains a designated $U(1)$-subgroup that we would like to call charge.
Such groups are equivalent to the symmetry groups introduced by Freed and Moore \cite[Definition 3.7]{freedmoore}.
We outline how their definition relates to the setting of this paper.
In Theorem \ref{th: main theorem}, we show that our Definition \ref{def:SPT phase} of free SPT phases to Definition \ref{def:unit charge SPTs} agrees with \cite{freedmoore}. For our theorem, we will have to restrict our representations to be of unit charge, see Section \ref{sec: charge 1}.
Just as for the non-charge conserving case, we show how to compute the group of unit charge SPT phases using a Peter--Weyl style result.

\subsection{Freed--Moore symmetry groups}
\label{sec: Freed--Moore symmetry}

The main algebraic input data for the twisted equivariant $K$-theory groups in \cite{freedmoore} is the notion of an extended QM symmetry class \cite[Definition 3.7(i)]{freedmoore}.
We reformulate their definition as follows:

\begin{definition}
A \emph{Freed--Moore group} is a second countable Hausdorff topological group $G^\tau$ together with 
\begin{itemize}
    \item a continuous homomorphism $\theta\colon G^\tau \to \Z_2$;
    \item a continuous homomorphism $\phi\colon G^\tau \to \Z_2$;
    \item a designated $U(1)$-subgroup $U(1)_Q \subseteq G$;
\end{itemize} 
with the property that 
\[
e^{i\alpha Q} g = g e^{i \phi(g)\alpha Q} \quad \forall g \in G^\tau,
\]
where we wrote elements of $U(1)_Q$ as $e^{i\alpha Q}$ where $\alpha \in [0,2\pi)$.
\end{definition}

Here we included the symbol $Q$ to remind us that we will require this $U(1)_Q$ through the charge operator as in Section \ref{sec: nambu}.
We denote the sum of the two homomorphisms to $\Z_2$ by $c := \theta + \phi$.

\begin{remark}
    The definition of Freed and Moore uses $c = \phi + \theta$ instead of our $\theta$ as a datum in their definition.
    They use the symbol $t$ to denote what we call $\theta$.
\end{remark}

We call the designated $U(1)$-subgroup $U(1)_Q$ the \emph{charge symmetry}.

\begin{remark}
    Since complex conjugation is the only continuous automorphism of $U(1)$, we could have equivalently defined a Freed--Moore group as a group $G^\tau$ together with a homomorphism $\theta\colon G^\tau \to \Z_2$ and a designated normal $U(1)$-subgroup $U(1)_Q \subseteq G^\tau$.
    By extension theory, there then exists a unique homomorphism $\phi$ from the quotient $H := G^\tau/U(1)$ to $\Aut_{cts}(U(1)) = \Z_2$ satisfying the desired properties.
\end{remark}

To provide our physical interpretation of a Freed--Moore group, we put the algebraic information into the context of Section \ref{sec: physical motivation}.
We do not claim our implementation is the unique correct one, but we will argue it to be a reasonable one.

As before, $\theta$ determines whether symmetries are time-preserving or time-reversing, i.e. whether they are required to act anti-unitarily or unitarily on Fock space.
On the other hand, $\phi$ determines whether symmetries act anti-unitary or unitary \emph{before} second quantizing.
We have seen in Section \ref{sec: nambu} that if particle-hole symmetries are present, then $\phi \neq \theta$.

We make a Freed--Moore group into a fermionic group by letting $(-1)^F$ be the nontrivial element $-1 = e^{i\pi Q}\in U(1)_Q$.
The physical reason we require this, is that we would like the spin-charge relation to hold, which is related as follows, also see Example \ref{example: class A}.

Consider a finite-dimensional representation $R$ of $G^\tau$ on $M$ such that $R((-1)^F) = -1$, which is the same as a finite-dimensional $C^*_f(G^\tau)$-module.
As in Section \ref{sec: nambu}, we denote the complexification of $M$ by $W$ and its real structure by $\gamma$.
We can then decompose the $U(1)_Q$-representation $W$ into weight spaces
\[
W = \bigoplus_{n \in \Z} W_n,
\]
so that $W_n$ is the subspace of \emph{charge $n$ fields} consisting of $w_n \in W$ such that $R(e^{i \alpha Q}) w_n = e^{i n \alpha} w_n$.
Therefore the condition that $R((-1)^F) = -1$ means that $W_{n} = 0$ if $n$ is even.
In other words, the spin-charge relation we impose on our fermion systems requires all particles to have odd charge.

Also note that since $R$ is real, $\gamma$ restricts to a complex anti-linear bijection $W_n \to W_{-n}$:
\[
R(e^{i \alpha Q}) \gamma w_n = \gamma R(e^{i \alpha Q}) w_n = \gamma e^{in \alpha} w_n =  e^{-in \alpha} \gamma w_n,
\]
so the particle-hole conjugation exchanges charge $n$ and charge $-n$ fields.
Moreover, if $g \in G^\tau$ is an element of a Freed--Moore group, then 
\begin{align*}
R(e^{i \alpha Q}) R(g) w_n &= R(e^{i \alpha Q} g) w_n = R(g e^{i (-1)^{\phi(g)} \alpha Q}) w_n =R(g) R( e^{(-1)^{\phi(g)} i \alpha Q}) w_n 
\\
&=  R(g) e^{ (-1)^{\phi(g)} i n \alpha} w_n = e^{(-1)^{\theta(g) + \phi(g)} i n \alpha} R(g)  w_n
\end{align*}
so $R(g) w_n \in W_{c(g)}$ for $c(g) = \theta(g) + \phi(g)$.
We will therefore interpret symmetries with $c(g) = 1$ as charge-reversing, i.e. they are the particle-hole symmetries.

It follows also from the above analysis that for every $n>0$ odd, the charge $(n,-n)$ sector $W_{-n} \oplus W_n$ is still a representation of $G^\tau$.
We conclude that a representation of $G^\tau$ splits up into sectors in which fields have the same or opposite charge.
Fixing one such sector, we can interpret $V := W_n$ as a one particle Hilbert space with complex structure $I := e^{i\pi Q/2} = i Q$, i.e. the charge operator gives a canonical polarization.
There are isomorphisms $W_{-n} \cong \overline{W_n} \cong \overline{V}$ given by $\gamma$. 
We get back to the original picture of Nambu space 
\[
W_n \oplus W_{-n} \cong V \oplus \ol{V}
\]
of Section \ref{sec: nambu} with particles of charge $n$ and antiparticles of charge $-n$.

To compare with results in the literature, we often need to restrict to the unit charge $Q = \pm 1$ subsector, as we discuss in detail in Section \ref{sec: charge 1}.
The main idea is

\begin{definition}
Let $G^\tau$ be a Freed--Moore group and $(W, \gamma, R)$ a representation of $G^\tau$ such that $R((-1)^F) = -1$.
We say $R$ has \emph{unit charge} if 
\[
R(e^{aiQ}) = \cos a + R(iQ) \sin a
\]
for all $e^{iaQ} \in U(1)_Q$.
\end{definition}

\begin{remark}
Restricting to the unit charge sector might kill many interesting representations of $G^\tau$ and representations of charge $n \neq 1$ can behave very differently.
For example, take $H:= G^\tau/U(1) = \Z_3 \times \Z_3$.
The possible extensions $G^\tau$ of $H$ are classified by $H^2(H,U(1))$ which can be shown to be isomorphic to $H^3(\Z_3 \times \Z_3, \Z) \cong \Z_3$ by the K\"unneth isomorphism.
For one of the two generators of the cohomology, the resulting $G^\tau$ is generated by $g_1,g_2, e^{iaQ}$ satisfying $g_1^3 =1, g_2^3 =1$ and $g_1 g_2 = e^{2\pi iQ/3} g_2 g_1$.
A representation $R$ of charge 3 satisfies $R(g_1)^3 =1, R(g_2)^3 =1$ with $R(g_1)R( g_2) = R(g_2) R(g_1)$, but a representation of charge 1 satisfies $R(g_1)R( g_2) = R(e^{2 \pi i/3}) R(g_2)R( g_1)$ instead.
Note that the algebra of charge 3 is commutative and hence there are $3 \cdot 3 = 9$ irreducible (necessarily one-dimensional) representations, while the algebra corresponding to charge 1 is noncommutative and so must have a irreducible representation of dimension at least two.
\end{remark}

\begin{remark}
    Note that BdG Hamiltonians that are symmetric under a Freed--Moore symmetry group are always charge-conserving and so have no off-diagonal terms $\Delta$.
Hence purely within the setting of this paper, Freed--Moore groups are at first sight inadequate for describing topological superconductors in the context of BdG theory directly.
However, including particle-hole symmetries to make the Freed--Moore group $U(1)_Q \rtimes \Z_2$ will give a Morita-equivalent symmetry algebra to the uncharged theory with symmetry group $\Z_2^F$, compare Example \ref{ex:class D}.
Therefore for understanding noninteracting topological phases this approach suffices.
For interacting phases more care is needed~\cite{stehouwer2022interacting}.
\end{remark}

\subsection{Symmetry algebras with charge}

For proving our main result, it is mathematically convenient to first connect with our formalism using symmetry algebras, as opposed to groups.
This will also avoid the unit charge problem, which we will tackle in the next section.

We abstract the situation of Freed--Moore groups as follows. 
Let $A$ be a real unital $C^*$-algebra with a $\Z_2$-grading $A = A_0 \oplus A_1$ denoted $\theta$.
Suppose we are given a designated $\Z_2$-graded $C^*$-subalgebra $\pi \colon \C \hookrightarrow A$ where we give $\C$ the trivial grading.
Instead of assuming that $A$ becomes a complex $C^*$-algebra, we will assume $A = A_u \oplus A_{au}$ comes with yet another $\Z_2$-grading $\phi$ such that 
\begin{equation}
\label{eq:anti-unitary}
    \pi(i) a = (-1)^{\phi(a)} a \pi(i) \in A
\end{equation} 
for homogeneous $a \in A$\footnote{The subscripts $u$ and $au$ stand for `unitary' and `anti-unitary'.}.
Since $\pi(i)$ commutes with the grading operator associated to $\theta$, the two grading operators also commute and we obtain a direct sum of vector spaces
\[
A = A_{u, 0} \oplus A_{au, 0} \oplus A_{u, 1} \oplus A_{au, 1}.
\]
Let $c$ denote the induced diagonal grading $\theta + \phi$.

\begin{lemma}
\label{lem:polsvsgradings}
    Let $J \colon M \to M$ be an $A$-symmetric polarization with respect to $\theta$.
    Then $\epsilon := J \pi(i)$ is a $\Z_2$-grading on $M$ with respect to the grading $c$.
\end{lemma}
\begin{proof}
    Since $J$ and $\pi(i)$ commute, $\epsilon^2 = 1$.
    Let $a \in A$, which we can assume is homogeneous for both gradings without loss of generality.
    By assumption, $a J = (-1)^{\theta(a)} Ja$ and $a\pi(i) = (-1)^{\phi(a)} \pi(i) a$. We see that $\epsilon a = (-1)^{c(g)} a \epsilon$, which means $\epsilon$ is a grading operator with respect to $c$.
\end{proof}

We will denote the Karoubi $K$-theory of $A$ by $K^{Kar, c}_0(A)$ if we want to emphasize $A$ is graded by $c$.

\begin{corollary}
\label{cor:FMvsSPTalgebrasunital}
    There is an isomorphism of groups $\SPT^A \cong K^{Kar, c}_0(A)$.
\end{corollary}
\begin{proof}
    The mapping $\Pol_A \to \Grad_A$ of Lemma \ref{lem:polsvsgradings} is clearly a bijection which is compatible with direct sums.
    By checking on free modules, we see it is a homeomorphism.
    Therefore there is an isomorphism $\SPT^A \cong K^{Kar, c}(A)$ given by 
    \[
    [M,J_1,J_2] \mapsto [M,J_1\pi(i), J_2 \pi(i)].
    \]
\end{proof}

We also need a similar result for $A$ non-unital.

\begin{proposition}
\label{prop:FMvsSPTalgebrasunital}
    Let $A$ be a $\Z_2$-graded real $C^*$-algebra equipped with another grading $A = A_u \oplus A_{au}$ and a $*$-subalgebra $\pi \colon \C \hookrightarrow M(A)$ which induces this grading in the sense of Equation \eqref{eq:anti-unitary}.
    Then $\SPT^A \cong K^{Kar, c}_0(A)$, where $c$ is the diagonal grading.
\end{proposition}
\begin{proof}
    First note that it still makes sense that the $\C$-subalgebra of the multiplier algebra induces the grading.
Indeed, because $A$ is an ideal in $M(A)$, we can make sense of the equation $\pi(i) a = (-1)^{\phi(a)} a \pi(i)$ in $A$.
Now note how $A_+$ admits the same direct sum decomposition into four pieces where $A_{u,0}$ becomes the unitization of $A_{u,0}$.
We therefore get the isomorphism $\SPT^{A_+} \cong K^{Kar, c}_0(A_+)$ from Corollary \ref{cor:FMvsSPTalgebrasunital}.
Since this isomorphism is compatible with the maps induced by $A_+ \to \C$ and so gives an isomorphism $\SPT^{A} \cong K^{Kar, c}_0(A)$ for nonunital $A$.

\end{proof}

\subsection{Charge 1 topological phases}
\label{sec: charge 1}

In this section, we will approach the unit charge requirement from a $C^*$-algebraic perspective in a similar spirit to Section \ref{sec:infinitegroups}.
Reflecting upon our $K$-theoretic Definition \ref{def:SPT phase} of free fermion SPTs, we could ask if $G^\tau$ is a Freed--Moore group, whether Definition \ref{def:freeSPTgroup} agrees with the twisted equivariant $K$-theory defined by Freed--Moore.
The answer to the question as stated is unfortunately no.
In this section, we will discuss how to modify the fermionic group algebra for a Freed--Moore group so that Definition \ref{def:SPT phase} recovers the Freed--Moore approach.

The idea is to appropriately incorporate the unit charge condition.
We have seen in Section \ref{sec:posdim} that if $D$ is a $\Z_2$-graded division algebra with $D_{ev} = \R$, then free SPT phases with symmetry $S(D)$ in spatial dimension $d$ are classified by $K_{2-d}(D^{op})$ in agreement with the tenfold way table.
However, if $D_{ev} = \C$, then we face the problem that $C^*_f(S(D))$ is huge and has many more representations than $D$.
For the quaternionic case $D_{ev} = \mathbb{H}$ the situation is similar and we restrict to the spin $1/2$ sector, see Section \ref{sec: spin1/2}.
We now treat the case $D = \C$ in detail to illustrate the situation, compare \cite[Section 2.4]{kapustinfreeinteracting}:

\begin{example}
\label{ex:revisitclassA}
We return to Example \ref{example: class A} but now interpret the symmetry group $U(1)$ as a Freed--Moore group.
Note that for the $\Z$-worth of complex irreducible representations of $U(1)$, charge $n$ and $-n$ for $n>0$ pair to a single irreducible representation over the real numbers of complex type.
We eliminated the even charge representations by imposing the spin-charge relation and so
\[
C^*_f(U(1)) \cong \bigoplus_{n>0 \text{ odd}} \C.
\]
So even though we cut down the size of $C^*(U(1))$, it still has an infinite amount of irreducible representations.
Since $K$-theory is additive in infinite direct sums of $C^*$-algebras, we get
\[
\SPT^{C^*_f(U(1))}_d \cong K_{2-d}(C^*_f(U(1))) = 
\begin{cases}
\bigoplus_{n >0 \text{ odd}} \Z & d \text{ even,}
\\
0 & d \text{ odd}.
\end{cases}
\]
This $K$-theory group can distinguish phases in which particles have different (odd) charge.
This is unlike the conventional class A classification, which is instead 
\[
K_{2-d}(\C) = 
\begin{cases}
\Z & d \text{ even,}
\\
0 & d \text{ odd},
\end{cases}
\]
the special case of $d=2$ being the quantum Hall effect.
In contrast to the representations of $U(1)$, there is just a single simple module of $\C$ and this module corresponds to the unit charge representations $\pm 1$.
So we will get back to a more conventional classification by enforcing representations to have charge $\pm 1$.
\end{example}

For a general Freed--Moore group $G^\tau$, we now apply a similar approach to how we defined the fermionic group algebra; 
consider $\C$ as a real $C^*$-algebra and the action $\alpha$ of $H:= G/U(1)_Q$ on $\C$ given by $z \mapsto \overline{z}$ if $\phi(h) = - 1$ and the identity otherwise (note that since $\phi\colon G^\tau \to \Z_2$ is continuous it factors through $H$).
As before, we assume a Borel measurable section $s\colon H \to G^\tau$ from which we can get a measurable $2$-cocycle $\tau\colon H \times H \to U(1)_Q$ where now $U(1)_Q$ is a nontrivial $H$-module under $\alpha$.
We get a twisted $C^*$-dynamical system $(\C, H, \alpha, \tau)$ and define

\begin{definition}
\label{def:unitcharge}
The \emph{unit charge group algebra} $C^*_u(G^\tau)$ of a Freed--Moore group $G^\tau$ is the associated twisted group $C^*$-algebra
\[
\C \rtimes_{\tau, \alpha} H
\]
associated to the extension
\[
1 \to U(1) \to G^\tau \to H \to 1.
\]
\end{definition}

We consider $C^*_u(G^\tau)$ as a real $C^*$-algebra, graded by $c$ or $\theta$ depending on context.
Note that if $\phi$ is nontrivial, then $C^*_u(G^\tau)$ does not admit a natural complex $C^*$-algebra structure because the canonical sub $C^*$-algebra $\C \subseteq M(C^*_u(G^\tau))$ given by multiples of the identity is not central.

The following theorem is proven at the end of Section \ref{sec:twistedpeterweyl}.

\begin{theorem}
\label{th:unit charge fd}
Let $(G^\tau, \theta, \phi, U(1)_Q)$ be a compact second-countable Freed--Moore group.
There is an isomorphism of real $C^*$-algebras
\[
C^*_{uc} (G^\tau) \cong
\bigoplus_{\substack{\text{f.d. irreps }\rho \\ \text{of unit charge}}} M_{\dim_{D(\rho))} (\rho)} (\End_{D(\rho)} \rho).
\]
The isomorphism respects the $\Z_2$-grading $\theta$.
\end{theorem}

\begin{definition}
\label{def:unit charge SPTs}
The group of \emph{continuum/lattice SPT phases of unit charge} with Freed--Moore symmetry $(G^\tau,\theta, \phi, U(1)_Q)$ in spatial dimension $d$ is the group of continuum/lattice SPT phases protected by $C^*_u(G^\tau)$.
\end{definition}

We see that continuum SPT phases are given by\footnote{Note that here we are taking the opposite with respect to the grading $\theta$.}
\[
\SPT_d^{C^*_u(G^\tau)} \cong K_{2-d}(C^*_{uc}(G^\tau)^{op}).
\]

\begin{example}
If $D$ is a real $\Z_2$-graded division algebra over $\R$ such that $D_{ev} = \C$, then by construction the unit charge group algebra of the resulting Freed--Moore group $S(D)$ is $C^*_{uc}(S(D)) = D$.
Therefore, SPT phases protected by the Freed--Moore group $U(1)$ recover the usual classification of class A phases.
\end{example}

\begin{remark}
An alternative and simpler approach to the unit charge problem that is mathematically viable, is to replace the fermionic group $U(1)_Q$ with the fermionic group $\Z_4^Q$ everywhere.
However, we choose to work with the symmetry group $U(1)$ because it more closely resembles the physics of charged particles.
For example, in continuum quantum electrodynamics, gauge fields are really $U(1)$-connections, while $\Z_4$-bundles have no nontrivial connections.
\end{remark}

\begin{remark}
Confusingly, is not true for Freed--Moore groups that $C^*_{uc}(G)^{op} \cong C^*_{uc}(G^{op})$ if $\phi$ is nontrivial.
Indeed, we can compute in $C^*_{uc}(G)^{op}$ that
\begin{align*}
    j(g)^{op} j(h)^{op} &= (-1)^{\theta(g) \theta(h)} (j(h) j(g))^{op} = (-1)^{\theta(g) \theta(h)} (\tau(h,g) j(hg))^{op} 
    \\
    &= (-1)^{\theta(g) \theta(h)} j(hg)^{op} \tau(h,g)
    = (-1)^{\theta(g) \theta(h) + \phi(hg)} \tau(h,g) j(hg)^{op},
\end{align*}
where $j$ denotes the embedding of $H = G^\tau/U(1)$ into $M(C^*_{uc}(G^\tau))$ as explained in Appendix \ref{sec:twistedsemidirect}.
This leads to the expectation that $C^*_{uc}(G)^{op}$ is a unit charge twisted group algebra of $G$ twisted by $\tau'(g,h) = (-1)^{\theta(g) \theta(h) + \phi(g) + \phi(h)} \tau(h,g)$.
In other words, the opposite of a Freed--Moore group should be defined differently from the opposite of its underlying fermionic group.
\end{remark}

\subsection{A comparison with Freed--Moore K-theory}
\label{sec: Freed--Moore K-theory}

In the work of Freed and Moore \cite{freedmoore}, an extensive rigorous framework is laid out to deal with free fermion SPT phases protected by a Freed--Moore symmetry group.
Instead of polarizations on Nambu space, flattened Hamiltonians on one particle Hilbert space give elements of Freed--Moore $K$-theory.
Mathematically, when we grade the $C^*$-algebra $C^*_{uc}(G^\tau)$ with $c = \theta + \phi$ instead of $\theta$, its graded representations are in one-to-one correspondence with $(\phi,\tau,c)$-twisted representations of $G^\tau$ as defined in \cite[Definition 3.7(ii)]{freedmoore}.
Up to the subtlety that finite-dimensional representations of $C^*_{uc}(G^\tau)$ are not the same as finitely generated Hilbert modules, (pairs of) such representations naturally live in the Karoubi $K$-theory group $K_0^c(C^*_u(G^\tau)^{op})$.\footnote{The opposite comes from a left vs right convention issue, see Remark \ref{rem:FMKthdetails}} 
Here we used the superscript $c$ to emphasize that we are using $c$ as our $\Z_2$-grading.
Note that one particle Hamiltonians naturally flatten to gradings as opposed to polarizations.
Moreover, the pseudosymmetries which anticommute with the Hamiltonian are modeled by elements with nontrivial $c(g)$ and so Hamiltonians are $G^\tau$-skewlinear with respect to the grading $c$.
We have thus motivated the following definition:

\begin{definition}
\label{def: Freed--Moore K-theory}
Let $(G^\tau,U(1)_Q,\theta,\phi)$ be a Freed--Moore group.
The \emph{Freed--Moore $K$-theory} of $G^\tau$ is defined as
\[
K_0^c(C^*_u(G^\tau)^{\op}).
\]
\end{definition}

We start with a few remarks for readers who are interested in a detailed comparison between this definition of Freed--Moore K-theory and the definitions in \cite{freedmoore}.

\begin{remark}\textcolor{white}{.} 
\label{rem:FMKthdetails}
\begin{enumerate}
    \item The opposite in Definition \ref{def: Freed--Moore K-theory} is related to our conventions for Karoubi $K$-theory, see Remarks \ref{rem:ABSconvention} and \ref{rem:Kasparovconvention}.

    \item Our Definition \ref{def: Freed--Moore K-theory} does not quite agree with \cite[Definition 7.1]{freedmoore} unless $c$ is trivial or $H$ is finite, because the finite-dimensional formulation of Freed--Moore diverges from the usual Karoubi formulation.
    This discrepancy is covered extensively in \cite[Section 4.3]{gomi2017freed}, see in particular the example at the end which can be reproduced in the current setting by taking $H = \Z \times \Z_2, c$ projection on the second coordinate and $\phi,\tau$ trivial.
    There is another minor pitfall in their definition of trivial $\Z_2$-graded modules, in which one should include the requirement that the odd automorphism squares to $-1$.\footnote{Requiring it to square to $1$ would result in the $K$-theory of the opposite. One can check that squaring to $-1$ gives the correct convention by comparing for example class AI and AII which are each other's opposite, also see item 1 in this remark.}
    \item In Definition 7.33 of Freed and Moore, $K$-theory is defined more generally for a space $X$ with an action of a group $G$ and twists $\tau$ not just depending on $G$ but also on $X$. 
    Our Definition \ref{def: Freed--Moore K-theory} could be compared to their Definition 7.1, which concerns the special case when $X$ is a point.
    However, for applications in condensed matter physics this is not a big loss: $X$ is typically the Brillouin zone torus (or sphere), which can be incorporated in the Freed--Moore group using the crystalline equivalence principle, see Section \ref{sec:posdim}.
    In particular, the reason for twisting the equivariant $K$-theory group of the Brillouin zone is then two-fold.
    The first is to incorporate Freed--Moore groups with a nontrivial extension by $U(1)$ (i.e. it deals with the twist $\tau$), the second is to accommodate for nonsymmorphic space groups, compare Remarks \ref{rem:Green-Julg} and \ref{rem:crystalline}.
\end{enumerate}
\end{remark}

We now provide a comparison between our Definition \ref{def:SPT phase} of SPT phases and Freed--Moore $K$-theory.

\begin{theorem}
    \label{th: main theorem}
    Let $G^\tau$ be a Freed--Moore group.
    The Freed--Moore $K$-theory of $G^\tau$ is isomorphic to the group of unit charge SPT phases protected by $G^\tau$.
\end{theorem}
\begin{proof}
Note that $A := C^*_u(G^\tau){op}$ is a(n in general non-unital) $\Z_2$-graded real $C^*$-algebra with another grading $\phi$ induced by the grading $\phi$ of $G^\tau$.
It comes equipped with a subalgebra $\C \subseteq M(C^*_u(G^\tau)^{op})$ which induces the $\phi$ grading.
The result follows from Proposition \ref{prop:FMvsSPTalgebrasunital}.
\end{proof}

\begin{remark}
    We expect that there is a close relationship between the above theorem and \cite[Theorem 3.12]{gomi2017freed}.
\end{remark}

\subsection{Spin 1/2 topological phases}
\label{sec: spin1/2}

We have seen in the last section that if the symmetry group $G^\tau$ contains a $U(1)$ normal subgroup $U(1)_Q \subseteq G^\tau$, we can require representations to have unit charge by looking at the twisted group algebra of $G^\tau / U(1)_Q$.
We encounter a similar situation if a fermionic group $G^h$ contains a distinguished normal subgroup $SU(2) \subseteq G^h$.
Namely, $SU(2)$ has representations of arbitrary spin $s \in \frac{1}{2}\Z_{\geq 0}$, but we might want to restrict to the case of spin $1/2$ corresponding to the canonical representation of $SU(2) \cong Sp(1)$ on $\mathbb{H}$.

This situation is relevant for several descriptions of topological phases that incorporate the $\Spin(3) = SU(2)$-rotation invariance of electrons directly into the symmetry group.\footnote{In the lab, $(2+1)$-dimensional systems are achieved by confining a spinful particle in our $(3+1)$-dimensional spacetime to a slab, which does not change the representation type of the $(3+1)$-dimensional spinors.}
For example, the original Altland--Zirnbauer paper argues that for example class C has an $SU(2)$-symmetry, which only gets broken to the class D $\Z_2^F$-symmetry when relevant impurities create spin-orbit interactions \cite{altlandzirnbauer}.
In that case where the $SU(2)$ symmetry is not internal, we might need to be more careful coupling this symmetry to $\R^3$ appropriately via $SO(3)$.

To achieve this, we generalize the real and complex case to the quaternions, omitting the details.
The only significant difference is that while the automorphism groups of the real $C^*$-algebras $\R$ and $\C$ are discrete, the automorphism group of $\mathbb{H}$ is $\Aut SU(2) \cong SO(3)$, which has interesting topology.
Indeed, all automorphisms are inner and the center of $\mathbb{H}$ is $\R$ so that $\{\pm 1\} \subseteq SU(2)$ acts trivially, resulting in 
\[
\Aut SU(2) \cong \frac{SU(2)}{\{\pm 1\}} \cong SO(3).
\]
So assume the existence of a Borel measurable section $s$ of
\[
1 \to SU(2) \to G^h \to H \to 1.
\]
For example, we might require this exact sequence to be a fiber bundle.
We then get a measurable action $\phi\colon \mathbb{H} \to \Aut SU(2) \cong SO(3)$ and a measurable nonabelian cocycle $\tau\colon \mathbb{H} \times \mathbb{H} \to SU(2)$.
We can then define the spin $1/2$ group algebra to be the twisted group algebra 
\[
C^*_{s = 1/2}(G) = \mathbb{H} \rtimes_{\phi, \tau} H.
\]

\begin{definition}
\label{def:spin 1/2 SPTs}
The group of \emph{SPT phases of spin 1/2} with symmetry $(G^h,\theta)$ and distinguished subgroup $SU(2) \subseteq G^h$ in spatial dimension $d$ is $\SPT^{C^*_{s = 1/2}(G)}_d$.
\end{definition}

By Proposition \ref{prop:posdim}, the group of SPT phases of spin 1/2 is isomorphic to $K_{2-d}(C^*_{s = 1/2}(G)^{op})$.

\begin{example}[class CI]
Consider the internal symmetry group
\[
 \Pin^+_3 \cong \frac{SU(2) \times \Z^{TF}_4}{\Z^F_2}
\]
associated to the $\Z_2$-graded division algebra $Cl_{+3}$.
The fermionic group algebra of $\Pin^+_3$ is infinite-dimensional, but
\[
C^*_{s = 1/2}(\Pin^+_3) \cong Cl_{+3}.
\]
Therefore SPT phases of spin 1/2 protected by the symmetry group $\Pin^+_3$ agree with SPT phases protected by the symmetry algebra $Cl_{+3}$.
We have seen in Section \ref{sec: ten division algebras} it is given in spatial dimension $d$ by
\[
K_{2-d}(C^*_{s = 1/2}(G)^{op}) = K_{2-d}(Cl_{-3}) \cong K_{-d-1}(\R),
\]
in agreement with the usual Bott-periodic classification of class CI topological phases.
\end{example}

\begin{remark}
The group $SU(2)$ has many $U(1)$-subgroups.
Fixing one of these, it makes sense to ask whether a representation of $G^h$ has unit charge.
In fact, from the representation theory of $SU(2)$ it follows that a representation has spin $1/2$ if and only if it has unit charge.
However, we cannot apply the theory of the Section \ref{sec: charge 1} to describe spin $1/2$ representation theory, since $SU(2)$ has no normal $U(1)$-subgroups.
\end{remark}

\begin{remark}
\label{rem:class AI}
It can happen that a symmetry group contains multiple subgroups $U(1),SU(2)$ for which we want to apply the unit charge and/or spin 1/2 relation.
We will not make the construction of the resulting $\Z_2$-graded $C^*$-algebra precise, but instead give a relevant example.
In the approach to class AI of \cite{altlandzirnbauer}, the symmetry algebra is
\begin{align*}
    \mathbb{H} \otimes_\R \frac{\R[iQ, T]}{((iQ)^2 = -1, T^2 = -1, iQT = -TiQ)} \cong \mathbb{H} \otimes Cl_{-2}
\end{align*}
where $T$ is time-reversing and hence odd.
If we want to describe this instead using a group such as $(SU(2) \times \Pin^-(2))/\Z_2$ (compare \cite[Table 1]{alldridgemaxzirnbauer}), we need to apply both unit charge and spin 1/2 relations.
See Remark \ref{rem:soc} how this example relates to the description of class AI by a time-reversal with square $1$.
\end{remark}

\section{Tenfold ways compared}
\label{sec:tenfold}

Recall from Sections \ref{sec: ten division algebras} and \ref{sec: ten fermionic groups} that the classification of $\Z_2$-graded division algebras gives a natural collection of ten fermionic groups.
Three of these (corresponding to class D,DIII and BDI) do not have a $U(1)$-subgroup and so are not Freed--Moore groups.
To incorporate these classes in the charged formalism, the idea is to instead introduce an extra particle hole symmetry to `cancel' the charge.
A succinct way to describe the symmetry classes of the tenfold way in this fashion was first introduced in \cite{schnyder2008classification} and made mathematically rigorous in \cite{freedmoore}.
The possibility to introduce two symmetries in addition to charge\footnote{The fact that in this formalism the symmetry group contains charge by default is often not explicitly mentioned.} is provided, a particle-hole symmetry $C$ and time-reversal symmetry $T$.
Here $C$ is complex linear on Fock space and anticommutes with charge but complex anti-linear on one particle space, also see the discussion in Section \ref{sec: nambu}.
Note that since our particles are by default charged, the symmetry class `without symmetries' here is class A in contrast to the approach described in Section \ref{sec:neutralfermion} where it is instead class D.
The goal of this section is to understand how the internal symmetry algebras and groups the CT mechanism provides are related to the symmetry groups introduced in Sections \ref{sec: ten division algebras} and \ref{sec: ten fermionic groups} respectively.

We start by recalling the mathematical description of such CT-groups following \cite[Section 6]{freedmoore}.
We then proceed to relate them to the ten Morita classes of simple $\Z_2$-graded algebras over $\R$ as described in Section \ref{sec: ten division algebras}.
Let $\mathcal{C} := \{1,T\} \times \{1,C\}$ be the group $\Z_2 \times \Z_2$ written multiplicatively, so $T^2 = C^2 = 1$.
We define $\theta\colon \mathcal{C} \to \Z_2$ as the homomorphism determined by $\theta(T) = 1$ and $\theta(C) = 0$ and $\phi\colon \mathcal{C} \to \Z_2$ as $\phi(T) = 1$ and $\phi(C) = 1$.
So $c = \phi + \theta\colon \mathcal{C} \to \Z_2$ satisfies $c(T) = 0$ and $c(C) = 1$.
A $CT$-group is now a choice of a subgroup of $\mathcal{C}$, together with a Freed--Moore group extending it:

\begin{definition}
A \emph{$CT$-group} $(G^\tau, \theta, \phi)$ is a Freed--Moore group such that
\[
G^\tau/U(1)_Q \cong A
\]
for some subgroup $A \subseteq \mathcal{C}$.
The homomorphisms $\theta,\phi\colon G^\tau \to \Z_2$ are required to be compatible with the fixed homomorphisms $\theta,\phi\colon \mathcal{C} \to \Z_2$.
\end{definition}

For the four $\Z_2$-graded division algebras $D$ with $D_{ev} = \C$, the corresponding internal symmetry group $S(D)$ is a Freed--Moore group.
In fact, they are all CT-groups. 

\begin{proposition}[{\cite[Proposition 6.4]{freedmoore}}]
\label{prop:ten CT-groups}
There are ten isomorphism classes of $CT$-groups.
More precisely, there are four isomorphism classes\footnote{An homomorphism of Freed-Moore groups is a homomorphism of extensions preserving both gradings.} of $CT$-groups with $A = \mathcal{C}$.
The four cases are distinguished by whether it is possible to find lifts $\hat{T}, \hat{C} \in G$ of $T,C \in A$ such that $\hat{T}^2 = \pm 1$ and $\hat{C}^2 = \pm 1$.
Similarly in the cases $A = \{1,C\}$ and $A = \{1,T\}$ there are two possible $G$ depending on the existence of a lift with square $-1$.
For $A = \{ 1,CT\}$, there is just one $CT$-group.
\end{proposition}
\begin{proof}
This can be reduced to a basic problem in group cohomology as follows.
An isomorphism class of $CT$-groups is specified by a subgroup $A \subseteq \mathcal{C}$ and an element of $H^2_m(A,U(1)_\phi) \cong H^3(BA,\Z_\phi)$.
Here $U(1)_\phi$ is the topological $A$-module $a \cdot z := z^{(-1)^{\phi(a)}}$ and $\Z_\phi$ the $A$-module $a \cdot n = (-1)^{\phi(a)} n$.
For $A$ one of the $\Z_2$ subgroups, we need the group cohomology of $\Z_2$ with and without twisted coefficients
\begin{align*}
    H^n(B\Z_2, \Z) &=
    \begin{cases}
    \Z & n=0,
    \\
    \Z_2 & n>0 \text{ even},
    \\
    0 & n \text{ odd}.
    \end{cases}
    \\
    H^n(B\Z_2, \Z_-) &=
    \begin{cases}
    \Z_2 & n>0 \text{ odd},
    \\
    0 & n \text{ even}.
    \end{cases}
\end{align*}
For $A = \mathcal{C}$, we can apply a K\"unneth formula.
The Tor terms vanish and the result is $H^3(\mathcal{C},\Z_\phi) = \Z_2 \times \Z_2$.
\end{proof}

For $H := G^\tau/U(1)$ finite, the unit charge group algebra 
\[
C^*_{uc}(G^\tau)
\]
of $(G^\tau,\theta,\phi)$ is isomorphic to the finite-dimensional real algebra generated by elements $i$ and $j(g)$ for $g \in H$ together with the relations 
\[
i^2 = -1, \quad j(g) j(h) = \tau(g,h) j(gh), \quad i j(g) = (-1)^{\phi(g)} j(g) i
\]
with the understanding that $i$ is even and $j(g)$ is odd or even depending on $\theta(g)$.

\begin{example}[class D]
\label{ex:class D}
    Take $A \subseteq \mathcal{C}$ to be the subgroup generated by $C$ and let $G^\tau$ be the trivial extension of $A$ by $U(1)$.
    In other words, $G^\tau = \Z_2 \rtimes U(1)$ with $\phi$ nontrivial and $\theta$ trivial.
    Then the unit charge group algebra is
    \[
    C^*_{uc}(G^\tau) \cong \frac{\R[i,C]}{(C^2 = 1, i^2 = -1, iC =- Ci)}.
    \]
    As an algebra graded by $\theta$ (which is trivial), it is therefore isomorphic to $M_2(\R)$.
    We see that $C^*_{uc}(G^\tau)$ is not isomorphic to the fermionic group algebra $C^*_f(\Z_2^F) \cong \R$ of the uncharged symmetry group $\Z_2^F$ of class D.
    However, they are Morita equivalent and so have isomorphic $K$-theory, which suffices for applications in weakly interacting phases, compare \cite{stehouwer2022interacting}.
    Note that under the $c$-grading used in \cite{freedmoore}, we have that $C^*_{uc}(G^\tau) \cong Cl_{+2}$ instead.
    This is an instance of the shift by two implicit in Theorem \ref{th: main theorem}.
\end{example}

\begin{proposition}
The ten real $\Z_2$-graded algebras obtained by forming the unit charge twisted group algebras $C^*_{uc}(A^\tau)$ corresponding to the ten $CT$-groups $(A,\tau)$ with $A \subseteq \mathcal{C}$ and $\tau \in H^2(A,U(1)_\phi)$ form a full set of representatives for the ten Morita classes of simple real $\Z_2$-graded algebras. 
\end{proposition}
\begin{proof}
The proposition follows by brute force computation of all ten cases.
The example below provides the case $A = \mathcal{C}$ and trivial twist $\tau$.
It easily adapts to the other three twists; we only have to recompute the square of $iCT$.
The other six cases are subalgebras of these; the results are displayed in Table \ref{tab:tenfold}.
\end{proof}

\begin{example}[class BDI]
The fermionic internal symmetry group describing class BDI as outlined in Section \ref{sec: ten fermionic groups} is given by $S(Cl_{+1}) = \Z_2^F \times \Z_2^T$.
A one-dimensional example of a nontrivial topological superconductor of this symmetry type is given by the time-reversal symmetric Kitaev chain~\cite{fidkowski2010effects}, see \cite{debraygunningham} for a mathematical treatment.

Consider the CT-group in which $A = \mathcal{C}$ and the cocycle is trivial, so that lifts of $C$ and $T$ square to $1$.
This results in a Freed--Moore symmetry group $G^\tau$ which is not isomorphic to $S(Cl_{+1})$.
We abuse notation and write the lifts in $A^\tau$ again as $C$ and $T$.
The unit charge twisted group algebra is
\begin{align*}
C^*_{uc}(G^\tau) =    \frac{\R[i,C,T]}{(i^2 = -1, C^2 = 1, T^2 = 1, CT = TC, iT = -Ti, iC = -Ci)}.
\end{align*}
Consider the generating set $\{T, iT, iCT\}$ of odd elements. A short computation verifies that they mutually anticommute and
\[
T^2 = (iT)^2 = - (iCT)^2 = 1.
\]
We see that the twisted group algebra is isomorphic to a Clifford algebra with two generators that square to one and one generator that squares to $-1$, so
\begin{align*}
    C^*_{uc}(G^\tau) \cong Cl_{+2} \otimes Cl_{-1} \cong M_{1|1}(Cl_{+1})
\end{align*}
is Morita equivalent as a $\Z_2$-graded algebra to $Cl_{+1}$.
We see that even though $C^*_{uc}(G^\tau)$ is not isomorphic to the fermionic group algebra $C^*_f(\Pin^+_1) = Cl_{+1}$ of class BDI, their $K$-theory groups agree in all degrees by Morita invariance.
Hence for describing free fermion phases, the Freed--Moore description via $G^\tau$ is equivalent to the description using the fermionic group $G = \Pin^+_1$.
\end{example}

\begin{corollary}
    As the symmetry group goes through the ten possible CT-groups, the group of SPT phases of unit charge in spatial dimension $d$ goes through the ten groups
    \[
    K^{d+k}(\pt)  \quad KO^{d+s}(\pt)
    \]
    for $k = 0,1$ and $s = 0, 1, \dots, 7$.
    The precise choices of $s$ given a CT-group are given in Table \ref{tab:tenfold}.
\end{corollary}

The connection with the ten fermionic groups associated to the ten real $\Z_2$-graded division algebras $D$ in Section \ref{sec: ten fermionic groups} is similar, but we have to take care of unit charge and spin $1/2$ restrictions so that $C^*_f(S(D)) \cong D$.
More precisely, we require unit charge if $D_{ev} = \C$ (Definition \ref{def:unit charge SPTs}) and spin 1/2 if $D_{ev} = \mathbb{H}$ (Definition \ref{def:spin 1/2 SPTs}).
After this, the $K$-theory of the twisted group $C^*$-algebra of $S(D)$ in a fixed spatial dimension will run through all ten cases in the periodic table, see Table \ref{tab:tenfold}.

We finally comment on the differences between the approach using $\Z_2$-graded division algebras and the Altland--Zirnbauer approach, see f.e. \cite[Table 1]{alldridgemaxzirnbauer}.
The only differences are in the classes AI and BDI; the former for example corresponds to the division algebra $Cl_{+2}$ which physically means charge and a time-reversal symmetry with square $1$.
On the other hand, Zirnbauer models class AI with charge, a time-reversal symmetry with square $-1$, and an $SU(2)$-symmetry.
See the following remark for further details.

\begin{remark}
\label{rem:soc}
We comment on spin-orbit coupling in a system protected by a fermionic group $G$.
In the original Altland--Zirnbauer article \cite{altlandzirnbauer}, it is argued that spin-orbit coupling breaks the $SU(2)$ internal spin degrees of freedom of the electron.
We interpret this in the setting of this article as using the symmetry group
\[
\begin{cases}
    G &\text{spin-orbit coupling present}
    \\
    \frac{G \times SU(2)}{\Z_2} &\text{spin-orbit coupling not present,}
\end{cases}
\]
where in the second line we quotiented by the diagonal $\Z_2$-subgroup generated by $((-1)^F, -1)$.
For example, let us assume a system of charged particles with a time-reversal symmetry $T^2 = (-1)^F$ (class AII).
We obtain the symmetry algebra
\[
\frac{\R[iQ, T]}{((iQ)^2 = -1, T^2 = -1, iQT = -TiQ)} \cong Cl_{-2}
\]
when spin-orbit coupling is present.
This is the same class AII description as we got before.
When there is no spin-orbit coupling, we impose both the spin 1/2 and unit charge conditions and obtain 
\begin{align}
\label{eq:ZirnbauerclassAI}
    \mathbb{H} \otimes_\R \frac{\R[iQ, T]}{((iQ)^2 = -1, T^2 = -1, iQT = -TiQ)} \cong \mathbb{H} \otimes Cl_{-2},
\end{align}
compare Remark \ref{rem:class AI}.
Some sources argue that introducing spin-orbit interactions effectively changes the square of time-reversal symmetry from $T^2 = 1$ to $T^2 = (-1)^F$, i.e. the symmetry algebra is 
\[
\frac{\R[iQ, T]}{((iQ)^2 = -1, T^2 = 1, iQT = -TiQ)} \cong Cl_{+2}.
\]
Note that this algebra is not isomorphic to the algebra displayed in \eqref{eq:ZirnbauerclassAI} because it has a different dimension.
However, by Bott periodicity of the Clifford algebras and $\mathbb{H}$ being Morita equivalent to $Cl_{-4}$, these two algebras are Morita equivalent and so have the same $K$-theory.
We thus conclude that the heuristic `spin-orbit coupling means replacing $T^2 = (-1)^F$ by $T^2 = 1$' is true up to Morita equivalence, and therefore might be a phenomenon that only holds in the weakly interacting setting~\cite{stehouwer2022interacting}.
\end{remark}

\begin{table}[h!] 
\centering
{
\renewcommand{\arraystretch}{1.5}
\resizebox{\columnwidth}{!}{
\begin{tabular}{l|l|l||l|l|l|l|l|l|l|l}
Symmetry class & A & AIII & D & BDI & AI & CI &  C & CII & AII & DIII 
\\
\hline
$\Z_2$-graded division algebra $D$ & $\C$ & $\C l_1$ & $\R$ & $Cl_{+1}$ & $Cl_{+2}$& $Cl_{+3}$ & $\mathbb{H}$ & $Cl_{-3}$& $Cl_{-2}$& $Cl_{-1}$
\\
\hline 
Internal symmetry $S(D)$ & $Spin_1^c$ & $Pin_1^c$ & $Spin_1$ & $Pin_1^+$  & $Pin_2^+$ & $Pin_3^+$ & $Spin_3$ & $Pin_3^-$ & $Pin_2^-$ & $Pin_1^-$ 
\\
\hline
$K$-theory group & $K^{d}$ & $K^{d+1}$ & $KO^{d-2}$ & $KO^{d-1}$ & $KO^{d}$ & $KO^{d+1}$ & $KO^{d+2}$ & $KO^{d+3}$ & $KO^{d+4}$ & $KO^{d+5}$ 
\\
\hline
$CT$-group & $00$ & diagonal & $+0$ & $++$ & $0 +$ & $- + $ & $-0$ & $--$ & $0-$ & $+-$ 
\\
\hline
unit charge $CT$-group algebra & $\C$  & $\C l_1$ & $M_{1|1}(\R)$ & $M_{1|1}(Cl_{+1})$ & $Cl_{+2}$  & $Cl_{+3}$   & $\mathbb{H}$   & $Cl_{-3}$   & $Cl_{-2}$  & $M_{1|1}(Cl_{-1})$   
\end{tabular}
}
}
\caption{A translation between different versions of the tenfold way. 
The $K$-theory group classifies the free fermion SPT phases of spatial dimension $d$ in the corresponding symmetry class (ignoring weak invariants).
It is given by the group of SPT phases protected by the symmetry algebra $D$ (Definition \ref{def:SPT phase}). 
Alternatively, it is given by the unit charge $K$-theory of the $CT$-group or equivalently by $K$-theory of the group $S(D)$, where we have to require unit charge if $D_{ev} = \C$ (Definition \ref{def:unit charge SPTs}) and spin 1/2 if $D_{ev} = \mathbb{H}$ (Definition \ref{def:spin 1/2 SPTs}).
The internal symmetry group $S(D)$ is defined in Section \ref{sec: ten fermionic groups}.
The $CT$-group row denotes the squares of the symmetries $C$ and $T$ and a zero means that the symmetry is absent, also see Proposition \ref{prop:ten CT-groups}. 
}\label{tab:tenfold}
\end{table}

\section{Relationships with interacting SPT phases}
\label{sec: interacting}

Even though recent years have seen a surge of interest \cite{ogata,bourneogata, ogata2021valued, Ogata_2022, kapustin2020thermal, bachmann2018quantization}, a comprehensive mathematical classification of strongly interacting SPT phases directly on the lattice remains incomplete. 
In spatial dimension one, SPT phases are well-understood in terms of matrix product states~\cite{bultinck2017fermionic, williamson2016fermionic, shiozaki2017matrix, kapustin2018spin, turzillo2019fermionic}.
On the other hand, bordism groups have been explored as an alternative approach \cite{kapustinfermionicSPT, freedhopkins} to amend proposals in terms of group (super-)cohomology.
In \cite[Section 9.2]{freedhopkins} two relationships between free fermions in continuum QFT and invertible topological quantum field theories (TQFTs)\footnote{Unitary invertible topological field theories are closely related to bordism groups because their partition functions are bordism invariants.} are sketched:
\begin{enumerate}
    \item given a $d$-dimensional massless free fermion living in some spin representation, its anomaly can be described by a $(d+1)$-dimensional invertible TQFT;
    \item a mass term for such a QFT trivializes the anomaly and makes the theory gapped. 
    Its low energy effective field theory is a $d$-dimensional invertible TQFT.
\end{enumerate}
Even though physically distinct, the map from $d$-dimensional massless free fermions to $(d+1)$-dimensional invertible TQFTs and the map from $d$-dimensional massive free fermions to $d$-dimensional TQFTs have strikingly similar mathematical definitions.
For the tenfold way symmetries, \cite{freedhopkins} give a description in the language of stable homotopy theory using the Atiyah-Bott-Shapiro orientation
\[
MSpin \to KO
\]
or analytically using Atiyah-Patodi-Singer eta invariants.
In joint work, we are developing a generalization of this description to arbitrary symmetry groups using twisted equivariant generalizations of the Atiyah-Bott-Shapiro orientation~\cite{hebestreit2020twisted, berwick2024power}, see \cite{bottspiral} for some special cases.
Alternatively, we expect it to be possible to define interacting phases protected by a symmetry $C^*$-algebra $A$ using bundles of $\Z_2$-graded bimodules over our bordisms in the spirit of \cite[Definition 2.3.1]{whatisanellipticobject}.

\begin{remark}  
The close relationship between our formulation of $K$-theory using polarizations and the formulation using mass terms~\cite{masstermgomiyamashita} allows one to write down a canonical free fermion Lagrangian from a $K$-theory class induced by a BdG Hamiltonian.
It is expected that the low-energy effective theory of this continuum quantum field theory gives the desired invertible TQFT.
\end{remark}

The goal of the rest of this section is to connect the second construction to the setting in this paper, in which we consider massive fermions in a condensed matter set-up.
We will not attempt to make a connection with lattice-based approaches such as \cite{bourneogata}.

\subsection*{The invertible topological field theory of a \texorpdfstring{$0+1d$}{0+1d} massive fermion}

We sketch our intuitive understanding of the low energy effective field theory of a massive particle in zero spatial dimensions.
Our hope is that this will set the stage for how to tackle higher dimensions, which we will work out in $1$d in joint work \cite{camluuk}.

Let $V$ be the finite-dimensional complex Hilbert space we think of as the one particle space of a charged particle.
Suppose first for simplicity that a finite symmetry group $G$ acts by unitary and charge conserving symmetries.
Let $h \colon V \to V$ be a free, charge-conserving, $G$-symmetric, gapped, self-adjoint Hamiltonian.

We want to describe the low-energy limit of this state of affairs as a one-dimensional invertible functorial quantum field theory.
It is reasonable to assign the lowest energy state on the many body Fock space to a point $+$.
This is the complex line given by the Fermi sea
\begin{equation}
\label{chargedgroundstate}
\mathcal{H}(+) := \bigwedge^{top} V_{-},
\end{equation}
where $V_{-}$ is the space of valence electrons for $h$.
Since $G$ commutes with $h$, $\mathcal{H}(+)$ is a unitary $G$-representation.
More generally, we could have allowed $G$ to have time-reversing symmetries; when $G$ acts anti-unitarily on $V$, it will also act anti-unitarily on $\mathcal{H}(+)$.
We get an invertible TQFT using the fact that that unitary TQFTs in one spacetime dimension are classified by (anti-)unitary representations of the symmetry group on $\mathcal{H}(+)$, see \cite[Proposition 4.10]{muller2023reflection}.
Here we take the target category of the TQFT to be $\Z_2$-graded Hilbert spaces, where we grade $\mathcal{H}(+)$ using $(-1)^F$.
Explicitly, the closed interval bordism from the point to itself for which the canonical flat $G$-connection on the trivial $G$-bundle has holonomy $g$ gets assigned the action of $g$ on $\mathcal{H}(+)$.

Note that the above approach does not obviously generalize to particle-hole symmetries; such symmetries will not map $\mathcal{H}(+)$ to $\mathcal{H}(+)$.
It it also not clear how to generalize to BdG Hamiltonians.
We now argue how our approach to $K$-theory using pairs of polarizations solves this by fitting naturally into the picture of functorial field theory, also see \cite{ludewiglangrangian}.

Namely, we claim it is more natural to consider the construction of a Fock module associated to a single polarization $J$ on a Majorana space $M$ as a boundary theory of a two-dimensional bulk TQFT.
This bulk TQFT is an extended theory with target $\mathbf{sAlg}_\C$ the $2$-category of complex $\Z_2$-graded algebras and given by assigning the finite-dimensional complex $\Z_2$-graded algebra $Cl(W) \in \mathbf{sAlg}_\C$ to a point.
For any polarization to exist, we need $W$ to be even-dimensional, which means that $Cl(W)$ is Morita-trivializable.

\begin{remark}
If $W$ is odd-dimensional, the nontrivial $2$-dimensional bulk TQFT measures the anomaly of an odd number of Majorana fermions, also see the discussion at the end of Section \ref{sec: intermezzo}.
We plan to work out the details of this anomaly with arbitrary symmetry group in future work \cite{freefermlukas}.
In particular, the $H^*$-algebra structure on $Cl(W)$ makes the anomaly TQFT unitary in the extended sense.
\end{remark}

Suppose now $W$ is even-dimensional.
Since there are two isomorphism classes of self-Morita equivalences of the (trivially graded) algebra $\C$, there are two Morita trivializations of $Cl(W)$.
One therefore needs to pick one of the two trivializations of the bulk theory to make sense of the boundary theory as a nontwisted field theory.
This choice is equivalent to an equivalence class of polarizations of $W$, see \cite[Remark 2.2.6]{whatisanellipticobject}.

Such a choice of polarization $J$ defines a Fock module $N_J$ over $Cl(W)$.
Now consider two polarizations $J_1,J_2$, which we imagine to be flattenings of gapped BdG Hamiltonians.
The idea of \cite{camluuk} is to consider the composition between the two induced Fock modules, giving an invertible $(\C,\C)$-bimodule 
\begin{equation}
\label{eq:groundstate}
N_{J_1}^T \otimes_{C l (W)} N_{J_2}.
\end{equation}
Here we made a left $Cl(W)$-module $N$ into a right $Cl(W)$-module $N^T$ using the adjoint operation on the Clifford algebra.
We claim that \eqref{eq:groundstate} defines the one-dimensional state space of the low-energy effective invertible field theory corresponding to the pair of polarizations $J_1$ and $J_2$.
It can be shown that this recovers \ref{chargedgroundstate} in the special case where our Hamiltonian is charge conserving and we take one of the polarizations to be the charge symmetry $iQ$.

The above considerations generalize to the case where there is an internal fermionic group symmetry $G$ to make \eqref{eq:groundstate} into a $G$-representation.
We will present proofs and further details in future work \cite{camluuk}.

\begin{remark}
\label{rem:continuousG}
    If $G$ is not finite, the resulting low energy effective theory will not be topological, also see \cite[Section 5.4]{freedhopkins}.
    For example, take $G = U(1)_Q$ to be the class A symmetry and consider the irreducible real representation $R$ consisting of an oddly charged $Q = n$ particle and its antiparticle of charge $-n$.
In other words we have
\[
R(e^{iaQ}) = e^{ina}
\]
on $V$ and $e^{-ina}$ on $V^*$.
Let us assume the Hamiltonian is zero for simplicity, as the general case splits up as a direct sum into one-dimensional representations.
Then a charge $n$ particle gives a 1d field theory with background $U(1)$ gauge field which assigns to a circle with holonomy $a \in \R/\Z$ the number $e^{ina} \in U(1)$.

Note that this is not a standard Atiyah-Segal style TQFT, but a functorial field theory with genuine geometry as it depends nontrivially on a $U(1)$-connection.
However, it is a smooth topological field theory in the sense of Stolz-Teichner \cite{stolz2011supersymmetric} (also see \cite{gradypavlov}).
When forgetting all the geometric information, our functorial field theory does induce a symmetric monoidal functor of $(\infty, 1)$-categories $\mathbf{Bord}^{U(1)}_{1,0} \to \mathbf{Line}_\C^{cts}$, where the target symmetric monoidal $(\infty,1)$-category corresponds to the topological category of complex lines.
In other words, it is the higher group $B^2 \Z$.
This process of taking the underlying continuous invertible field theory of a smooth invertible field theory corresponds to taking the underlying $\Z$-valued cohomology class of a class in differential cohomology, see \cite{differentialcohomology, davighi2024differential, yamashita2023differential} for discussions of differential cohomology in this context.

It is known that in spacetime dimension $1$, such continuous invertible TQFTs in class A are classified by the group $\Z$.
We can thus conclude that without imposing a unit charge condition, the free-to-interacting map should be given by 
\[
K_0(C^*_f(U(1))) \cong \bigoplus_{\mathbb{N}} \Z \to \Z \quad (n_1, n_2, n_3 \dots) \mapsto n_1 + 2n_2 + 3n_3 + \dots
\]
\end{remark}

\appendix

\section{Real \texorpdfstring{$\Z/2$}{Z2}-graded \texorpdfstring{$C^*$}{C*}-algebras}
\label{sec: C*}

We review the basics on $\Z_2$-graded real $C^*$-algebras, see \cite[Section 14]{blackadar} for a textbook account.

An algebra $A$ over $\R$ is called a $*$-algebra if it comes equipped with an anti-involution $a \mapsto a^*$, i.e. satisfying $a^{**} = a$ and $(ab)^* = b^* a^*$.
The main example is the algebra of continuous linear endomorphisms $B(\mathcal{H})$ of a real Hilbert space $\mathcal{H}$ in which $*$ is defined to be the adjoint.
If $A,B$ are $*$-algebras a $*$-homomorphism from $A$ to $B$ is an algebra homomorphism preserving $*$.
Real $C^*$-algebras $A$ are equivalent to complex $C^*$-algebras $B$ equipped with a \emph{real structure} (a complex-anti-linear involutive $*$-homomorphism $r \colon B \to B$).
The correspondence is implemented using $B := A \otimes_\R \C$ and taking the fixed points of $r$ in the other direction.

In case $A = A_{ev} \oplus A_{odd}$ is additionally a $\Z_2$-graded algebra, we call $A$ a $\Z_2$-graded $*$-algebra if the involution respects the grading.
We call $A$ a $\Z_2$-graded $C^*$-algebra if it admits an injective $*$-homomorphism into $B(\mathcal{H})$ for some real Hilbert space $\mathcal{H}$.
In that case $A$ admits a canonical norm for which it becomes a real Banach algebra.

If $A$ is a $\Z_2$-graded real $C^*$-algebra, then its complexification $A \otimes_\R \C$ is a $\Z_2$-graded complex $C^*$-algebra under $(a \otimes z)^* = a^* \otimes \ol{z}$.
It comes equipped with a canonical real structure $\ol{(.)}\colon A \otimes_\R \C \to A \otimes_\R \C$, i.e. a complex anti-linear $*$-homomorphism which is an involution.
Conversely, given a $\Z_2$-graded complex $C^*$-algebra $B$ with real structure, its real elements (those satisfying $\ol{b} = b$) form a real $\Z_2$-graded $C^*$-algebra.

Given two $\Z_2$-graded algebras $A_1,A_2$, the graded tensor product $A_1 \otimes A_2$ becomes a $\Z_2$-graded algebra with the multiplication given by the Koszul sign rule
\[
(a_1 \otimes a_2) (b_1 \otimes b_2) = (-1)^{|a_2||b_1|}  a_1 b_1 \otimes a_2 b_2.
\]
If $A_1$ and $A_2$ are $*$-algebras, then $A_1 \otimes A_2$ becomes a $*$-algebra by
\begin{equation}
\label{eq:starontensor}
(a_1 \otimes a_2)^* := (-1)^{|a_1||a_2|} a_1^* \otimes a_2^*.
\end{equation}
If $A_1$ and $A_2$ are $C^*$-algebras, we take $A_1 \otimes A_2$ to mean the minimal tensor product of $A_1$ and $A_2$, see \cite[Section 14.4]{blackadar} for details.
If one of $A_1$ and $A_2$ is nuclear, then all norm completions on the tensor product agree, see \cite{bruckler1999tensor} for a review. Since group $C^*$-algebras over locally compact amenable groups are nuclear, this is a situation we will always be in.

To do analysis and topology on modules, we need more structure on modules over $C^*$-algebras:  

\begin{definition}
    A \emph{Hilbert $A$-module} is a right $A$-module $M$ together with an $A$-valued pairing
    \[
    \langle .,. \rangle\colon M \times M \to A
    \]
    such that
    \begin{align*}
    \langle a m_1, m_2  \rangle &=  a \langle m_1, m_2 \rangle 
    \\
    \langle m_1, m_2 \rangle &=  \langle m_2, m_1 \rangle^*
    \\
    \langle m, m \rangle& = 0 \implies m = 0
    \\
    \langle m, m \rangle& \geq 0
    \end{align*}
    where in the last row we mean that $\langle m, m \rangle$ is a positive element of the $C^*$-algebra $A$.
    It can be shown that under the above conditions
    \[
    \| x \| := \| \langle x, x \rangle \|^{\frac{1}{2}}
    \]
    defines a norm and we require $M$ to be complete in this norm.
    If $A$ is a $\Z_2$-graded $C^*$-algebra and $M$ a $\Z_2$-graded module, we say the Hilbert pairing is \emph{graded} if
    \[
    |\langle m_1, m_2 \rangle| = |m_1| + |m_2|.
    \]
\end{definition}

If $A$ is not purely even, the odd and even parts of $M$ are in general not orthogonal, i.e. there may exist $m_1, m_2 \in M$ such that $|m_1| \neq |m_2|$ but $\langle m_1, m_2 \rangle \neq 0$.

\begin{remark}
It is more common to work with right $A$-modules in the $C^*$-algebra literature and inner products are assumed to be $A$-linear in the right entry.
We choose to work with left modules because humans think of operators as acting on the left.
\end{remark}

\begin{remark}
\label{rem:repsvsmods}
A representation of a $C^*$-algebra is defined to be a $*$-homomorphism $A \to B(\mathcal{H})$ for some Hilbert space $\mathcal{H}$.
The difference between Hilbert modules and representations is therefore only that representations have their inner product valued in the ground field $\mathbb{K} = \R,\C$.

Let $\tau\colon A \to \mathbb{K}$ be a faithful tracial state, i.e. a linear map such that $\tau(ab) = \tau(ba), \tau(a^* a) \geq 0$ and $\tau(a^* a) =0$ implies $a = 0$.
Then we can build $\mathbb{K}$-valued inner products from $A$-valued inner products.
In particular, for a Hilbert $A$-module $M$, we can define a corresponding representation of $A$ as the completion of the corresponding positive definite inner product.
    
    For example, a locally compact amenable group always admits a trace~\cite{ng2015strictly}.
    As another example, let $N$ be a compact oriented manifold and let $A = C(N; \mathbb{K})$ be continuous functions on $N$.
    A choice of volume form defines a faithful tracial state $\int_N\colon A \to \mathbb{K}$.
    By the Serre-Swan theorem, a finitely generated projective module over $A$ corresponds to a vector bundle $E \to M$ over $\mathbb{K}$.
    The structure of a Hilbert module on the $A$-module $\Gamma(E)$ of continuous sections corresponds to a Hermitian metric on $E$.
    The corresponding representation is the Hilbert space of $L^2$-functions on $E$.
\end{remark}

\begin{remark}
\label{rem:gradedvsClvalued}
    The author is not aware how the correspondence between $\Z_2$-graded $A$-modules and ungraded $Cl_{+1} \otimes A$-modules extends to Hilbert module structures. 
    Namely, if $(M, \langle .,. \rangle)$ is an ungraded Hilbert $A \otimes Cl_{+1}$-module and $e \in Cl_{+1}$ the preferred generator, then its eigenspace grading $M = M_0 \oplus M_1$ does not obviously make $M$ into a $\Z_2$-graded Hilbert $A$-module.
    
    Taking ideas from Remark \ref{rem:repsvsmods}, we take a linear map $\tau\colon A \otimes Cl_{+1} \to A$ and define $(m_1, m_2) := \tau(\langle m_1, m_2 \rangle)$.
    If $M$ is not purely even, this will not define a $\Z_2$-graded $A$-valued inner product.
    Indeed, we will never be able to achieve the desired $(em_1,em_2) = (-1)^{(m_1,m_2)} (m_1,m_2)$ since we have that
    \[
    \langle e m_1, e m_2 \rangle = \langle m_1, e^* e m_2 \rangle = \langle m_1, m_2 \rangle.
    \]
\end{remark}

\begin{remark}
    The equation $(ab)^* = b^* a^*$ violates the Koszul sign rule for $a,b$ odd.
    The fact that this sign is `wrong' from a categorical perspective follows from the assumption that we need to give the monoidal category of $\Z_2$-graded vector spaces an interesting symmetric braiding.
    However, defining $\Z_2$-graded real $C^*$-algebras to satisfy 
    \begin{equation}
    \label{eq:correctsign}
        (ab)^\dagger =(-1)^{|a||b|} b^\dagger a^\dagger 
    \end{equation} leads to all sorts of complications. For example, Clifford algebras will not admit any $*$-algebra structure.

    On the other hand, for complex $\Z_2$-graded algebras $A$, there is a one-to-one correspondence between $*$-algebras and involutions $\dagger$ satisfying \eqref{eq:correctsign}.
    Explicitly, one choice of correspondence is given by 
    \[
    a^\dagger :=
    \begin{cases}
        a^* & a \text{ even,}
        \\
        i a^* & a \text{ odd.}
    \end{cases}
    \]
    However, the positivity conditions on a graded complex $C^*$-algebra get convoluted. 
    For example, $a^\dagger a$ has its spectrum lying in $i\R_{\geq 0}$ if $a$ is odd.
    Working with $\dagger$ can be helpful to figure out the signs in equations such as \eqref{eq:starontensor} or the tensor product of $\Z_2$-graded Hilbert $A$-modules~\cite[14.4.4.]{blackadar}.

    If $A$ comes equipped with a real structure $\ol{(.)}\colon A \to A$, it will satisfy the Koszul-sign rule violating equation $\ol{a^\dagger} = (-1)^{|a|} \ol{a}^\dagger$.
    If we would try to define another complex antilinear operation by
    \[
    \tilde{a} :=
    \begin{cases}
        \ol{a} & a \text{ even,}
        \\
        i \ol{a} & a \text{ odd.}
    \end{cases}
    \]
    then we still have $\tilde{\tilde{a}} = a$, but $\widetilde{a^\dagger}$.
    However, we also have $\widetilde{ab} = (-1)^{|a||b|} \tilde{a} \tilde{b}$, so a concession must me made somehow.
\end{remark}

The multiplier algebra of $A$ is the largest unital $C^*$-algebra which has $A$ as its essential ideal~\cite[Proposition 2.2.14]{weggeolsen}. 
A useful equivalent definition of $M(A)$ uses a suitable embedding as bounded operators on a Hilbert space \cite[Definition 2.2.2]{weggeolsen}:

\begin{definition}
\label{def:multiplieralg}
    A $*$-isometric embedding of $A$ as bounded operators in a Hilbert space $\mathcal{H}$ is called \emph{nondegenerate} if for every $v \in \mathcal{H}$ there exists an $a \in A$ such that $av \neq 0$.
    In that case the \emph{multiplier algebra} $M(A)$ of $A$ is the collection of operators $T$ in $B(\mathcal{H})$ such that $TA \subseteq A$ and $A T \subseteq A$. 
\end{definition}

Note that if $A \subseteq B(\mathcal{H})$ we can always ensure it is nondegenerate by throwing out the subspace of $v \in \mathcal{H}$ such that $av = 0$ for all $a \in A$, so the multiplier algebra makes sense by the Gelfand-Naimark theorem.
The multiplier algebra is a unital $C^*$-algebra $M(A)$ together with a canonical subalgebra $A \subseteq M(A)$.
 This $C^*$-algebra is independent of the choice of $\mathcal{H}$ up to $*$-isomorphism.
Let $UM(A)$ denote the unitary elements in the multiplier algebra.

Note that taking the multiplier algebra preserves finite direct sums.
Also note that if $A$ is $\Z_2$-graded, $M(A)$ inherits a $\Z_2$-grading.

Sometimes it is necessary to put more conditions on  morphisms between non-unital $C^*$-algebras.
The following definition generalizes the notion of a nondegenerate representation.

\begin{definition}
A $*$-homomorphism is called \emph{nondegenerate} if ideal generated by the image is dense in the codomain. 
\end{definition}

Even though the multiplier algebra construction is not functorial for general $*$-homomorphisms, there is a functor from $C^*$-algebras with nondegenerate morphisms to $C^*$-algebras with unital morphisms.
Similarly, a nondegenerate homomorphism $A \to B$ induces a unital homomorphism $A_+ \to B_+$.

    \sloppy Gelfand duality gives another motivation for considering nondegenerate $*$-homomorphisms\footnote{There are also formulations of non-unital Gelfand duality which don't assume nondegeneracy~\cite{henry2014constructive}.}: 
    the category of commutative $C^*$-algebras and non-degenerate $*$-homomorphisms is equivalent to the category of locally compact Hausdorff spaces and proper continuous maps.
Separately, the category of commutative $C^*$-algebras $A,B$ with non-degenerate $*$-homomorphisms $A \to M(B)$ is equivalent to the category of locally compact Hausdorff spaces and all continuous maps~\cite[Section 1]{MR2681756}.
Since $M(C^0(X)) \cong C_b(X) \cong C(\beta X)$ is the $C^*$-algebra of continuous functions on the Stone-\v Cech compactification of the locally compact Hausdorff space $X$~\cite[Proposition 2.55]{raeburn1998morita}, and $\beta$ is left adjoint to the inclusion of compact Hausdorff spaces into locally compact Hausdorff space, this begs the question of whether $M$ is right adjoint to the functor from unital $C^*$-algebras to $C^*$-algebras with nondegenerate homomorphisms.
The reader is referred to \cite{MR4566000} for an excellent exposition of various Gelfand dualities.

We will often need the notion of nondegenerate morphism $A \to M(B)$ for twisted group $C^*$-algebras, see Appendix \ref{sec:twistedsemidirect}.
A homomorphism $A \to B$ is nondegenerate if and only if for every approximate unit of $A$, its image is an approximate unit of $B$. 

\section{Twisted semidirect products}
\label{sec:twistedsemidirect}

For defining twisted semidirect products, we first need to define projective (or twisted) actions of groups on $C^*$-algebras.
We provide a definition reminiscent of a (unital) nonabelian $2$-cocycle on a group $G$ valued in a $C^*$-algebra $A$:

\begin{definition}
    Let $G$ be a locally compact group and $A$ be a $C^*$-algebra.
    A \emph{twisted $C^*$-dynamical system} of $G$ on $A$ is a pair $(u,\alpha)$ of a twist $u$ and a $u$-twisted action of $G$ on $A$ in the following sense.
    Let $\alpha\colon G \to \Aut A$ be a map into the $*$-automorphisms of $A$ such that the map $G \to A$ given by $\alpha_\bullet (a)$ is Borel-measurable for all $a \in A$.
    Let $u\colon G \times G \to UM(A)$ be a Borel-measurable map into the unitary elements of the multiplier algebra such that
    \begin{align*}
    \alpha_{g_1} \alpha_{g_2} &= u(g_1,g_2) \alpha_{g_1 g_2} u(g_1, g_2)^{-1} 
    \\
    \alpha_{g_1} (u(g_2,g_3)) u(g_1, g_2 g_3) &= u(g_1,g_2) u(g_1 g_2, g_3)
    \\
    \alpha_1 = \id_A \quad u(g,1) &= u(1,g) = 1
    \end{align*}
    for all $g_1, g_2, g_3 \in G$.
    If $C$ is another $C^*$-algebra, a \emph{covariant homomorphism} of $(G,A,\alpha,u)$ into $M(C)$ consists of a nondegenerate homomorphism $\pi\colon A \to M(C)$ and a Borel-measurable map $v\colon G \to UM(A)$ such that
    \begin{align*}
        v_{g_1} v_{g_2} = \pi u (g_1, g_2) v_{g_1 g_2} \quad \pi \alpha_g(a) = v_g \pi(a) v_g^{-1}
    \end{align*}
\end{definition}

\begin{remark}
    In the above we chose to work with Borel-measurable instead of continuous cocycles on $G$.
    This is for technical reasons related to the fact that group cohomology of topological groups is poorly behaved.
    There are possibly better behaved alternatives to this choice using an appropriate abstract notion of cohomology, see Remark \ref{rem:segalcoh} for further discussion.
    However, we prefer to work with the above definition in order to use results of \cite{packerraeburnI} and \cite{packerraeburnII}.
\end{remark}

The \emph{twisted semidirect $C^*$-algebra} $A \rtimes_{\alpha,u} G$ is a $C^*$-algebra satisfying a universal property that maps out of $M(A \rtimes_{\alpha,u} G)$ are exactly covariant homomorphisms.
More precisely, there is a canonical covariant homomorphism of $(G,A,\alpha,u)$ into $M(A \rtimes_{\alpha,u} G)$ written $i_G\colon G \to UM(A \rtimes_{\alpha,u} G)$ and $i_A\colon A \to M(A \rtimes_{\alpha,u} G)$ such that the following theorem holds.
We will not go into the explicit construction of the twisted semidirect product, see \cite[Chapter II]{renault2006groupoid} for a concrete textbook account.
Instead we prefer to only use the universal property with the idea that our theorems generalize to setups of a similar categorical flavour.

\begin{theorem} \cite[Theorem 1.2(a)]{packerraeburnII}
\label{th:univprop}
    Let $(\pi,v)$ be a covariant homomorphism of $(G,A,\alpha,u)$ into $M(C)$.
    \begin{enumerate}
        \item There exists a unique nondegenerate homomorphism $A \rtimes_{\alpha,u} G \to M(C)$ with the property that it induces $\pi$ when restricting to $A$ under $i_A$ and induces $v$ when restricting to $G$ under $i_G$.
        \item If $D$ is a $C^*$-algebra equipped with a covariant homomorphism $(j_G,j_A)$ with the property that for every covariant homomorphism $(\pi,u)$ into another $C^*$-algebra $M(C)$ there exists a unique nondegenerate homomorphism $D \to M(C)$ which restricts to $\pi$ under $j_A$ and to $v$ under $j_G$, then there is an isomorphism $D \cong A \rtimes_{\alpha,u} G $ intertwining $i_G$ with $j_G$ and $i_A$ with $j_A$.
    \end{enumerate}
\end{theorem}

\[
\begin{tikzcd}[column sep = 34]
    A \rtimes G \ar[r,"{\forall (\pi,u)}"] \ar[d,"{(j_G,j_A)}"] & M(C) \\
    M(D) & D \ar[l, hookrightarrow] \ar[u,"\exists !", dashed]
\end{tikzcd} 
\implies (j_G,j_A) \text{ induces } A \rtimes G \cong D
\]

\begin{example}
Let $G$ be a locally compact group, take $A = \R$, $\alpha$ to be identically equal to $\id_A$ and $u$ to be identically equal to $1 \in UM(A)$. 
    The \emph{group $C^*$-algebra} is defined to be $C^*(G) := \R \rtimes_{\alpha, u} G$.
\end{example}

We need a slight generalization of the above to the $\Z_2$-graded setting.
Let $\theta\colon G \to \Z_2$ be a group homomorphism.
We define a $\Z_2$-grading on $A \rtimes_{\alpha,u} G$ by giving the involution which is $-1$ on the odd part.
Applying the universal property of the semidirect product, the grading automorphism of $A \rtimes_{\alpha,u} G$ is defined uniquely by the map $G \to UM(A \rtimes_{\alpha,u} G)$ given by $[g] \mapsto (-1)^{\theta(g)} j_G(g)$ where $j_G\colon G \to UM(A \rtimes_{\alpha,u} G)$ is the canonical inclusion inducing the identity on $A \rtimes_{\alpha,u} G$.
This is well-defined because $\theta((-1)^F) = 1$ and the automorphism squares to one by the universal property of the identity on $\R \rtimes_\nu G_b$.
Since we want $A$ to be even, the map $A \to UM(A \rtimes_{\alpha,u} G)$ is just given by $j_A$.
Because $\theta(g_1 g_2) = \theta(g_1) + \theta(g_2)$ this is a covariant homomorphism of $(G,A,\alpha,u)$ into $A \rtimes_{\alpha,u} G$.
The square of this operator is the identity by the universal property of the twisted semidirect $C^*$-algebra.

\begin{definition}
     Let $\theta\colon G \to \Z_2$ be a Borel group homomorphism, $(G,A,\alpha,u)$ a $C^*$-dynamical system, $C$ a $\Z_2$-graded $C^*$-algebra and $(\pi,v)$ a covariant homomorphism of $(G,A,\alpha,u)$ into $M(C)$.
     Then $(\pi,v)$ is called \emph{graded} if $\pi$ lands in the even part of $C$ and $v(g)$ is even or odd depending on $\theta(g)$.
\end{definition}

The analogous universal property in the $\Z_2$-graded case if we assume covariant homomorphisms are graded:

\begin{corollary}
    Let $\theta\colon G \to \Z_2$ be a continuous group homomorphism, $(G,A,\alpha,u)$ a $C^*$-dynamical system, $C$ a $\Z_2$-graded $C^*$-algebra and $(\pi,v)$ a graded covariant homomorphism of $(G,A,\alpha,u)$ into $M(C)$.
    Then there exists a unique even nondegenerate homomorphism $A \rtimes_{\alpha,u} G \to M(C)$ with the property that it induces $\pi$ when restricting to $A$ under $i_A$ and induces $v$ when restricting to $G$ under $i_G$.
\end{corollary}
\begin{proof}
    Using the universal property of the twisted semidirect product a nondegenerate homomorphism $A \rtimes_{\alpha,u} G \to M(C)$ is equivalent to a covariant homomorphism $(\pi,v)$ of $(G,A,\alpha,u)$ into $M(C)$.
    Note that this homomorphism preserves the grading if and only if the covariant homomorphism is graded.
\end{proof}

\subsection{Twisted Peter--Weyl theorem}
\label{sec:twistedpeterweyl}

To prove the Peter--Weyl style theorem for the fermionic group algebra, we need the following lemma for the special case where $N = \Z_2^F$ and the extension is central.

\begin{lemma}
\label{lemma: bijection}
Let $A$ be a unital $C^*$-algebra, let
\[
1 \to N \to G \xrightarrow{[.]} G/N \to 1
\]
be a group extension of topological groups, $\mu\colon N \to A$ a map such that $\mu(n_1 n_2) = \mu(n_1) \mu(n_2)$ for all $n_1, n_2 \in N$ and $s\colon G/N \to G$ a Borel measurable section such that $s([1]) = 1$ with corresponding measurable two-cocycle $\nu$.
The map $f(u)([h]) := u(s([h]))$ defines a bijection 
\begin{align*}
\{u\colon G \to A: & u(g_1 g_2)= u(g_1) u(g_2), u|_N = \mu\}
\\
&\downarrow f
\\
\{v\colon G/N \to A: & v(h_1 h_2)= \mu \nu(h_1,h_2) v(h_1) v(h_2), 
\\
&\mu\left(s(g) n s(g)^{-1}\right) v(h) = v(h) \mu(n), h_1,h_2,h \in G/N, n \in N\}
\end{align*}
preserving Borel measurability.
\end{lemma}
\begin{proof}
The map is well-defined because firstly
\begin{align*}
    f(u)([h_1 h_2]) &= u(s([h_1 h_2])) = u(\nu([h_1],[h_2]) s([h_1]) s([h_2])) 
    \\
    &= \mu\nu([h_1],[h_2]) u(s([h_1])) u(s([h_2])) 
    \\
    &= \mu\nu([h_1],[h_2]) f(u)([h_1]) f(u)([h_2])
\end{align*}
and secondly 
\begin{align*}
    u(s([h])) \mu(n) &= u(s([h]) n) = u(s([h]) n s([h])^{-1} s([h]) ) 
    \\
    &= \mu(s([h]) n s([h])^{-1}) u(s([h])).
\end{align*}
For injectivity, suppose that $f(u_1)([h]) = f(u_2)([h])$ for all $h \in G$.
Pick $n \in N$ such that $n s([h]) = h$.
Then
\[
u_1(h) = u_1(n s([h])) = u_1(n) u_1(s([h])) = \mu(n) u_2(s([h])) = u_2(h)
\]
and so $f$ is injective.

For surjectivity, let $v\colon G/N \to A$ satisfy the two relevant equations.
Define a map $u\colon G \to A$ as follows.
Given $g \in G$ note that $n:= g s([g])^{-1} \in N$ is the unique element such that $g = n s([g])$.
We set $u(g) := \mu(n) v([g])$.
Note that $f(v) = u$ and $u|_N = \mu$, so we are done when we can show that $u(g_1 g_2) = u(g_1) u(g_2)$.
We get $n_1 = g_1 s([g_1])^{-1}, n_2 = g_2 s([g_2])^{-1} \in N$ and compute
\begin{align*}
    g_1 g_2 &= n_1 s([g_1]) n_2 s([g_2]) 
    \\
    &= n_1 s([g_1]) n_2 s ([g_1])^{-1} s ([g_1]) s ([g_2])
    \\
    &= n_1 s([g_1]) n_2 s ([g_1])^{-1} \nu([g_1],[g_2]) s([g_1 g_2]).
\end{align*}
Since $n_1 s([g_1]) n_2 s ([g_1])^{-1} \nu([g_1],[g_2]) \in N$ this results in 
\begin{align*}
    u(g_1 g_2) &= \mu(n_1 s([g_1]) n_2 s([g_1])^{-1} \nu([g_1],[g_2]) ) v([g_1 g_2])
    \\
    &= \mu(n_1) \mu( s([g_1]) n_2 s([g_1])^{-1}) \mu \nu([g_1],[g_2])  v([g_1 g_2])
    \\
    &=\mu(n_1) \mu( s([g_1]) n_2 s([g_1])^{-1}) v([g_1]) v([g_2]) 
    \\
    &= \mu(n_1)  v([g_1]) \mu(n_2) v([g_2]) 
    \\
    &= u(g_1) u(g_2).
\end{align*}
Because $s$ is Borel measurable, $f$ preserves Borel measurability.
\end{proof}

\begin{corollary}
Let $A$ be a unital $C^*$-algebra.
There is a bijection 
\begin{align*}
\{u\colon G \to A: u(g_1 g_2) &= u(g_1) u(g_2), u((-1)^F) = -1\}
\\
&\downarrow f
\\
\{v\colon G_b \to A: v(h_1 h_2) &= \nu(h_1,h_2) v(h_1) v(h_2)\}
\end{align*}
preserving Borel measurability.
\end{corollary}
\begin{proof}
Apply the lemma.
Take $\mu\colon \Z_2^F \to A$ to be $(-1)^F \mapsto -1$.
Since $\Z_2^F \subseteq G$ is central, the second equation in the codomain of $f$ is redundant.
\end{proof}

The following shows in particular that for a compact fermionic group, the fermionic group algebra is independent of the choice of representative $\nu$.

\begin{proposition}
\label{prop: fermionic group alg}
For a compact fermionic group $(G,(-1)^F,\theta)$ the fermionic group $C^*$-algebra is isomorphic as a real $C^*$-algebra to
\[
\bigoplus_{\substack{(V,\rho) \text{ f.d. irrep of } G \\ \text{s.t. } \rho((-1)^F) = -1}} \End V.
\]
It becomes $\Z_2$-graded by requiring $\rho(g)$ to be even/odd depending on whether $\theta(g) = 0$ or $\theta(g) = 1$.
\end{proposition}
\begin{proof}
We use the universal property \ref{th:univprop} to construct the desired isomorphism.
We therefore first need a covariant homomorphism from the twisted $C^*$-dynamical system $(\R, G_b, \alpha = 1, \nu)$ to the multiplier algebra $M(C^*_f(G))$.
We first construct an explicit model of $M(C^*_f(G))$ as follows.
Let
\[
\mathcal{H} := \bigoplus_{\substack{(V,\rho) \text{ f.d. irrep of } G}} V \qquad \mathcal{H}' := \bigoplus_{\substack{(V,\rho) \text{ f.d. irrep of } G \\ \text{s.t. } \rho((-1)^F) = -1}} V.
\]
be the direct sum of real Hilbert spaces where we have given the irreducible representations compatible orthogonal complex or quaternionic structures depending on the representation type.
It follows by the Peter--Weyl theorem that $C^*(G) \subseteq B(\mathcal{H})$ and $C^*_f(G) \subseteq B(\mathcal{H}')$ nondegenerately.
Therefore the multiplier algebras can be identified with the collection of multipliers in $B(\mathcal{H})$ respectively $B(\mathcal{H}')$ giving
\begin{align*}
M(C^*(G)) = \{T \in B(\mathcal{H}) &\text{ structure-preserving}: T(V) \subseteq V \, \forall \text{ f.d. irreps}\}    
\\
M(C^*_f(G)) = \{T \in B(\mathcal{H}')  &\text{ structure-preserving}: 
\\
&T(V) \subseteq V \, \forall \text{ f.d. irreps with } \rho((-1)^F) = -1\},
\end{align*}
where by structure-preserving we mean that $T|_V$ is required to commute with the complex or quaternionic structures on $V$ if $V$ is of complex respectively quaternionic type. 

We take the covariant homomorphism $(\R, G_b, \alpha = 1, \nu) \to M(C^*_f(G))$ to be the nondegenerate $*$-algebra map $\R \to M(C^*_f(G))$ given by including the unit and the map $v\colon G_b \to M(C^*_f(G))$ is defined as the composition
\[
G_b \overset{s}{\to} G \overset{j}{\to} UM(C^*(G)) \overset{p}{\to} UM(C^*_f(G)).
\]
Here $j(g) := \oplus_{(V,\rho) \text{ f.d. irrep }} \rho(g) \in M(C^*(G))$ is Borel, lands in unitaries because all representations are unitary and $p(T) := T|_{\mathcal{H}'}$ is a $*$-algebra map.
By definition of the twists and the fact that $j$ is multiplicative in $G$, we have 
\[
j(s(gh)) = j(\nu(g,h)) j(s(g)) j(s(h)).
\]
Also $p j((-1)^F) = -1$ and so we are done showing that this defines a covariant homomorphism.
The induced map of $C^*$-algebras respects the $\Z_2$-gradings by construction.

We now show that this covariant homomorphism satisfies the universal property \ref{th:univprop}.
So let $(\pi',v')\colon (\R, G_b, \alpha = 1, \nu) \to C$ be another covariant homomorphism.
Since $\pi'\colon \R \to M(C)$ is a nondegenerate algebra map, it must be the inclusion of $1$.
On the other hand $v'\colon G_b \to M(C)$ is an arbitrary Borel map such that $v'(g) v'(h) = \nu(g,h) v'(gh)$, where we abused notation and denoted the $\{\pm 1\} \subseteq M(C)$-valued cocycle corresponding to the $\Z_2^F \subseteq G$-valued cocycle again by $\nu$.
Let $u\colon G \to M(C)$ be the associated Borel map with $u(g h) = u(g) u(h)$ and $u((-1)^F) = -1$ given by the bijection of Lemma \ref{lemma: bijection}.
Using the universal property of $C^*(G)$, this gives a nondegenerate map $\psi\colon C^*(G) \to M(C)$ which extends to $\psi\colon M(C^*(G)) \to M(C)$ with the property that it agrees with $u$ on $G \subseteq M(C^*(G))$.
In particular for $g = (-1)^F$ we get that the element of $M(C^*(G))$ which is $-1$ on $\mathcal{H}'$ and $1$ on the orthogonal complement is mapped to $-1 \in M(C)$ by $\psi$ (note that because $(-1)^F$ is central and of square one, we either have $\rho((-1)^F) = 1$ or $\rho((-1)^F)$ if $\rho$ is irreducible).
So $\psi$ maps the projection operator $P$ onto $(\mathcal{H}')^\perp$ to zero:
\[
\psi(P) = \frac{\psi((-1)^F) + \psi(1)}{2} =0.
\]
Now if $T \in C^*(G)$ is in the kernel of the projection to $C^*_f(G)$, then $T|_{\mathcal{H}'} = 0$ and so $\psi(T) = \psi(P) \psi(T) = 0$.
We conclude that $\psi$ factors through a map $\hat{\psi}\colon C^*_f(G) \to M(C)$.
We are done once we show that $\hat{\psi}$ agrees with the map $v'\colon G_b \to M(C)$ on the image of the map $v\colon G_b \to M(C^*_f(G))$ constructed in the first part of the proof.
But by construction of the map $f$ of Lemma \ref{lemma: bijection}, $u$ agrees with $v'$ on the image of $s\colon G_b \to G$.
Therefore because $\hat{\psi}$ agrees with $u$ on $G \subseteq M(C^*(G))$, the diagram
\[
\begin{tikzcd}
G_b \arrow[r, "v"] \arrow[dr, "v'"] & M(C^*_f(G)) \arrow[d,"\hat{\psi}"]
\\
& M(C)
\end{tikzcd}
\]
commutes.

We now have to show that $\hat{\psi}\colon UM(C^*_f(G)) \to M(C)$ is the unique map with this property. 
So let $\hat{\psi}'\colon UM(C^*_f(G)) \to M(C)$ be another map which agrees with $v'$ on the image of $v$.
Restricting gives a map $\psi'\colon UM(C^*(G)) \to M(C)$ and a map $u'\colon G \to M(C)$.
Notice that since $\hat{\psi}'$ exists, $\psi'(P) = 0$ and so 
\[
u'((-1)^F) = \psi' j((-1)^F) =  -1.
\]
Since it also satisfies $u'(g_1 g_2) = u'(g_1) u' (g_2)$ and $u'(s(g)) = v'(g)$, we get $u' = u$.
By the universal property of $C^*(G)$, we have that $\psi\colon UM(C^*(G)) \to M(C)$ is the unique extension of $u$ and so $\psi' = \psi$.
By surjectivity of $p$ also $\hat{\psi}' = \hat{\psi}$.

This ends the proof.
\end{proof}

To prove the next theorem, we need the following direct consequence of Lemma \ref{lemma: bijection}.

\begin{lemma}
Let $A$ be a unital real $C^*$-algebra and let $\pi\colon \C \hookrightarrow A$ be a subalgebra isomorphic to the real $C^*$-algebra of complex numbers. 
Let $\mu\colon U(1)_Q \to \C \overset{\pi}{\to} A$ be the obvious composition.
There is a bijection
\begin{align*}
    \{u\colon G^\tau \to A: u(g_1 g_2) &= u(g_1) u(g_2), u|_{U(1)_Q} = \mu\}
    \\
    &\downarrow f
    \\
    \{ v\colon H \to A: v(h_1) v(h_2) &= \mu \tau(h_1,h_2), z^{\phi(h)} v(h)  = v(h) z, h \in H, z \in U(1)\}
\end{align*}
\end{lemma}

\begin{proof}[Proof of Theorem \ref{th:unit charge fd}] 
In similar spirit to Proposition \ref{prop: fermionic group alg}, the strategy is to apply Theorem \ref{th:univprop}.
So we have to construct a covariant homomorphism $(\C, H, \alpha, \tau) \to M(C^*_{uc}(G^\tau))$.
First note that if we set
\[
\mathcal{H} := \bigoplus_{(V,\rho) \text{ f.d. irrep of } G^\tau} V
\quad \mathcal{H}' := \bigoplus_{\substack{(V,\rho) \text{ f.d. irrep of } G^\tau \\ \text{of unit charge}}} V
\]
to be the orthogonal direct sum of real Hilbert spaces, then 
\begin{align*}
M(C^*(G)) = \{T \in B(\mathcal{H}) &\text{ structure-preserving}: T(V) \subseteq V \, \forall \text{ f.d. irreps}\}    
\\
M(C^*_{uc}(G)) = \{T \in B(\mathcal{H}')  &\text{ structure-preserving}: 
\\
&T(V) \subseteq V \, \forall \text{ f.d. irreps with unit charge}\}.
\end{align*}
Now let $\pi\colon \C \to M(C^*_{uc}(G^\tau))$ be the map of real $C^*$-algebras uniquely defined by $\pi(1) = 1$ and
\[
\pi(i) = \bigoplus_{\substack{\text{unit charge} \\ \text{f.d. irrep } \rho}} \rho(iQ)
\]
Because unit charge representations have odd charge, $\pi(i)^2 = -1$.
It is a $*$-homomorphism because all $\rho(iQ)$ are unitary.
The homomorphism is clearly nondegenerate.

Define a map $v\colon H \to M(C^*_{uc}(G^\tau))$ by the composition
\[
H \overset{s}{\to} G^\tau \overset{j}{\to} M(C^*(G^\tau)) \overset{\sigma}{\to} M(C^*_{uc}(G^\tau))
\]
with $j(g) = \oplus_{(V,\rho) \text{ irrep of }G^\tau} \rho(g)$ and $\sigma$ the projection map.
We have to show that $v(h_1) v(h_2) = \pi(\tau(h_1,h_2)) v(h_1 h_2)$.
Note that
\begin{align*}
    j(s(h_1) j(s(h_2)) = j(s(h_1) s(h_2)) = j(\tau(h_1,h_2) s(h_1 h_2) = j(\tau(h_1,h_2)) j( s(h_1 h_2))
\end{align*}
so it suffices to show that $\sigma j(e^{iaQ}) = \pi(e^{ia})$ for all $e^{iaQ} \in U(1)_Q$.
We compute
\begin{align*}
    \sigma j(e^{iaQ}) &= \sigma\left( \bigoplus_{(V,\rho) \text{ f.d. irrep of }G^\tau} \rho(iaQ) \right) 
    \\
    &= \sigma\left( \bigoplus_{n \geq 0} \bigoplus_{\substack{(V,\rho) \text{ f.d. irrep of }G^\tau \\ \text{ of charge } n}} \cos (na) + \sin (na) \rho(iQ) \right) 
    \\
    &= \bigoplus_{\substack{(V,\rho) \text{ f.d. irrep of }G^\tau \\ \text{ of charge } 1}} \cos (a) + \sin (a) \rho(iQ)
    \\
    &= \pi(\cos a + i \sin a)
\end{align*}
as desired.
The final part in showing that $(\pi,v)$ defines a covariant homomorphism is to compute the commutation relation
\begin{align*}
    v(h) \pi(e^{iaQ}) &= \sigma(j(s(h)) \cdot \bigoplus_{(V,\rho) \text{ f.d. irrep of }G^\tau} \rho(e^{iaQ})
    \\
    &= \sigma j(s(h) e^{iaQ}) = \sigma j(e^{\phi(h) iaQ} s(h)) = \pi(e^{\phi(h) ia}) v(h).
\end{align*}
The resulting $*$-homomorphism respects the $\Z_2$-grading.

To finish the proof of the theorem, we have to show the universality condition.
So let $(\pi',v')\colon (\C,H,\alpha,\tau) \to M(C)$ be any covariant homomorphism.
We have to show that there exists a unique $\hat{\psi}\colon C^*_{uc}(G^\tau) \to M(C)$ such that $\hat{\psi} \pi = \pi'$ and $\hat{\psi} v = v'$.

We first show $\hat{\psi}$ is unique.
Define $u$ and $\psi$ by the diagram
\[
\begin{tikzcd}
H \arrow[r,"s"] \arrow[dr, "v'"] & G^\tau \arrow[r,"j"] \arrow[d, "u"] & UM(C^*(G^\tau)) \arrow[r, "\sigma"] \arrow[dl, "\psi"] & UM(C^*_{uc}(G^\tau)) \arrow[dll, bend left, "\hat{\psi}"]
\\
& M(C) & &
\end{tikzcd}
\]
We show that $u|_{U(1)_Q} = \pi$.
Once this is done, $u$ is uniquely determined by applying Lemma \ref{lemma: bijection} in the case where $N = U(1), \mu = \pi$.
Then $\psi$ is uniquely determined by the universal property of $C^*(G^\tau)$ and hence $\hat{\psi}$ is uniquely determined by surjectivity of $\sigma$.
Therefore we compute 
\begin{align*}
    u(e^{iaQ}) &= \psi\left( \bigoplus_{\rho \text{ f.d. irrep}} \rho(e^{aiQ}) \right)
    = \hat{\psi}\left( \bigoplus_{\substack{\rho \text{ f.d. irrep} \\ \text{ of unit charge}}} \rho(e^{aiQ}) \right)
    \\
    &= \hat{\psi}\left( \bigoplus_{\substack{\rho \text{ f.d. irrep} \\ \text{ of unit charge}}} \cos a + \rho(iQ) \sin a \right)
    \\
    &= \cos a + \sin a \hat{\psi} \pi(i)
    \\
    &= \cos a + \sin a \pi'(i)
    \\
    &= \pi'(e^{ia}).
\end{align*}
We conclude $\hat{\psi}$ is unique.

By Lemma \ref{lemma: bijection} we know $v$ exists and by the universal property of $C^*(G^\tau)$ we know $\psi$ exists.
So to show existence of $\hat{\psi}$ we only have to show that $\psi$ factors through the projection $\sigma$, i.e. that it vanishes on irreducible representations that are not of unit charge.

Let $p$ be a prime. 
Let $P_j \in M(C^*(G))$ for $j \in \Z_p$ denote the projections onto
\[
\bigoplus_{\substack{(\rho,V) \text{ f.d. irrep} \\ \text{ of charge } j \mod p}} V.
\]
It suffices to show that $\psi(P_j) = 0$ if $j \neq 0$. Indeed, let $(V,\rho)$ be an irreducible representation of charge $n$ and $a \in \End V \cap M(C^*(G))$. 
For a prime $p$ not dividing $n$, we have $P_n a = a$ and $n \neq 0 \mod p$ so that $\psi(a) = \psi(P_n) \psi(a) = 0$.

To show that $\psi(P_j) = 0$ if $j \neq 0$, note that for $a = 2 \pi i j / p$
\begin{align*}
    \pi'(e^{iaQ}) &=
    \cos\left( 2 \pi i\frac{j}{p}\right) + \pi'(i) \sin\left( 2 \pi i \frac{j}{p}\right) 
    \\
    &=  
    \psi\left( \bigoplus_{n = 1}^p \bigoplus_{\substack{(\rho,V) \text{ f.d. irrep} \\ \text{ of charge } n \pmod p}} \cos\left( 2 \pi i\frac{ j n}{ p}\right) + \rho(iQ) \sin\left(2 \pi i\frac{ jn}{ p}\right) \right)
    \\
    &= \sum_{n=1}^p \left(\cos\left( 2 \pi i\frac{ j n}{ p}\right) + \pi'(i) \sin\left(2 \pi i\frac{ jn}{ p}\right) \right) \psi(P_n)
\end{align*}
By nondegeneracy $\psi$ is unital and so we can write $1 = \psi(P_1) + \dots + \psi(P_p)$.
We arrive at a system of $p$ linear equations of the form $Ax = 0$ in the variables $x_j := \psi(P_j)$ given by
\[
A_{jn} = 1 - e^{2 \pi i \frac{jn}{p}}.
\]
Since we are not interested in $x_1$ we get rid of the identically zero columns and rows $j = 1$ and $n = 1$ to arrive at a square matrix of size $(p-1)$ with coefficients in $\C \subseteq M(C)$.
This matrix is invertible over $\C$; if $x_2, \dots, x_p \in \C$ is a solution of $Ax = 0$, then the polynomial
\[
\sum_{n=2}^{p}(y^n-1) x_n
\]
in $y$ is of degree at most $p-1$ and has at least $p$ distinct roots.
Hence $x_2 = \dots = x_p = 0$ and we conclude that $\psi(P_j) = 0$ for all $j \in \Z_p \setminus 0$ as desired.
\end{proof}

\bibliography{biblio.bib}{}
\bibliographystyle{plain}

\end{document}